%% file: main.tex
\documentclass[%
 reprint,
superscriptaddress,
amsmath,amssymb,
aps,
floatfix,
]{revtex4-2}

\usepackage[english]{babel}

\usepackage{amsmath,amssymb,amsthm}
\usepackage{tikz}

\usepackage{pgfplots}
\pgfplotsset{width=10cm,compat=1.9}

\usepackage{physics}
\usepackage{graphicx}
\usepackage{tikz}
\usetikzlibrary{calc}
\usepackage[colorlinks=true, allcolors=blue]{hyperref}

\definecolor{blau}{RGB}{50,50,150}

\newtheorem{theorem}{Theorem}
\newtheorem{lemma}{Lemma}
\newtheorem{definition}{Definition}

\DeclareMathOperator{\poly}{poly}
\DeclareMathOperator{\ad}{ad}

\DeclareMathOperator*{\esssup}{ess\,sup}

\begin{document}
\preprint{APS/123-QED}

\title{Polynomial equivalence of the global transverse-field Ising model and the gate model of quantum computation}

\author{Matthias Werner}
\email{matthias.werner@qilimanjaro.tech}
\affiliation{Qilimanjaro Quantum Tech., Carrer de Veneçuela, 74, Sant Martí, 08019 Barcelona, Spain}
\affiliation{Departament de Física Quàntica i Astrofísica (FQA), Universitat de Barcelona (UB), Carrer de Martí i Franqués, 1, 08028 Barcelona, Spain}
\affiliation{Institut de Ci\`encies del Cosmos, Universitat de Barcelona, ICCUB, Carrer de Martí i Franqu\`es, 1, 08028 Barcelona, Spain}

\thispagestyle{empty}
\pagestyle{empty}

\begin{abstract}
    \input{00_abstract}
\end{abstract}

\maketitle

\input{01_introduction}
\input{02_overview}
\input{03_related_works}
\input{04_methods}
\input{05_numerical_validation}
\input{06_conclusion}
\input{07_acknowledgment}

\bibliography{references}

\onecolumngrid
\appendix
\input{08_appendix_a}
\input{09_appendix_b}
\input{10_appendix_c}

\end{document}

%% file: 00_abstract.tex
The transverse-field Ising model has attracted a lot of attention in recent years, especially in the quantum simulation and quantum computation literature. This interest is driven by many platforms for analog quantum computation, which implement the transverse-field Ising model for solving optimization problems, such as quantum annealing. However, it has remained an open question whether the Ising model with a global transverse field is equivalent to the gate model of quantum computation. Here we answer this question affirmatively for the case of a non-monotonic time-dependent transverse field. Building on a recent result by Cesa and Pichler on global control of Rydberg atoms, we provide a construction that allows simulating arbitrary quantum circuits using the Ising model with global transverse field with polynomial overhead in time, qubit number, and energy scale. Although the polynomial overheads we establish here are large relative to what is feasible on real-world quantum hardware, our result motivates the development of more sophisticated methods for simulating quantum circuits using the Ising model with a global transverse field. Additionally, under the assumption that quantum computing is strictly more powerful than classical computing, our result serves as a no-go theorem for efficient classical simulation of the transverse-field Ising model with a time-dependent global transverse field. Therefore, our finding is relevant for multiple communities, from analog quantum simulation and quantum optimization on various platforms to complexity and control theory.

%% file: 01_introduction.tex
\section{Introduction}
In recent years, quantum computers have reached unprecedented system sizes and coherence times, allowing for the first experiments where the simulability of the results using classical hardware becomes highly questionable, if not impossible~\cite{King_2025}. In parallel with the gate model of quantum computation, analog approaches to quantum computing, such as quantum annealing and analog quantum simulation, have established themselves as alternative computing paradigms dedicated to optimization problems~\cite{Albash_2018} and quantum simulation of quantum many-body physics~\cite{Daley_2023, King_2025}, chemistry~\cite{Arg_ello_Luengo_2019} and other fields.\\
A relevant question for these novel models of computation is the issue of universality, or more precisely polynomial equivalence to the gate model of quantum computation. Recently, it has been an active research question to establish universality results for quantum optimization and simulation paradigms~\cite{Lloyd_2018, Morales_2020, Zheng_2024, Hu_2026}, which are based on the transverse-field Ising model (TFIM). \\
In this work, we present a method to simulate arbitrary quantum circuits using the TFIM with a global transverse field. We use a precisely timed non-monotonic global transverse field to realize effective gates on groups of qubits simultaneously. The qubits are strongly coupled, which allows us to move logical quantum information around the device and process it. The approach is based on a method proposed by Cesa and Pichler~\cite{Cesa_2023}, who showed how to perform universal quantum computation using global control of two distinct species of atoms only. Here, we will use fully global time-dependent control and static local control. Combined with the non-monotonic global transverse field, this makes our method in principle applicable to quantum annealing hardware with the reverse-annealing feature, as well as other platforms that implement a time-dependent TFIM. Furthermore, we provide an error estimate and show that any quantum circuit can be simulated with polynomial overhead in resources, thus establishing the polynomial equivalence between the TFIM with a global time-dependent transverse field and the gate model of quantum computation.

We require only a few mild assumptions on the time-dependence of the global transverse field, making our result applicable for many devices. This proves that various quantum simulation platforms based on the TFIM are implementations of a universal model of quantum computation. Our finding also has implications for the question of efficient classical simulability of this model. Assuming that quantum computers can efficiently solve tasks that are exponentially expensive for classical computers, i.e. if $\text{BQP} \neq \text{BPP}$, our result establishes that no classical algorithm can efficiently simulate the dynamics of the TFIM with global transverse field, providing a no-go theorem for classical simulation.

This work is structured as follows. First, we state the problem and the main result in Section~\ref{sec:ProblemAndMainResult}. Then, we provide a discussion of related work and situate our finding in the context of the field in Section~\ref{sec:RelatedWork}. In Section~\ref{sec:Methods} we outline the proof of the main theorem, which we numerically validate in Section~\ref{sec:NumericalValidation}. The technical details of the proof are given in the appendix. Ultimately, we summarize our conclusions in Section~\ref{sec:Conclusion}.

%% file: 02_overview.tex
\section{Problem definition and main result}
\label{sec:ProblemAndMainResult}
We first introduce the global transverse-field Ising model (TFIM). We will denote the Pauli operators acting on qubit $n$ as $X_n$, $Y_n$ and $Z_n$. The TFIM consists of a diagonal Hamiltonian
\begin{equation}
    H_Z = \sum_{m} h_m Z_m + \sum_{(m,n) \in E(\mathcal{G})} J_{mn} Z_m Z_n
    \label{eq:DefHz}
\end{equation}
with local longitudinal fields $h_m$ and coupling constants $J_{mn}$. The qubits are arranged on a graph $\mathcal{G}$, where the edges $(m,n) \in E(\mathcal{G})$ correspond to $ZZ$-interactions, while the nodes correspond to qubits. The quantum fluctuations are introduced via homogeneous transverse fields on each qubit
\begin{equation}
    H_X = \sum_{m} X_m \ .
    \label{eq:DefHx}
\end{equation}
The full Hamiltonian of the global TFIM is time-dependent, where the time-dependence is introduced via $H_X$ and the function $\Gamma_t: [0, T] \rightarrow \mathbb{R}$, i.e.
\begin{equation}
    H_t = \Gamma_t H_X + H_Z \ .
    \label{eq:DefFullHamiltonian}
\end{equation}
The control field $\Gamma_t$ is also called the schedule. For many devices, one has to assume $\Gamma_t \geq 0$. While this is not generally required for all platforms, we will work within this limitation to maintain the broad applicability of our claim. With this we define:
\begin{definition}
    (Global TFIM as computational model) A system of $n_q$ qubits initialized in $\ket{0}$ evolving in time under the time-dependent Hamiltonian $H_t$ from Eq.~\eqref{eq:DefFullHamiltonian}, where $h_m, J_{mn}, |\Gamma_t| = \poly(n_q)$ for all $m$, $n$ and $0 \leq t \leq T$. After the time evolution, the qubits are measured in the computational basis.
\end{definition}
Note that the global TFIM is commonly used in the quantum annealing literature~\cite{Albash_2018, Callison_2021}. For this particular application, one starts with $\Gamma_0 \gg 1$ and decreases over time $T$ to $\Gamma_T = 0$ at the end of the anneal. In the conventional mode of quantum annealing, also called forward annealing, $\Gamma_t$ decreases monotonically, i.e. one anneals the system directly from dominating fluctuations to the classical target Hamiltonian $H_Z$. However, some devices have implemented the reverse annealing feature~\cite{Venturelli_2019, Passarelli_2020, Pelofske_2023}, where one prepares the system in the target Hamiltonian $H_Z$ and in a given classical state and for a certain amount of time moves to a Hamiltonian with $\Gamma_t > 0$ and then back to the target Hamiltonian $H_Z$. It is also possible to repeat this process multiple times.\\
In some devices designed to perform optimization via quantum annealing, there are additional constraints on the schedule $\Gamma_t$, such as bounds on the gradient with respect to time $|\partial_t \Gamma_t | \leq K$, or, more generally, Lipschitz continuity $|\Gamma_{t_2} - \Gamma_{t_1}| \leq K |t_2 - t_1|$ for $K=\mathcal{O}(1)$. Here, we show how a large class of $\Gamma_t$, including those respecting such constraints, can still be used for universal quantum computation in the sense that the global TFIM is polynomially equivalent to the gate model of quantum computation.\\

We will use a definition of computational equivalence closely related to the definition used for Hamiltonian simulation~\cite{Kohler_2022}. To this end, let us first define polynomial simulation.
\begin{definition} (Polynomial simulation)
    The computational model $\mathcal{M}_2$ polynomially simulates the computational model $\mathcal{M}_1$ if and only if any computation by $\mathcal{M}_1$ on $n_1$ qubits in time $T_1$ can be done by $\mathcal{M}_2$ on $n_2$ qubits in time $T_2$ up to an acceptable error $\varepsilon > 0$ with $n_2, T_2 = \poly (\varepsilon^{-1}, n_1, T_1)$.
\end{definition}
For digital computing models, the time $T_i$ is usually measured in the number $p$ of basic operations, also called gates, so that $p$ and $T_i$ are often used interchangeably. Based on the definition of polynomial simulation, we can define polynomial equivalence.
\begin{definition} (Polynomial equivalence)
    \label{def:PolyEquivalence}
    Two computational models $\mathcal{M}_1$ and $\mathcal{M}_2$ are polynomially equivalent if and only if $\mathcal{M}_1$ polynomially simulates $\mathcal{M}_2$ and $\mathcal{M}_2$ polynomially simulates $\mathcal{M}_1$.
\end{definition}
Here we will give a constructive method to simulate any quantum circuit consisting of $p$ quantum gates acting on $n_q$ qubits using the global TFIM. As we will argue, the required resources in terms of coupling strength, qubit and time overhead are polynomial. This allows us to make the following statement, which is our main result.
\begin{theorem} (informal)
    The global TFIM with non-monotonic schedules is polynomially equivalent to the gate model of quantum computation.
    \label{theorem:MainResult}
\end{theorem}
A formal version of Theorem~\ref{theorem:MainResult}, together with its proof, is given in Appendix~\ref{sec:AppendixProofOfMainResult}. To show polynomial equivalence, it is required to show that the gate model simulates the global TFIM and vice versa. It is widely known that the gate model efficiently simulates the global TFIM, for example, by Trotterization. Thus, to show our claim, it is necessary to prove that the global TFIM can simulate the gate model, which is what we prove here.\\

The formal version of Theorem~\ref{theorem:MainResult} also states the resource scaling. The global TFIM can simulate a quantum circuit consisting of $p$ gates that act on $n_q$ qubits up to an acceptable error $\varepsilon$ using $n_q' = \mathcal{O}\left( n_q^2 \right)$ qubits, a coupling $|J| = \mathcal{O} \left( p n_q^3 \varepsilon^{-1} \right)$ and evolving over time $T = \mathcal{O} \left( p^{8+2\nu} n_q^{30 + 8\nu} \varepsilon^{-7-2\nu} \right)$ for any $\nu > 0$. The exponents are quite large, which is due to compounding effects by several techniques to reduce approximation errors, as we shall discuss in the following. This includes the need to increase the number of pulses due to weak driving fields and the method of gaining group-wise control of the qubits.

It seems relevant to point out that while our result shows the polynomial equivalence in a strict sense, the exponents are too large for our method to be easily implementable. This is not only a concern with respect to coherence times, but also with respect to the coupling energies~\cite{Harley_2024}. However, our result can also be understood as a no-go-theorem for efficient classical simulability of the global TFIM, assuming that superpolynomial quantum advantage over classical computing exists, which is a widely held assumption in complexity theory. Furthermore, in the proof, we use fairly rudimentary techniques, and we cannot make claims about the tightness of the worst-case scaling we derive here. We conjecture that more favorable bounds can be found by more sophisticated mathematical machinery or by different physical approaches, as has been explored in the literature~\cite{Planckian_2025, Planckian_2026a, Planckian_2026b}.\\
A direct consequence of our finding is that various analog quantum computation platforms, such as those manufactured by D-Wave~\cite{Venturelli_2019, Passarelli_2020, Pelofske_2023}, Pasqal~\cite{Henriet_2020, Scholl_2021, Scholl_2022, Leclerc_2025} and QuEra~\cite{Wurtz_2023, Bombieri_2025} are, in principle, implementations of a universal model of quantum computation. However, this claim holds in the sense of a classification of the theoretical computational model, while concealing substantial issues with short coherence times.

%% file: 03_related_works.tex
\section{Related work}
\label{sec:RelatedWork}
\begin{table*}[]
    \centering
    \begin{tabular}{|c||c|c|c|c|c|}
        \hline
        & Cesa \& Pichler~\cite{Cesa_2023} & Lloyd~\cite{Lloyd_2018} & Hu et al.~\cite{Hu_2026} & Zheng et al.~\cite{Zheng_2024} & this work \\
        \hline \hline
         Fully & Pulses applied exactly & evolution under $H_X$  & three global control & site-resolved & static local control \\
         global control & to two species $\mathcal{A}$ and & and $H_Z$ + control & fields + one field to & $X$-, $Z$- and & and fully global time- \\
         & $\mathcal{B}$ separately & of a single qubit & break symmetry & $ZZ$-terms & dependent control via $\Gamma_t$\\
         \hline
         $|J| < \infty$ & $\times$ & \checkmark & \checkmark & \checkmark & \checkmark \\
         \hline
         Slow control & \checkmark & $\times$ & $\times$ & N/A & \checkmark\\
         \hline
         Error analysis & N/A & \checkmark & $\times$ & \checkmark & \checkmark \\
         \hline
         Poly. equiv. & $\times$ & \checkmark & $\times$ & \checkmark & \checkmark \\
         \hline
    \end{tabular}
    \caption{Comparison of the results in this work with Cesa \& Pichler, Lloyd, Hu et al. and Zheng et al. on whether the control is fully global, whether finite interaction strength is assumed, whether slow control suffices or fast switching is required, whether an error analysis is provided, and whether polynomial equivalence to the gate model follows from the result in the sense of Definition~\ref{def:PolyEquivalence}. For the work by Zheng et al. the question of slow control does not arise, since they consider time-independent Hamiltonians. Furthermore, Cesa and Pichler do not require an error analysis, as their method can realize any unitary exactly.}
    \label{tab:ResultsSummary}
\end{table*}
The universality of a computational model can be shown by establishing computational equivalence with a known universal model. The quintessential universal model of quantum computation is gate-based quantum computation. However, various other models have been proven to be universal over the years, most notably measurement-based quantum computation~\cite{Raussendorf_2001}, adiabatic quantum computation~\cite{Aharonov_2008}, and quantum walks~\cite{Childs_2009}. Furthermore, there is extensive literature on the universality of quantum Hamiltonians, where time-independent Hamiltonians can encode quantum computations in their ground states~\cite{Cubitt_2018, Kohler_2022}. These mappings often require more exotic interactions, or utilize interacting subsystems that go beyond two-level systems. In the context of Hamiltonian simulation, the classical Ising model is known to be universal for classical Hamiltonians~\cite{Cuevas_2016}, i.e. Hamiltonians that are diagonal in the computational basis, while the quantum Ising model with transverse field is universal for stoquastic Hamiltonians~\cite{Bravyi_2017}, i.e. Hamiltonians with non-positive off-diagonal elements in the computational basis. These models have weaker simulation capabilities~\cite{Cubitt_2018}, and are thus considered non-universal for quantum computation. Note that these results are concerned with equilibrium properties. Here, we introduce time-dependence of the transverse field and show that this results in a computationally universal model. Thus, we consider an out-of-equilibrium setting.

In particular, it has been an open question in the literature, whether quantum annealing as a quantum computational model is universal. By quantum annealing we mean programming an optimization problem into a diagonal Ising Hamiltonian $H_Z$ and initializing the system in the ground state of a large transverse field $H_X$, which is slowly turned off over time. The hope is to reach the ground state of $H_Z$ if the turning-off of the transverse field is sufficiently slow, such that as a consequence of the adiabatic theorem, the system remains in, or close to, the ground state.

In this setting, quantum annealing is expected not to be universal for quantum computation. This assertion is in part grounded in the observation that the TFIM is described by a stoquastic Hamiltonian, and thus its equilibrium properties can be computed by Quantum Monte Carlo (QMC) methods~\cite{Bravyi_2015, Hen_2021}. While this does not necessarily imply the absence of quantum advantage, it is believed that stoquastic Hamiltonians cannot encode any given quantum computation. Out of equilibrium, however, QMC does not guarantee efficient simulation and the claims about classical simulation do not hold.

Recently, there have been attempts to show how to reproduce a universal set of quantum gates using quantum annealing hardware. In~\cite{Imoto_2024}, it was shown how to realize individual gates without the ability to compose them into a complete quantum circuit using a global TFIM.

The Quantum Alternating Operator Ansatz, also called the Quantum Approximate Optimization Algorithm~\cite{Farhi_2014} (QAOA), is a discrete time algorithm inspired by quantum annealing. Lloyd showed the polynomial equivalence of QAOA to the gate model by demonstrating its relation to quantum cellular automata~\cite{Lloyd_2018}. Morales et al. augmented the claim with an analysis of the dynamical Lie algebra~\cite{Morales_2020}, showing that under mild assumptions, QAOA can generate any unitary. However, from their argument, the polynomial equivalence to the gate model does not follow, since their result is an existence claim. QAOA differs further from the setting we consider here in that it requires decomposing the time evolution generated by $H_X$ and $H_Z$ into a set of quantum gates, or alternatively that a device can switch arbitrarily fast between the two Hamiltonians. Additionally, mapping to quantum cellular automata as proposed by Lloyd requires the ability to manipulate at least one qubit during computation to initialize the computation, breaking the assumption of global control.

In recent work by Zheng et al. it was shown that time evolution under the time-independent transverse-field Ising model is BQP-complete and therefore also universal for quantum computation~\cite{Zheng_2024}. This result is closely related to the claim we are making here, with the key distinction that the authors assume the transverse field to be site-resolved, i.e. they consider the Hamiltonian
\begin{equation}
    H = \sum_m h_m^x X_m + h_m^z Z_m + \sum_{n>m}  J_{mn} Z_m Z_n \ ,
\end{equation}
which is a finer-grained level of control than is currently available on many large-scale quantum simulators.

In this context, Hu et al. showed by Lie-algebraic means that the time-dependent transverse-field Ising model is universal in a control theoretic sense~\cite{Hu_2026}, meaning that any unitary can in principle be generated. However, the analysis of the resource complexity of the compilation was left for future research. In contrast, here we show that universality can be reached with polynomial resources, thus establishing the polynomial equivalence to the gate model of quantum computation. Hu et al. also assume slightly greater control over the Hamiltonian than we do in our construction here. They assume independent control over the global transverse field $H_X = \sum_m X_m$, the global longitudinal field $H_Z =\sum_m Z_m$, here using Hu et al.'s notation, and the two-body interactions of a one-dimensional chain $H_{ZZ} = \sum_m Z_m Z_{m+1}$. Additionally, their construction requires an additional control field that breaks the reflection symmetry of the one-dimensional chain. In contrast, our equivalence proof assumes a fixed, always-on diagonal Hamiltonian $H_Z$ as in Eq.~\eqref{eq:DefHz} in combination with a time-dependent transverse field with a programmable schedule $\Gamma_t$.

Our proof builds on a result by Cesa \& Pichler~\cite{Cesa_2023}, who gave a construction to simulate arbitrary quantum circuits on globally driven Rydberg atom simulators. The method requires the ability to apply driving pulses to two distinct species of atoms. The atoms are arranged in wires, where atoms of the two species are lined up alternatingly. There is one wire for each logical qubit. Carefully designed pulse sequences can move logical information back and forth on the wires, while impurities within or between wires are used to realize logical single- and two-qubit gates, respectively. Additionally, the Cesa-Pichler method relies on the Rydberg blockade between the atoms to transfer information. It assumes that neighboring atoms interact in an infinitely strong way when both are in the excited state.

The first key assumption, the ability to drive two distinct groups of atoms, does not hold for many analog quantum simulators. The second key assumption, infinitely strong diagonal interactions, also does not hold for many physical systems such as the global TFIM. In our proof, we address both these issues. We show that the independent drives of the qubit species can be simulated by the global TFIM. Furthermore, we show that the infinitely strong Rydberg blockade can also be replaced by a polynomially large but finite Ising interaction $J$. We introduce the Cesa-Pichler method in more detail in the following and provide an overview of the most relevant related works in Table~\ref{tab:ResultsSummary}.

%% file: 04_methods.tex
\section{Methods}
\label{sec:Methods}
\begin{figure*}
    \centering
    \includegraphics[scale=.425]{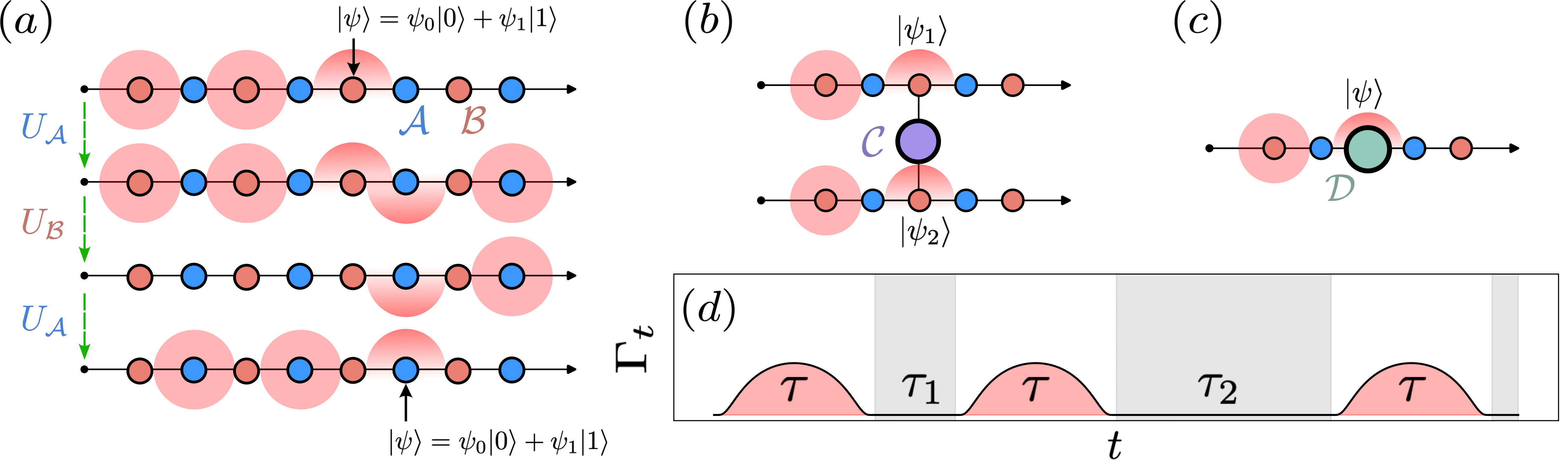}
    \caption{(a): Basic mechanism of the Cesa-Pichler method~\cite{Cesa_2023} to move quantum information along the qubit wires. $U_{\mathcal{A}}$ and $U_{\mathcal{B}}$ represent $X$-gates on the respective qubit species $\mathcal{A}$ (blue) and $\mathcal{B}$ (red), conditional on being blockaded by its neighboring qubits; (b): arrangement for the application of entangling gates, where a $2\pi$-pulse is applied to a qubit of species $\mathcal{C}$ (purple); (c):  qubit arrangement for the application of single-qubit gates in the Cesa-Pichler method, where the gates are executed when the logical qubit is located on a qubit of species $\mathcal{D}$ (green); (d): schedule $\Gamma_t$ consisting of a concatenation of $\Gamma$-pulses over a time $\tau$ (red shaded), in between the $\Gamma$-pulses, the system evolves under $H_Z$ only, i.e. $\Gamma_t = 0$ (gray shaded), resulting in phase gates of time $\tau_i$.}
    \label{fig:BasicScheme}
\end{figure*}
Here, we give a high level overview of the proof, while we refer to the appendix for the technical details. We show the polynomial equivalence by mapping a given quantum circuit from the gate model to a $H_Z$ and $\Gamma_t$ for the global TFIM. To this end, we leverage the method proposed by Cesa and Pichler~\cite{Cesa_2023}, which we will refer to as the CP-method. We will make use of the core ideas proposed in~\cite{Cesa_2023}, but with minor modifications which are more suited for the global TFIM.

\subsection{Universal quantum computation with global control}
Cesa and Pichler propose a method for mapping quantum circuits consisting of $p$ gates acting on $n_q$ qubits to an arrangement of $n_q' = \mathcal{O}(n_q^2)$ Rydberg atoms, together with a sequence of $p' = \mathcal{O}(n_qp)$ global pulses that act on two distinct types of atoms. The two subsets of qubits are denoted by $\mathcal{A}$ and $\mathcal{B}$, respectively. In their proposal, four atoms of each type can be grouped together to form \textit{superatoms}. However, as we will explain below, here we replace the superatoms with two additional groups of qubits $\mathcal{C}$ and $\mathcal{D}$.\\
We first introduce the key realization by Cesa and Pichler, which is that using blockade interactions, it is possible to move logical information along wires of physical qubits in both directions. We then connect wires via qubits of a distinct type, which allows for the implementation of logical entangling gates. Ultimately, we discuss how impurities within a wire can be used to implement logical single-qubit gates.

\subsubsection{Logical qubits on physical wires}
\label{sec:PropagationSequence}
The qubits form a bipartite simple graph $\mathcal{G}$, where qubits of type $\mathcal{A}$ and $\mathcal{C}$ only have neighbors of type $\mathcal{B}$ or $\mathcal{D}$. The edges $E(\mathcal{G})$ of $\mathcal{G}$ correspond to diagonal interactions via the Hamiltonian
\begin{equation}
    \begin{aligned}
        H_J = \frac{J}{4} \sum_{(m,n) \in E(\mathcal{G})} (1-Z_m)(1-Z_n) \ .
    \end{aligned}
    \label{eq:DefHJ}
\end{equation}
In the limit of $J\rightarrow \infty$, these interactions prohibit two adjacent qubits to be in the excited state $\ket{1}$ at the same time, a phenomenon also called the Rydberg blockade. Pulses applied to any of the types $\mathcal{K} \in \{ \mathcal{A}, \mathcal{B}, \mathcal{C}, \mathcal{D} \}$ only act non-trivially if the neighboring qubits are in the ground state. This blockade mechanism is exploited to propagate quantum information along wires of qubits using a particular pulse sequence, as depicted in Figure~\ref{fig:BasicScheme}(a). Each wire in the so-called standard configuration is divided in two phases: a $\mathbb{Z}_2$-ordered phase, where the qubits are in classical states $\ket{...01010101...}$ and a disordered phase where all qubits are in the ground state $\ket{...0000...}$. The logical qubit is encoded in the physical qubit at the boundary of these two phases and can be in any qubit state. In Figure~\ref{fig:BasicScheme}(a), this qubit is annotated with a quantum state $\ket{\psi} = \psi_0 \ket{0} + \psi_1 \ket{1}$. We therefore write the state of the wire in the standard configuration as
\begin{equation}
    \ket{\psi_k} = |\underbrace{..1010}_{k-1-\text{times}} \psi 00...0 \rangle \ ,
    \label{eq:ExampleState}
\end{equation}
where the logical state $\ket{\psi}$ is located at physical qubit $k$. Let us assume for this example that the logical information is located at a $\mathcal{B}$-type qubit, which are depicted in red in Figure~\ref{fig:BasicScheme}(a). By applying a certain sequence of conditional $X$-gates to the $\mathcal{A}$ and $\mathcal{B}$ qubits in an alternating fashion, denoted as $U_{\mathcal{A}}$ and $U_\mathcal{B}$, the logical qubit can be moved to the right, as shown in Figure~\ref{fig:BasicScheme}(a), or to the left, if the alternating sequence were to start with $U_\mathcal{B}$. We will call these sequences the propagation sequence.\\
Consider the gate $U_\mathcal{A}$ to be an $X$-gate applied to all qubits of type $\mathcal{A}$ if none of their respective neighbors are in the excited state. Since the logical state is located at a $\mathcal{B}$-qubit and all $\mathcal{B}$-qubits in the $\mathbb{Z}_2$-ordered phase to the left of the interface are excited, $U_\mathcal{A}$ acts trivially on all $\mathcal{A}$-qubits to the left of the interface. At the same time, all $\mathcal{A}$-qubits in the disordered phase at sites $k+3$ or larger are not blockaded by their neighbors and are free to transition from $\ket{0}$ to $\ket{1}$. Only the $\mathcal{A}$-qubit at location $k+1$ is blockaded if the $\mathcal{B}$-qubit at $k$ is in $\ket{1}$ and free otherwise. Therefore, the full state will transition to
\begin{equation}
    \begin{aligned}
        &U_\mathcal{A} \ket{\psi_k} = U_\mathcal{A} |\underbrace{..1010}_{k-1-\text{times}} \psi 00...0 \rangle \\
        &= \psi_0 |\underbrace{..1010}_{k-1-\text{times}} 01 \ 0101... \rangle + \psi_1 |\underbrace{..1010}_{k-1-\text{times}} 10 \ 0101... \rangle \ .
    \end{aligned} 
\end{equation}
Note that this results in a state that may entangle qubits $k$ and $k+1$, while all other qubits remain in a classical product state. Applying now $U_\mathcal{B}$, a conditional $X$-gate to all $\mathcal{B}$-qubits, will in turn de-excite all $\mathcal{B}$-qubits in the $\mathbb{Z}_2$-ordered phase to the left of the interface, i.e. at locations $k-2$ or smaller, while all $\mathcal{B}$-qubits to the right of the interface are blockaded. Only the $\mathcal{B}$ at position $k$ will de-excite if it originally was excited, or remain in its ground state due to the blockade by the interaction with the $\mathcal{A}$-qubit at $k+1$. This gives the state
\begin{equation}
    \begin{aligned}
        &U_\mathcal{B}U_\mathcal{A} \ket{\psi_k}\\
        &=\psi_0 |\underbrace{..0000}_{k-1-\text{times}} 01 \ 0101... \rangle + \psi_1 |\underbrace{..0000}_{k-1-\text{times}} 00 \ 0101... \rangle \ ,
    \end{aligned}
    \label{eq:NonStandardConfig}
\end{equation}
which means that we have moved the logical information one site to the right and applied an $X$-gate to it. Furthermore, note that the $\mathbb{Z}_2$-ordered phase is now on the right of the interface, so the state in Eq.~\eqref{eq:NonStandardConfig} is not yet in the standard configuration. The standard configuration can be recovered by again applying an $U_\mathcal{A}$-gate, which results in
\begin{equation}
    \begin{aligned}
        &U_\mathcal{A} U_\mathcal{B} U_\mathcal{A} \ket{\psi_k}\\
        &=\psi_0 |\underbrace{..0101}_{k-1-\text{times}} 00 \ 0000... \rangle + \psi_1 |\underbrace{..0101}_{k-1-\text{times}} 01 \ 0000... \rangle \\
        &=\psi_0 |\underbrace{..01010}_{k-\text{times}} \ 0 \ 00... \rangle + \psi_1 |\underbrace{..01010}_{k-\text{times}} \ 1 \ 00... \rangle \\
        &= |\underbrace{..01010}_{k-\text{times}} \psi 00... \rangle \\
        &= \ket{\psi_{k+1}} \ .
    \end{aligned}
\end{equation}
Now, the logical state $\ket{\psi}$ is located at qubit $k+1$, which is of type $\mathcal{A}$. By the same argument presented here, it is now easy to show that the sequence $U_\mathcal{B} U_\mathcal{A} U_\mathcal{B}$ produces the state $\ket{\psi_{k+2}}$. Similarly, applying $U_\mathcal{A} U_\mathcal{B} U_\mathcal{A}$ would move the logical state back to $\ket{\psi_k}$. Thus, by alternatingly applying $U_\mathcal{A}$ and $U_\mathcal{B}$, the logical information can be moved left and right along the wire, while skipping $U_\mathcal{A}$ or $U_\mathcal{B}$, depending on the current location of the logical information, changes the direction of movement.\\

We will introduce one wire for every logical qubit. The wires will have the same ordering of $\mathcal{A}$- and $\mathcal{B}$-qubits, and the logical information will be located at the same site $k$. Then, applying the propagation sequence globally will move all logical information along the wires in parallel. This arrangement is depicted in Figure~\ref{fig:UniversalConfig}. There are two more qubit types in this arrangement which allow for processing of the logical information, as we will now demonstrate. We will first show how to entangle qubits on distinct wires. Note that in the qubit arrangement in Figure~\ref{fig:UniversalConfig} there are impurities not only between the wires but also within. These are the $\mathcal{D}$-qubits (green), which we will use to apply logical single-qubit gates, as explained in the following. The $\mathcal{D}$-qubits always replace a $\mathcal{B}$-qubit in the chain and are, for the purpose of the propagation sequence, considered to be $\mathcal{B}$-qubits, which means that when $U_\mathcal{B}$ is applied, we apply a conditional $X$-gate to the $\mathcal{B}$- and $\mathcal{D}$-qubits at the same time.

\subsubsection{Logical entangling gates between wires}
\label{sec:LogicalEntanglingGates}
In order to process many-qubit states, each logical qubit is mapped to its own wire of physical qubits. Each wire is of the same length, and the logical information moves on each wire in parallel to the left or to the right, according to the propagation sequence. Impurities in the wires, as well as between the wires, allow one to apply logical single-qubit and entangling gates. The placement of the impurities is shown in Figure~\ref{fig:BasicScheme}(b) and (c). The impurities consist of qubits of the types $\mathcal{C}$ and~$\mathcal{D}$.\\
The qubits of type $\mathcal{C}$ are placed between the wires, as depicted in Figure~\ref{fig:BasicScheme}(b). They are not driven during the propagation sequence and remain in their initial ground state $\ket{0}$. They are driven only when entangling gates are applied to the logical qubits on the adjacent wires. When the logical qubit on each wire is adjacent to the $\mathcal{C}$-qubit, we can apply a $2\pi$-pulse to the $\mathcal{C}$-qubit. If one or both of the adjacent qubits are in the excited state, the $\mathcal{C}$-qubit is blockaded and the $2\pi$-pulse acts trivially. However, if both adjacent qubits are in the ground state, the $\mathcal{C}$-qubit can perform the $2\pi$-rotation. This will move the $\mathcal{C}$-qubit back to $\ket{0}$, but the state will pick up a phase $-1$. Thus, the $2\pi$-pulse on the $\mathcal{C}$-qubits acts as a logical controlled phase gate $\mathbb{I} - 2\ket{00}\bra{00}$.\\
When considering this technique in the setting of more than two logical qubits, as depicted in Figure~\ref{fig:UniversalConfig}, the arrangement of qubits is such that at most two wires are connected by a $\mathcal{C}$-qubit at the position of the information carrying interface. Then, the $\mathcal{C}$-qubits in the $\mathbb{Z}_2$-ordered phase are trivially blockaded and do not contribute a phase, while the $\mathcal{C}$-qubits in the disordered phase are trivially free to perform the $2\pi$-rotation and thus only contribute inconsequential global phases.

\subsubsection{Logical single-qubit gates}
\label{sec:Logical1QBGates}
The second type of impurity, in our case qubits of type $\mathcal{D}$, is introduced to the wires, where they replace a $\mathcal{B}$-qubit. If we propagate the logical information along the wire, the $\mathcal{D}$-qubits are driven with a conditional $X$-gate together with the $\mathcal{B}$-qubits. In order to implement logical single-qubit gates, we proceed as follows. Assume the logical information is located at a $\mathcal{D}$-qubit at location $k$ on the wire where we want to apply a gate. The information on the other wires is then located on $\mathcal{B}$-qubits by design, as shown in Figure~\ref{fig:UniversalConfig}. By applying a $2\pi$-pulse to the blue $\mathcal{A}$-qubits, the $\mathcal{A}$-qubits in the $\mathbb{Z}_2$-ordered part of the wire are all blockaded and thus the pulse acts as the identity. The $\mathcal{A}$-qubits in the disordered phase at locations $k+3$ or larger of the wires are all free to rotate and thus will pick up inconsequential global phases. Only the $\mathcal{A}$-qubits at location $k+1$ to the immediate right of the interface, where the logical qubits are located, will be blockaded if the logical qubit is in $\ket{1}$ and the $2\pi$-pulse will act as the identity, or the logical qubit is in $\ket{0}$ and the $\mathcal{A}$-qubit can perform the $2\pi$-rotation, resulting in a phase $-1$. Therefore, the $2\pi$-pulse on the $\mathcal{A}$-qubits acts as $-Z$-gates on all the logical qubits. We can use this logical global phase gate denoted by
\begin{equation}
    Z_{\text{total}} = \bigotimes_\nu (-Z) \ ,
\end{equation}
where the index $\nu$ runs over all logical qubits, to single out an impurity at the interface and apply single-qubit gates.\\
Consider the rotation gates
\begin{equation}
    \begin{aligned}
        R_X(\alpha) = \cos(\alpha /2) \mathbb{I} - i \sin(\alpha/2) X  \ , \\
        R_Y(\alpha) = \cos(\alpha /2) \mathbb{I} - i \sin(\alpha/2) Y \ .
    \end{aligned}
\end{equation}
We can realize single-qubit gates if the interface is located at a $\mathcal{D}$-qubit for the wire we wish to manipulate, while the interface on all other wires is located at a $\mathcal{B}$-qubit. Then we apply $R_X(\alpha)$ to the $\mathcal{D}$-qubits, followed by $Z_{\text{total}}$ and the inverse rotation $R_X(-\alpha)$ again applied to the $\mathcal{D}$-qubits, where $Z_{\text{total}}$ is applied via a $2\pi$-rotation on the $\mathcal{A}$-qubits, as described above. The $\mathcal{A}$-qubits in the $\mathbb{Z}_2$-phase are all blockaded, so the $2\pi$-pulse acts as an identity, while in the disordered phase, each $\mathcal{D}$-qubit has two $\mathcal{A}$-neighbors by construction, thus the phases $-1$ from both qubits' rotations cancel. Only the $\mathcal{A}$-qubit at the immediate right of the interface experiences the $2\pi$-pulse and thus the $\mathcal{D}$-qubit at the interface experiences a $-Z$ gate. As a consequence, the $\mathcal{D}$-qubits in the $\mathbb{Z}_2$-ordered and the disordered phase experience $R_X(-\alpha)R_X(\alpha) = \mathbb{I}$, while the $\mathcal{D}$-qubit at the interface experiences $R_X(-\alpha) Z R_X(\alpha) = ZR_X(2\alpha)$. The $\mathcal{B}$-qubits at the interface also experience the global phase gate $Z_{\text{total}}$. This can be corrected by applying $Z_{\text{total}}$ again after the $R_X$ acting on the $\mathcal{D}$-qubits, which will also correct the superfluous $Z$ acting on the $\mathcal{D}$-qubit at the interface. In summary, the steps
\begin{enumerate}
    \item apply $R_X(\alpha)$ to all $\mathcal{D}$-qubits
    \item apply $Z_{\text{total}}$ (i.e. $2\pi$-rotation on all $\mathcal{A}$-qubits)
    \item apply $R_X(-\alpha)$ to all $\mathcal{D}$-qubits
    \item apply $Z_{\text{total}}$
\end{enumerate}
will act trivially on all qubits, except the $\mathcal{D}$-qubit at the interface, where it acts as a $R_X(2\alpha)$-gate. The same argument holds for $R_Y(\alpha)$, which allows us to apply $R_Y(2\alpha)$. Since any single-qubit gate can be decomposed into sequences of $X$- and $Y$-rotations, arbitrary single-qubit gates can be realized acting only on the $\mathcal{D}$-qubit at the interface. Therefore, together with the logical controlled phase gate, a universal set of quantum gates on the logical qubits can be realized.

\subsubsection{Initialization}
\label{sec:InitializationSequence}
So far, we have assumed that the wire is in the standard configuration, where the logical quantum information is located at the interface of a $\mathbb{Z}_2$-ordered and a disordered phase, where all qubits are in the state $\ket{0}$. However, in the beginning of each computation, all qubits are in the ground state $\ket{0}$. If the first qubit in a wire is of type $\mathcal{D}$, it only has one neighbor of type $\mathcal{A}$ and thus single-qubit gates can be applied targeted to the first qubit. This works by the same procedure as above, since all $\mathcal{D}$-qubits within the wires have two neighbors of type $\mathcal{A}$, so that only the $\mathcal{D}$-qubits at the beginning of the wire experience the phase gates $Z_{\text{total}}$.

The ability to apply single-qubit gates to the first $\mathcal{D}$-qubit at location $k=1$ allows to initialize the wire. First, we excite the qubits at $k=1$ and move the excitations to $k=3$ via the propagation sequence. Then, we excite the qubit at $k=1$ again, and apply again the propagation sequence. This produces the state $\ket{10101000...}$ on each wire. Logical single-qubit gates can be applied once the interface is located at $k=5$. If the interface was at $k=3$, single-qubit gates would not be possible, since the $\mathcal{A}$-qubit at $k=2$ would have two $\mathcal{D}$-neighbors, which hinders the application of $Z_{\text{total}}$ to the $\mathcal{D}$-qubit at the interface only.

\subsubsection{Universal arrangement}
These considerations allow us to arrange the qubits in a universal arrangement, as depicted in Figure~\ref{fig:UniversalConfig}. Here, an example for four logical qubits is shown. For any number of logical qubits $n_q$, we refer to this graph as the universal arrangement $\mathcal{G}$. In the configuration proposed by Cesa and Pichler, to simulate $p$ quantum gates acting on $n_q$ qubits, there are $n_q$ wires each representing one logical qubit. Each wire needs to be $\mathcal{O}(n_q)$ qubits long to be able to create the staggered configuration of $\mathcal{D}$-qubits between the wires to realize logical entangling gates. For each of the $p$ gates, it is required to move the logical information along the wires to the respective impurities. This requires $\mathcal{O}(n_q)$ pulses from the propagation sequence to move to the appropriate location for each of the $p$ logical gates. Hence, the CP-method requires $n_q' = \mathcal{O}(n_q^2)$ physical qubits and $p' = \mathcal{O}(p n_q)$ pulses. The universal arrangement $\mathcal{G}$ corresponds to $n_q$ logical qubits on a one-dimensional chain with nearest-neighbor interactions. Interactions between qubits that are further apart must be compiled using SWAP gates.

\subsubsection{Summary of the CP-method}
In summary, the CP-method allows to simulate a given quantum circuit $U_{\text{target}}$ consisting of $p$ gates acting on $n_q$ qubits exactly by following these steps on a device consisting of $n_q'=\mathcal{O}(n_q^2)$ physical qubits arranged on a graph $\mathcal{G}$ consisting of wires and impurities, as shown in Figure~\ref{fig:UniversalConfig}:
\begin{enumerate}
    \item Given a logical quantum circuit
    \begin{equation}
        U_{\text{target}} = U_p U_{p-1} ... U_1
    \end{equation}
    consisting of $p$ gates acting on $\ket{0}^{\otimes n_q}$. Compile each of the logical gates $U_i$ into $R_X$ and $R_Y$, as well as controlled phase gates $\mathbb{I} - 2\ket{00} \bra{00}$ on an open chain of qubits with nearest-neighbor interactions.
    \item Initialize the device in $\ket{0}^{\otimes n_q'}$, where $n_q' = \mathcal{O}(n_q^2)$. 
    \item For each logical gate obtained from step 1, apply the propagation sequence from Section~\ref{sec:PropagationSequence} to move the interface to the respective impurity on or between wires and apply the logical quantum gate as described in Section~\ref{sec:LogicalEntanglingGates} and Section~\ref{sec:Logical1QBGates}. Prepend the pulses required to run the initialization sequence described in Section~\ref{sec:InitializationSequence}. This gives a sequence of $p' = \mathcal{O}(pn_q)$ gates
    \begin{equation}
        U_{CP} = U_{\mathcal{I}_{p'},p'} U_{\mathcal{I}_{p'-1},p'-1} ... U_{\mathcal{I}_{1},1} 
    \end{equation}acting on subsets of qubits $\mathcal{I}_i \in \{ \mathcal{A}, \mathcal{B}, \mathcal{C},\mathcal{D}, \mathcal{B} \cup \mathcal{D} \}$. 
    \item After applying all the necessary gates, either move the logical information to a dedicated register for further processing (e.g. read-out) by appending the required gates to $U_{CP}$ or proceed using only the qubits where the logical information is located. We denote the procedure as the partial trace $\Tr_{RO}$ over all qubits that do not carry the logical information. Since the CP-method is exact, this will result in
    \begin{equation}
        \begin{aligned}
            &U_{\text{target}} \ket{0}^{\otimes n_q} \bra{0}^{\otimes n_q} U_{\text{target}}^\dagger = \\
            &\Tr_{RO} \left( U_{CP} \ket{0}^{\otimes n_q'} \bra{0}^{\otimes n_q'} U_{CP}^\dagger \right) \ .
        \end{aligned}
        \label{eq:CPMethodRO}
    \end{equation}
\end{enumerate}
Having introduced the CP-method, we will continue to argue how it can be implemented using the global TFIM.
\begin{figure}
    \centering
    \includegraphics[scale=.45]{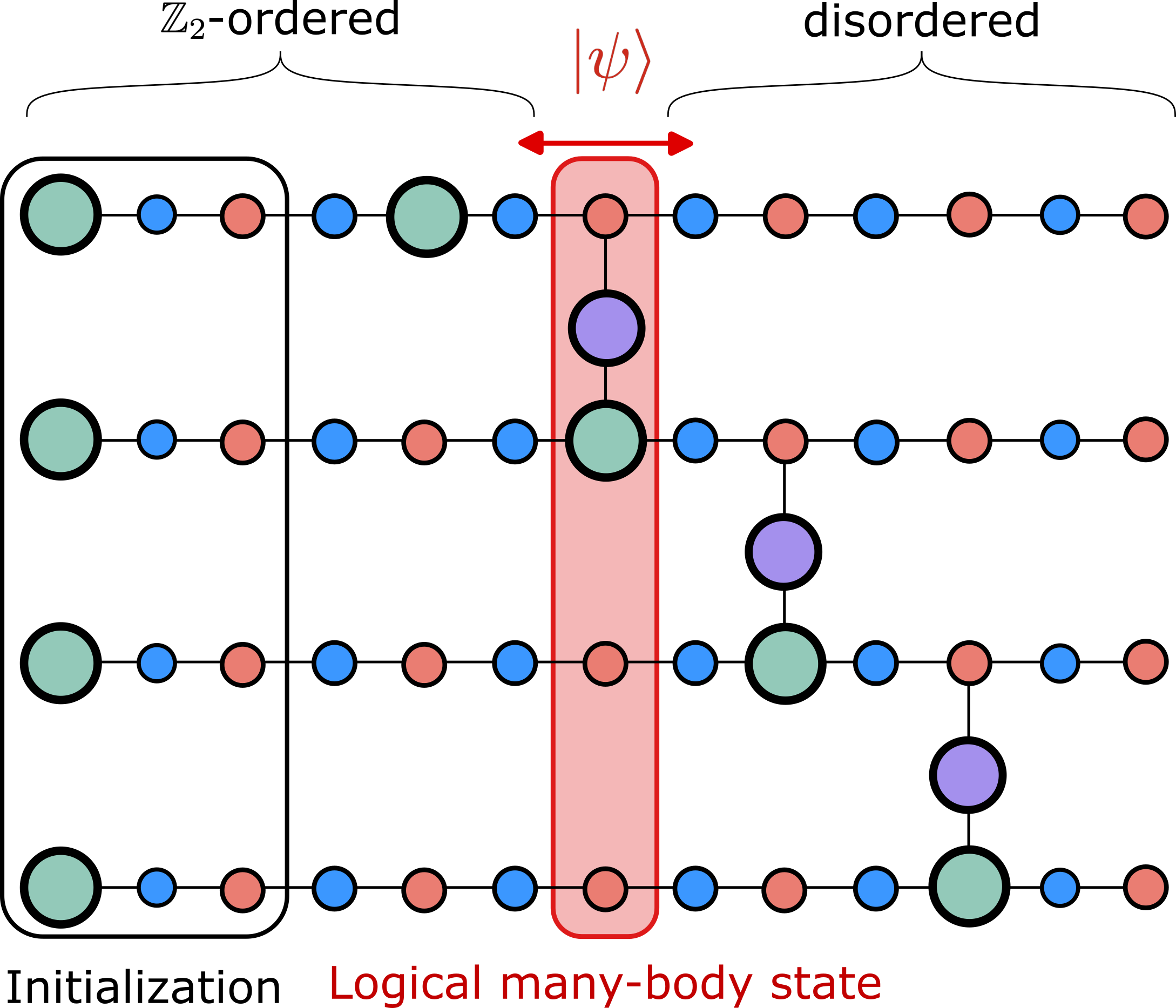}
    \caption{Graph $\mathcal{G}$ of the universal arrangement proposed by Cesa \& Pichler for a system of four logical qubits. The first three qubits (black box) on the left-hand side of each chain serve to initialize the standard configuration, while every chain has an impurity for single-qubit gates and a connection via an impurity to its neighbors for logical two-qubit gates. Then, depending on the order of pulses in the propagation cycle, the logical many-body state (red box) can move left and right, so that logical operations on any qubit and between neighboring pairs can be executed. This topology corresponds to a chain of five logical qubits with nearest-neighbor couplings. Interactions between arbitrary pairs can be compiled efficiently into logical SWAP gates.}
    \label{fig:UniversalConfig}
\end{figure}

\subsection{Polynomial equivalence of the global TFIM and the gate model of quantum computation}
Here we outline the proof of Theorem~\ref{theorem:MainResult}. The full technical detail is given in the appendix. Although we have introduced the CP-method in the previous sections with minor modifications, we will deviate further in two key points:
\begin{enumerate}
    \item The CP-method assumes the ability to target the qubit groups $\mathcal{A}$ and $\mathcal{B}$ and their superatoms with specific pulses, which we have modified to assume four independently addressable qubit groups $\mathcal{A}$, $\mathcal{B}$, $\mathcal{C}$ and $\mathcal{D}$, where arbitrary single-qubit gates can be applied to each group. However, for the global TFIM, we can only assume a fully global drive via the homogeneous transverse fields, as illustrated in Figure~\ref{fig:CPvsTFIM}. We will show how the required group-wise control can be achieved under global control.
    \item Cesa and Pichler consider the limit of infinitely strong blockade interactions $|J|$, while we consider finite fields and couplings. For a rigorous claim to polynomial equivalence, it is required to show that the interaction strength is at most polynomial in the system size $n_q$ and circuit depth $p$.
\end{enumerate}
Showing that the CP-method can be implemented on the global TFIM given these two restrictions is our main analytical contribution. We first show how to reach universality in the limit of infinite blockade interactions $|J|$, and in a second step bound the errors due to finite $|J|$.
\begin{figure}
    \centering
    \includegraphics[scale=.4]{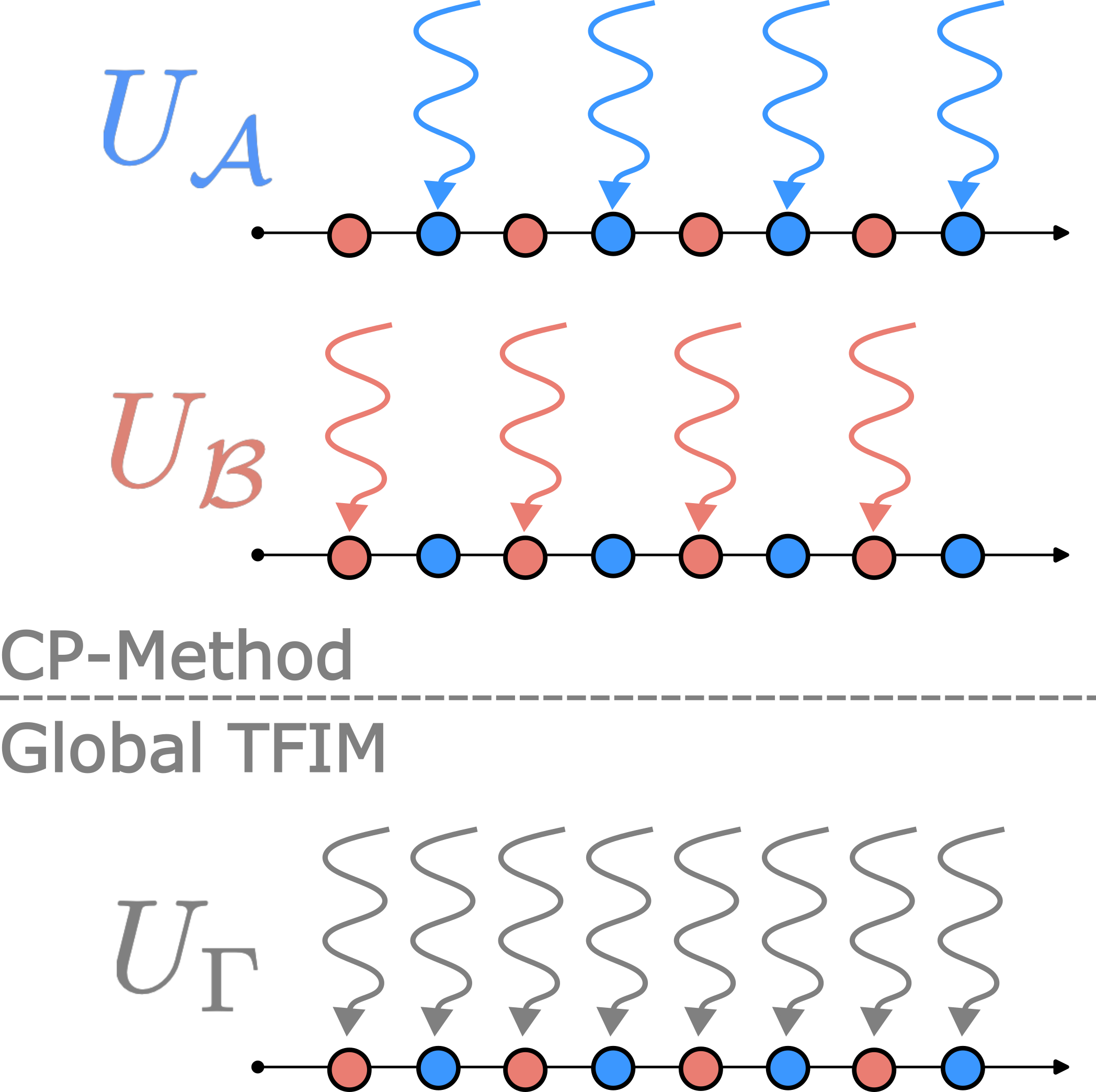}
    \caption{Illustration of the assumptions of semi-global control as per the Cesa-Pichler method (top) and the global drive due to the transverse field generating $\Gamma$-pulses on all qubits simultaneously (bottom).}
    \label{fig:CPvsTFIM}
\end{figure}

\subsubsection{Group-level control from a global transverse field}
The qubit types $\mathcal{K} \in \{ \mathcal{A}, \mathcal{B}, \mathcal{C}, \mathcal{D} \}$ can be distinguished by introducing distinct longitudinal local fields to the Hamiltonian $H_Z$, where the $Z_n$ coefficient is $h_\mathcal{K}$ if $n\in\mathcal{K}$. Thus, the diagonal part of the Hamiltonian $H_Z$ reads
\begin{equation}
    \begin{aligned}
        H_Z =& \sum_{\mathcal{K} \in \{ \mathcal{A}, \mathcal{B}, \mathcal{C}, \mathcal{D} \}} \sum_{n \in \mathcal{K}} h_\mathcal{K} Z_n + H_J \ ,
    \end{aligned}
    \label{eq:RydbergHamiltonian}
\end{equation}
where $H_J$ is defined as in Eq.~\eqref{eq:DefHJ}. It is easy to see that this is a particular choice of $H_Z$ in Eq.~\eqref{eq:DefHz} up to a constant inconsequential offset.\\
We design the transverse field $\Gamma_t$ to drive the qubits. We construct $\Gamma_t$ from a sequence of what we call $\Gamma$-pulses. A $\Gamma$-pulse is given by the unitary
\begin{equation}
    U_{0,\Gamma} = \mathcal{T}\exp \left( -i \int_0^\tau ( H_Z + \Gamma_t H_X ) dt\right) \ ,
\end{equation}
where over some fixed amount of time $\tau$, $\Gamma_t \neq 0$. Between the $\Gamma$-pulses, we let the system evolve under $H_Z$ only, i.e. $\Gamma_t = 0$. By choosing the waiting times $\tau_i$ between the $\Gamma$-pulses, we gain universal single-qubit control of the distinct qubit groups $\mathcal{K}$, as per Lemma~\ref{lemma:UniversalSU2}~\cite{Werner_2026}, i.e. we can apply arbitrary single-qubit gates to all qubits of group $\mathcal{K}$. This is possible if we choose the $h_\mathcal{K}$ to be $\mathbb{Q}$-linearly independent algebraic numbers, i.e. algebraic numbers not related by rational linear combinations. Then, we are able to quasi-independently control the phases each qubit group picks up during the waiting time, up to an arbitrarily small error $\varepsilon_{\text{phase}} > 0$, as per Theorem~\ref{theorem:FillingTime}~\cite{Bourgain_1998, Berti_2003} in Appendix~\ref{sec:AppendixAuxLemmas}. The cost is that the waiting times need to be chosen $\tau_i = \mathcal{O}(\varepsilon_{\text{phase}}^{-3-\nu})$ for any $\nu > 0$, where the choice of $\nu$ changes the prefactor of the asymptotic scaling. In fact, this method of obtaining individual control is the main factor driving up the exponents in the time overhead of this reduction. Note that Lloyd used this technique to single out individual terms in the Hamiltonian to generate the dynamics of the quantum cellular automaton~\cite{Lloyd_2018}. In contrast, we will use it to generate single-qubit gates independently from a global pulse. While $U_{0, \Gamma}$ stimulates transitions of the qubits that are not blockaded by their neighbors, a sequence of $U_{0,\Gamma}$ with controllable phases between them allows to construct arbitrary single-qubit gates on non-blockaded qubits. An illustrative example of a $\Gamma_t$ consisting of $\Gamma$-pulses for times $\tau$ and waiting times $\tau_i$ between them is shown in Figure~\ref{fig:BasicScheme}(d).\\
In the limit of $|J|\rightarrow \infty$, $H_J$ energetically forbids transitions between the subspace where all blockade constraints are respected, i.e. the subspace spanned by computational basis states $\ket{z}$ where no two adjacent qubits are in $\ket{1}$, and the subspace where at least one constraint is violated. The subspace where all blockade constraints are respected is denoted by $\mathcal{P}$, while the projector onto this subspace is given by $P$. In Appendix~\ref{sec:ComputationalSubspaceDefinition} we characterize $P$ further. Here, it suffices to realize that for $|J|\rightarrow \infty$, the unitary generated by $H_t$ reads
\begin{equation}
    \begin{aligned}
        U_{\Gamma} &= \mathcal{T} \exp \left( -i\int_0^\tau P H_t P dt \right) \\
        &= \mathcal{T} \exp \left( -i\int_0^\tau  \left( \sum_n h_n Z_n + \Gamma_t P X_n P \right) dt \right) \ ,
    \end{aligned}
    \label{eq:GammaPulseMainTextDefinition}
\end{equation}
since $PH_JP = 0$ and $[P, Z_n] = 0$. The $\Gamma$-pulse in Eq.~\eqref{eq:GammaPulseMainTextDefinition} is used to stimulate transitions between $\ket{0}$ and $\ket{1}$ for each qubit conditioned on the state of its neighbors. Letting the system evolve under $H_Z$, i.e. $\Gamma_t = 0$ for a time $\tau_i$ realizes
\begin{equation}
    U_Z(\tau_i) = \exp \left( -i \tau_i H_Z\right) \ ,
\end{equation}
which acts on the computational subspace $\mathcal{P}$ as a phase gate $\exp \left( -i \tau_i \sum_n h_n Z_n\right)$. In Appendix~\ref{sec:AppendixProofOfMainResult}, we argue that sequences of $\Gamma$-pulses and phase gates $U_Z$ can realize arbitrary single-qubit gates for each of the qubit groups conditioned on the state of the neighboring qubits. This is shown in Lemma~\ref{lemma:UniversalGateSequence}. Concretely, we can construct a sequence of $\Gamma$-pulses and phase gates that for any gate from the CP-method $U_{\mathcal{I}_i, i}$ act on the computational subspace $\mathcal{P}$ like $U_{\mathcal{I}_i, i}$ up to a controllable error for some appropriate sequence length $N$, i.e.
\begin{equation}
    \begin{aligned}
        &\left\| \left(U_Z(\tau_0) \prod_{i=1}^N U_\Gamma U_Z(\tau_i) - U_{\mathcal{I}_i, i} \right)P \right\| \\
        &= \mathcal{O}(n_q' \varepsilon_{\text{phase}}/ \Gamma^* + n_q'^2 \Gamma^*) \ ,
    \end{aligned}
    \label{eq:PerGateErrorMainText}
\end{equation}
where $\Gamma^* = \max_{t\in [0, \tau]} |\Gamma_t|$ is the maximum of the $\Gamma$-pulse,
which must be chosen carefully to control the error, as we discuss now. In Appendix~\ref{sec:AppendixProofOfMainResult}, we will also argue that $\varepsilon_{\text{phase}}$ can be controlled independently from $\Gamma^*$ and therefore the error can be made arbitrarily small.\\
The fully global drive $U_\Gamma$, in comparison to the microwave pulses on distinct groups, implies interactions between adjacent qubits, which results in unwanted blockade interaction between the qubits. This introduces an error into the time evolution. However, this error can be controlled by choosing the $\Gamma$-pulse such that $\Gamma^*$ is small, i.e. we require precise control, where $\Gamma_t$ only deviates slightly from zero. This way, the amplitude of exciting each qubit is small, $\mathcal{O}(\Gamma^*)$, and the amplitude of exciting adjacent qubits, and thus the error, is even smaller, namely $\mathcal{O}(\Gamma^{*2})$. Let us denote with $\mathcal{I}$ the set of qubits to be driven. Note that for the CP-method, $\mathcal{I} \in \{ \mathcal{A}, \mathcal{B}, \mathcal{C}, \mathcal{D}, \mathcal{B} \cup \mathcal{D} \}$, while we denote its complement by $\mathcal{I}^c$. By appropriately choosing the phases acting on $\mathcal{I}$ and $\mathcal{I}^c$ respectively, a second $\Gamma$-pulse can de-excite the qubits in $\mathcal{I}^c$, while amplifying the drive on the qubits in  $\mathcal{I}$. We show this rigorously in Lemma~\ref{lemma:SecondStep} in Appendix~\ref{sec:AppendixProofOfMainResult}. The mechanism is illustrated in Figure~\ref{fig:GlobalDriveIssue}. Thus, we can compile all the gates required by the CP-method in sequences of $\Gamma$-pulses and appropriately chosen phases. As we argue in Appendix~\ref{sec:Appendix:UniversalSU2} based on Lemma~\ref{lemma:UniversalSU2}~\cite{Werner_2026}, this compilation can be done efficiently.\\
On the other hand, reducing $\Gamma^*$ also implies that we need longer concatenations of $\Gamma$-pulses to realize the required gates, i.e. $N$  grows. However, using Lemma~\ref{lemma:UniversalSU2}, we show that we require $N = \mathcal{O}(\Gamma^{*-1})$ $\Gamma$-pulses, while the error scales as $\mathcal{O}(\Gamma^{*2})$, so overall the error vanishes as $\mathcal{O}(\Gamma^{*})$. Considering that the phase gates also have an error $\varepsilon_{\text{phase}}$, this results in the per-gate error shown in Eq.\eqref{eq:PerGateErrorMainText}.
\begin{figure}
    \centering
    \includegraphics[scale=.8]{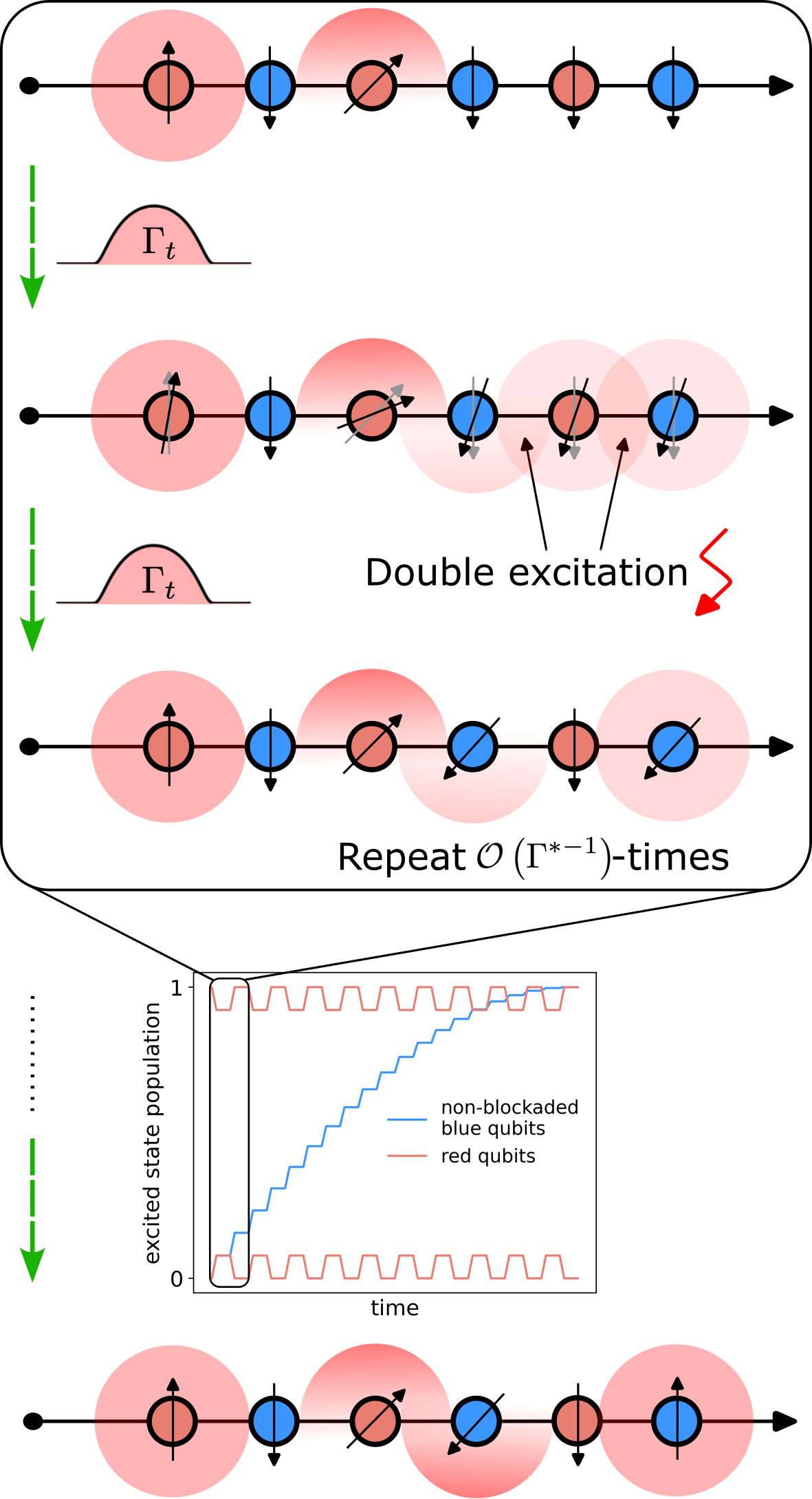}
    \caption{Illustration of the CP-method with a global drive. The black box in the main figure and the inset correspond to each other. In the CP-method it is assumed that distinct groups of atoms can be driven independently. Here, we only have the ability to apply a transverse field to all qubits at the same time. Therefore, adjacent qubits blockade each other due to the global drive, resulting in errors due to these unwanted blockade interactions. We can recover individual control of non-adjacent qubit groups using specifically chosen local $Z$-fields and exploiting constructive and destructive interference. By appropriately chosen phase-gates between the $\Gamma$-pulses, the non-driven qubits can be immediately de-excited, while the drive on the driven qubits is amplified.}
    \label{fig:GlobalDriveIssue}
\end{figure}

\subsubsection{Error due to finite blockade interactions}
The last step is to show that the error resulting from finite $|J|$ can be controlled at a reasonable cost. This can be done using the Schrieffer-Wolff (SW) transformation~\cite{Schrieffer_1966, Bravyi_2011}. The SW transformation is given by a unitary $\exp(S_t)$ with a time-dependent anti-Hermitian generator $S_t$. Choosing $S_t$ appropriately, it is possible to approximately block-diagonalize a Hamiltonian into the computational subspace and its complement, provided that these subspaces are energetically separated. In this coordinate system, the Hilbert space decomposes into the computational subspace $\mathcal{P}$ and its complement $\mathcal{Q}$, where $\mathcal{Q}$ is spanned by the computational basis states that violate the blockade constraint. It is possible to choose $S_t$ such that its norm is bounded by $\mathcal{O}(|J|^{-1})$ (Lemma~\ref{lemma:SBounds}), which then allows one to prove a similar bound on the error between the effective Hamiltonian and the desired one (Lemma~\ref{lemma:HamiltonianErrorBound}). The resulting error in the time evolution is then easily bounded using Lemma~\ref{lemma:TimeEvolutionError}, from which, together with Lemma~\ref{lemma:FullConditionalGateWithError}, follows a bound on the error $\delta_{\text{gate}}$ for the compilation of a given gate $U_{\mathcal{I}_i,i}$ from the CP-method into $\Gamma$-pulses and phase gates
\begin{equation}
    \delta_{\text{gate}} = \mathcal{O}\left( n_q' \varepsilon_{\text{phase}} / \Gamma^{*} + n_q'^2 \Gamma^* + n_q' |J|^{-1} \right) \ ,
\end{equation}
where $\varepsilon_{\text{phase}}$ is the maximal error on the phase gates, which can be controlled by choosing the waiting times $\tau_i = \mathcal{O} \left( \varepsilon_{\text{phase}}^{-3-\nu}\right)$. Since the CP-method requires $p' = \mathcal{O}\left( p n_q\right)$ gates, this results in a total error $\delta_{\text{total}} = \mathcal{O}(p' \delta_{\text{gate}})$. Using that $n_q' = \mathcal{O}\left( n_q^2 \right)$, the error $\delta_{\text{total}}$ can then be guaranteed to remain below an acceptable error $\varepsilon$ by choosing the coupling $|J| = \mathcal{O} \left( p n_q^3 \varepsilon^{-1} \right)$ and the maximal transverse field $\Gamma^* = \mathcal{O} \left( p^{-1} n_q^{-5} \varepsilon \right)$. The required evolution time $T = \mathcal{O} \left( p^{8+2\nu} n_q^{30 + 8\nu} \varepsilon^{-7-2\nu} \right)$ can then be estimated by counting the number of $\Gamma$-pulses and phase gates. Since $n_q', T, |J| = \poly \left( n_q, p, \varepsilon^{-1} \right)$, this concludes the proof of polynomial equivalence.

\subsection{Summary of our method}
Our proposed constructive method proceeds as follows:
\begin{enumerate}
    \item Choose the longitudinal fields $h_\mathcal{K}$ for $\mathcal{K} \in \{ \mathcal{A}, \mathcal{B}, \mathcal{C}, \mathcal{D}\}$ to be algebraic $\mathbb{Q}$-linearly independent numbers, a strong coupling $J$ and the Ising Hamiltonian $H_Z$ with connectivity graph $\mathcal{G}$ (as illustrated in Figure~\ref{fig:UniversalConfig}).
    \item Given a quantum circuit, determine the pulse sequence according to the CP-method.
    \item Compile the computed CP-pulse sequence into global $\Gamma$-pulses and a sequence of waiting times $\tau_i$ to obtain the full schedule $\Gamma_t$. $\Gamma_t$ evolves the system for time $T$ according to Hamiltonian Eq.~\eqref{eq:DefFullHamiltonian}.
    \item Execute $U_T = \mathcal{T} \exp \left( -i\int_0^T H_t dt \right)$ with $H_t = \Gamma_t H_X + H_Z$.
    \item For appropriately chosen $n_q', T, |J| = \poly \left( n_q, p, \varepsilon^{-1} \right)$, this results in
    \begin{equation}
        \begin{aligned}
            &\left\| U_{CP} \ket{0}^{\otimes n_q'} - U_{T} \ket{0}^{\otimes n_q'}  \right\| \leq \frac{\varepsilon}{2} \ .
        \end{aligned}
    \end{equation}
    Together with Eq.~\eqref{eq:CPMethodRO} and the standard identity
    \begin{equation}
        \begin{aligned}
            &\| U \ket{\psi}\bra{\psi}U^\dagger - V \ket{\psi}\bra{\psi}V^\dagger \| \\
            &\leq 2 \| U - V\|
        \end{aligned}
    \end{equation}
    for any state $\ket{\psi}$ and any unitaries $U$ and $V$, this shows that $U_T$ simulates an arbitrary circuit $U_{\text{target}}$ up to an error $\varepsilon > 0$ on a subset of qubits.
\end{enumerate}
Using this procedure, an arbitrary quantum circuit can be simulated by a time-dependent global TFIM. The CP-method operates by applying individual gates to groups of qubits, where certain gates are inhibited by the Rydberg blockade. The purpose of step 1 is to allow for independent control of the groups of qubits. Step 2 derives a gate sequence according to the CP-method. Given the choices made in steps 1 and 2, we can take the gate sequence given by the CP-method and convert it into a sequence of waiting times $\tau_i$ between the $\Gamma$-pulses, which constitutes the schedule $\Gamma_t$. Executing this schedule, i.e. by applying the $\Gamma$-pulses and waiting times $\tau_i$ to an Ising Hamiltonian $H_Z$, the circuit can be simulated.\\
A central challenge of using the global drive instead of the bipartite drive from the CP-method, is that neighboring qubits get excited, resulting in erroneous blockade interactions. The solution is to operate in the limit of a weak field, i.e. $\Gamma^* \ll 1$. Then, it is possible to choose the phases between two consecutive $\Gamma$-pulses such that the qubits that are not meant to be driven are immediately de-excited, while the driven qubits experience a non-diagonal rotation. This is illustrated in Figure~\ref{fig:GlobalDriveIssue}. Since the qubits are now driven only weakly, it requires many more $\Gamma$-pulses to generate any arbitrary gate. However, as we show, the error scales with $\mathcal{O}(\Gamma^{*2})$, while the number of required pulses is $\mathcal{O}\left( \Gamma^{*-1} \right)$, resulting in a net suppression of the error at the cost of more evolution time.\\
It is worth highlighting that while we assume identical $\Gamma_t$ for each qubit type, this requirement is not strict. In fact, even if the global transverse field acts qubit-type dependent, i.e. each type $\mathcal{K}$ is driven by $\Gamma_{\mathcal{K},t}$ with $\max_{t\in[0, \tau]}|\Gamma_{\mathcal{K}, t}| = \Theta(\Gamma^*)$, the reduction still applies. These differences can arise in systems where the distinct qubit types experience the global transverse field slightly differently, e.g. due to crosstalk in superconducting flux qubits. Therefore, there is inherent robustness to control inconsistencies.

%% file: 05_numerical_validation.tex
\section{Numerical validation of key step}
\label{sec:NumericalValidation}
As a simple validation of one of the central steps in our proof, we simulate the key step of the mapping numerically. We consider the movement of logical information along the wire for a single step, as depicted in Figure~\ref{fig:BasicScheme}(a). We simulate a wire of length $L$ initialized in the state $\ket{0101000...}$ and test the propagation sequence to move the logical information, i.e. the phase interface, one site to the right. We apply a single propagation cycle compiled into many weak $\Gamma$-pulses and measure the fidelity $\mathcal{F}$ with the target state $\ket{1010100...}$. We run the sequence for various values of the blockade coupling $J$, for decreasing $\Gamma^{*}$ and for several system sizes $L$. Since in the simple wire there are only two types of qubits, we use the local fields $h_{\mathcal{K}} = \alpha, \alpha^2$, where $\alpha$ is given as the solution to $\alpha^5 - \alpha -1 = 0$. With these choices, we note that a single propagation cycle requires $\mathcal{O}(10^4)$~or even $\mathcal{O}( 10^5)$~$\Gamma$-pulses, which highlights the issues for practical implementation of our method.

We simplify the simulation by omitting the approximation of the phase gates using Theorem~\ref{theorem:FillingTime}, as this step would prolong the numerical simulations significantly and is also based on a well-established mathematical result. Therefore, we compile a sequence of $\Gamma$-pulses, as described in Appendix~\ref{sec:Appendix:UniversalSU2}, except that the phase gates are applied to the respective qubit type exactly, i.e. $\varepsilon_{\text{phase}} = 0$. Then, the central lemma (Lemma~\ref{lemma:FullConditionalGateWithError}) bounds the error due to the compilation into $\Gamma$-pulses by $\mathcal{O}\left( L^2 \Gamma^* + L |J|^{-1} \right)$. The error in Lemma~\ref{lemma:FullConditionalGateWithError} is given by the operator distance $\|U_{target} \ket{\psi_0} - U_T \ket{\psi_0} \|$ for a $\ket{\psi_0} \in \mathcal{P}$, which translates into a decay of the infidelity $1-\mathcal{F} = \mathcal{O}\left( L^4 \Gamma^{*2} \right)$, since for close unitaries, the infidelity scales with the square of the distance. The predicted decay of the infidelity as a function of $\Gamma^*$ can be seen in Figure~\ref{fig:NumericalResult}(a) for several values of $J$. As a guide to the eye, we show as a black dashed line the worst-case scaling from our bound. The numerical results match the predicted scaling closely.

In Figure~\ref{fig:NumericalResult}(b) we show the dependence of the infidelity on the chain length $L$, again for various values of $J$. The Lemma~\ref{lemma:FullConditionalGateWithError} predicts a quadratic dependence on $L$ as the worst case for the distance in operator norm, which translates into a worst case infidelity of $\mathcal{O}(L^4)$, which is shown as a black dashed line for comparison. Although we can only test small system sizes, we observe a milder dependence on $L$ than predicted. In fact, we observe a significantly sub-quartic scaling in simulations with large $J$. Therefore, one of the key steps in our proof, Lemma~\ref{lemma:FullConditionalGateWithError}, is supported by the numerical evidence presented here. All simulations were implemented using QuTiP~\cite{qutip5}.

\begin{figure}
    \centering
    \begin{tikzpicture}
        \node[anchor=north west, inner sep=0] (fig) at (0,0)
            {\includegraphics[width=0.9875\linewidth]{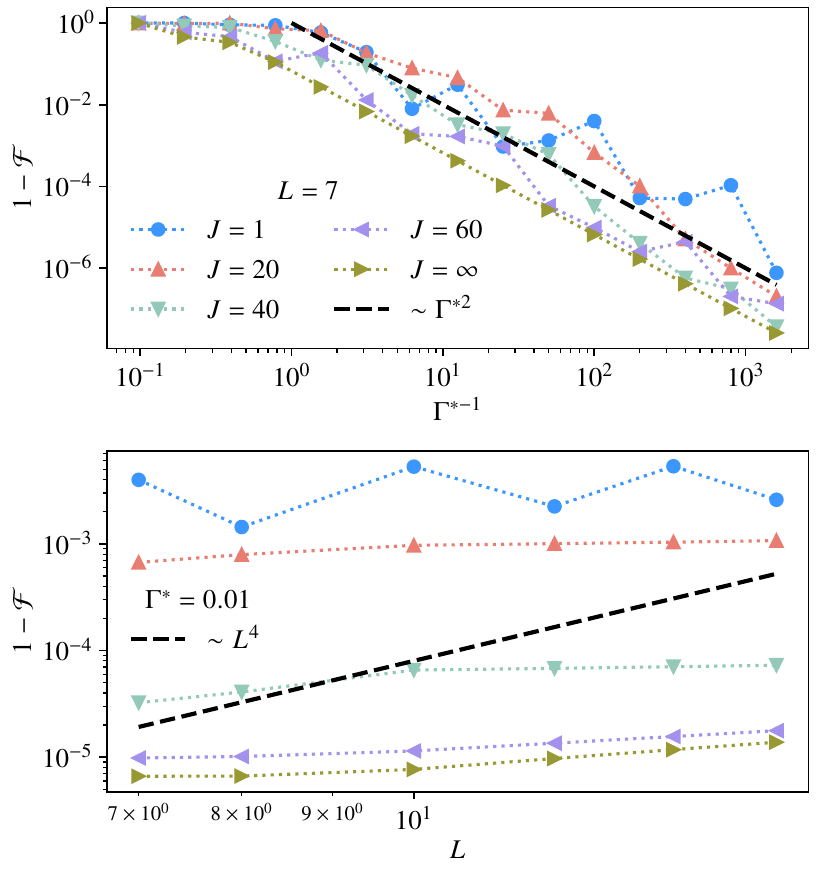}};
        \node[font=\large] at ($(fig.north west) + (fig.north west)!0.035!(fig.north east)$) {(a)};
        \node[font=\large] at ($(fig.north west)!0.5!(fig.south west) + (fig.north west)!0.035!(fig.north east)$) {(b)};
    \end{tikzpicture}
    \caption{Infidelity $1-\mathcal{F}$ of a single propagation cycle as shown in Figure~\ref{fig:BasicScheme}(a) compiled in global $\Gamma$-pulses as a function of the inverse maximal transverse field $\Gamma^{*-1}$~(a) and as a function of the system size~(b). The data in~(a) is generated with a fixed $L=7$, while in~(b) we fix the maximal transverse field at $\Gamma^* = 0.01$. The colors and markers indicate the same $J$ in both figures. $J=\infty$ corresponds to a simulation without blockade Hamiltonian $H_J$, but instead projected onto the computational subspace $\mathcal{P}$. We can see a clear quadratic decay in~(a) consistent with the theoretical prediction, while in~(b) the observed scaling is much slower than the predicted worst case. The dashed black line indicates the predicted worst-case scaling $\sim \Gamma^{*2}$ and $\sim L^4$ in (a) and (b), respectively, as a guide to the eye.}
    \label{fig:NumericalResult}
\end{figure}

%% file: 06_conclusion.tex
\section{Conclusion \& outlook}
\label{sec:Conclusion}
We have proposed a constructive method to simulate arbitrary quantum circuits on the globally driven transverse-field Ising model. The required time, physical qubits and coupling strength scale polynomially with the qubit and gate count of the circuit to be simulated, as well as the acceptable error. This establishes the polynomial equivalence between the globally driven transverse-field Ising model and the gate model of quantum computation. From this it follows directly that the Ising model with a global time-dependent transverse field is universal in a complexity-theoretic sense. We expect our result to be of interest for various communities, spanning experimental quantum computing and simulation on various platforms, classical and quantum complexity theory, quantum control theory and many-body physics. Furthermore, our results are relevant for quantum hardware design, where recent works attempt to minimize the number of control lines in order to boost scalability of the design~\cite{Planckian_2025, Planckian_2026a, Planckian_2026b}. Our result shows that the number of time-dependent control fields can be reduced to one, compared to methods requiring multiple control fields.\\
The exponents of our mapping are impractically large for current quantum hardware. However, it seems highly likely that more sophisticated proof techniques can yield tighter bounds. Our result can also serve as a conditional hardness theorem, since it implies that the globally driven transverse-field Ising model cannot be simulated efficiently by a classical computer, under the widely-held complexity-theoretic assumption that quantum advantage exists.\\
Our work motivates several directions for further research. First, since the techniques we use here to establish the bounds rely heavily on the submultiplicativity of norms and triangle inequalities, which are notoriously loose in many cases, it seems plausible that tighter bounds can be found, reducing the exponents of our polynomial equivalence claim. Furthermore, we employ the Schrieffer-Wolff transformation globally, which does not exploit the local structure of the Hamiltonian we consider here. It is likely that a local Schrieffer-Wolff transformation~\cite{Bravyi_2011} can tighten the bounds.\\
Second, the method proposed by Cesa and Pichler has been explored in other works on universal quantum computation with global control~\cite{Planckian_2025, Planckian_2026a, Planckian_2026b}, some of which reduce the number of physical qubits to $n_q' = \mathcal{O}(n_q)$, which is enabled by an additional type of qubit or alternative strategies to initialize the system. This suggests that the exponents in our method could be reduced by exploring other layouts and initialization schemes. Furthermore, the main contribution to our time overhead stems from phase wrapping of the qubit rotation under their local longitudinal fields, which we use to gain independent control of the four groups of qubits. In the CP-method, there are only two groups and their superatoms, which are characterized by their distinct Rabi frequency. Potentially, one can address the distinct qubit groups in our method by similar arguments on the Rabi frequency, at the cost of making more restrictive assumptions on the properties of the $\Gamma$-pulses used to drive the qubits. This could drastically decrease the exponent of the time required to obtain a sufficiently small $\varepsilon_{\text{phase}}$ and thus lead to a significant reduction in time overhead. Note that here we only make very loose assumptions about the driving pulses, such that our result is quite general and applies to a large class of devices with the ability to realize non-monotonic schedules.\\
Additionally, as has already been discussed by Cesa and Pichler, the CP-method can be augmented by error mitigation techniques such as mid-circuit measurements. In particular, since all the qubits that are not carrying the logical quantum many-body state are supposed to be in a known classical product state at certain steps of the computation, it is possible to actively set them to their target state after measuring them. It will be instructive to investigate how to optimally combine the proposed method with error mitigation and correction techniques.\\
Furthermore, it will be interesting to introduce further restrictions to analyze other related models of computation. Here we rely on the ability to perform non-monotonic schedules $\Gamma_t$, which is a feature not necessarily available on all quantum annealing hardware, or in all modes of operation of existing devices. It will be illuminating to examine constructions that rely only on monotonic schedules, thus performing quantum annealing in its original sense.

%% file: 07_acknowledgment.tex
\section{Acknowledgments}
I thank Artur García-Sáez, Bruno Julia-Díaz, Marta P. Estarellas and Arnau Riera for many helpful discussions. The author gratefully acknowledges RES resources provided by Barcelona Supercomputing Center in MareNostrum 5 to INNO-2026-1-0004. 

%% file: 08_appendix_a.tex
\section{Definitions and notation}
Before we begin with the proof, let us first introduce the relevant concepts and notation. As a general remark, here we consider systems composed of $n_q$ qubits. Hence, we consider operators acting on the Hilbert space $\mathcal{H} = (\mathbb{C}^2)^{\otimes n_q}$.

\subsection{$\mathcal{O}$-, $\Theta$- and $\Omega$-notation}
We will make use of $\mathcal{O}$-, $\Theta$- and $\Omega$-notation to describe asymptotic behavior in the conventional sense. However, attention must be paid as to which limit is considered. We make use of two cases. Consider two functions $f(x)$ and $g(x)$. We say that $f(x) = \mathcal{O}(g(x))$ for large $x$, or as $x \rightarrow \infty$, iff there exist a $x_0 > 0$ and a constant $C > 0$ such that $\forall x > x_0: |f(x)| \leq C |g(x)|$. Similarly, we say that $f(x) = \mathcal{O}(g(x))$ for small $x$, or as $x \rightarrow 0$, iff there exist a $x_0 > 0$ and a constant $C>0$ such that $\forall |x| < |x_0|: |f(x)| \leq C |g(x)|$. \\
Analogously, we define $\Theta$ in the same limits. We say that $f(x) = \Theta(g(x))$ for large $x$ iff $f(x) = \mathcal{O}(g(x))$ and $g(x) = \mathcal{O}(f(x))$ for large $x$, and $f(x) = \Theta(g(x))$ for small $x$ iff $f(x) = \mathcal{O}(g(x))$ and $g(x) = \mathcal{O}(f(x))$ for small $x$. We use the same terminology to denote asymptotic lower bounds. By the same logic as before, we say $f(x) = \Omega(g(x))$ for large $x$ (small $x$) iff there exist a $x_0 > 0$ and a constant $C>0$ such that $\forall x > x_0$ ($\forall |x| < x_0$) we find $|f(x)| \geq C|g(x)|$.\\
Sometimes, it is convenient to write $A(x) = \mathcal{O}(f(x))$, or $A(x) = \Theta(f(x))$, where $A(x)$ is a matrix-valued function. In these cases, the asymptotic limits are considered with respect to the operator norm, i.e. $A(x) = \mathcal{O}(f(x))$ for large / small $x$ iff $\|A(x) \| = \mathcal{O}(f(x))$ for large / small $x$ and analogously for $A(x) = \Theta(f(x))$. In particular, we will often equate $A = B + \mathcal{O}(f)$ with $\| A - B \| = \mathcal{O}(f)$.

\subsection{Computational subspace}
\label{sec:ComputationalSubspaceDefinition}
The CP-method was conceived with Rydberg atom systems in mind. Here, atoms / qubits can be placed closely together such that the Rydberg blockade energetically prohibits two adjacent qubits to be in the excited state at the same time. We will make use of the same mechanism. For a given connectivity graph $\mathcal{G}$ of qubits, these interactions are represented by the Hamiltonian
\begin{equation}
    H_J = \frac{J}{4}\sum_{(m,n) \in E(\mathcal{G})} (1-Z_m)(1-Z_n)
\end{equation}
for large $|J|$. In the setting of Rydberg atoms, $J$ is a positive number depending on the distance of the atoms. However, the error bounds we consider here depend on $|J|$, so in principle negative $J$ are feasible as well. In the limit of $|J| \rightarrow \infty$, $H_J$ energetically forbids transitions between the subspace where all blockade constraints are respected, which we shall denote by $\mathcal{P}$, and the orthogonal complement $\mathcal{Q}$, where at least two neighboring qubits are in the excited state. Hence, if a state is initialized in $\mathcal{P}$, the Hamiltonian $H_J$ constrains the dynamics to $\mathcal{P}$. These constraints are represented by the projector
\begin{equation}
    P = \prod_{(m,n) \in E(\mathcal{G})} \left( 1 - \frac{1}{4} (1-Z_m)(1-Z_n) \right) \ .
\end{equation}
Note that $P$ is diagonal in the computational basis and, therefore, commutes with diagonal operators. The kernel of $P$ is spanned by the computational basis states with at least two adjacent qubits in the state $\ket{1}$. Adjacency here is defined by the graph $\mathcal{G}$. $P$ projects onto the computational subspace $\mathcal{P}$, i.e. the subspace spanned by all computational basis states where no two neighboring qubits are in the excited state. Analogously, we define the projector $Q=\mathbb{I} - P$ onto the subspace $\mathcal{Q}$.\\
Furthermore, note that $P$ is the product of projectors $P_{mn} = 1- \frac{1}{4}(1-Z_m)(1-Z_n)$, which project onto the subspace where qubits $m$ and $n$ are not simultaneously excited. In other words, $P_{mn}$ enforces the Rydberg blockade on the edge $(m,n)$ in the connectivity graph. Similarly, we can define the projectors
\begin{equation}
    P_m = \prod_{n\in \mathcal{N}(m)} P_{mn} \ ,
\end{equation}
where $\mathcal{N}(m)$ is the set of neighbors of $m$ in $\mathcal{G}$. $P_m$ can then be interpreted as enforcing the Rydberg blockade in the neighborhood of qubit $m$. Since $P_{mn}^2 = P_{mn}$ and $P_{mn} = P_{nm}$, it follows that
\begin{equation}
    P = \prod_m P_m \ .
\end{equation}
In our proof below we will first prove the main result in the limit of $|J|\rightarrow \infty$, so that our arguments can make use of the projector $P$. In a second step, we will analyze and bound the error due to finite $|J|$, using the dynamics generated by $H_J$. In the limit of infinite $|J|$, single-qubit gates are applied conditionally, depending on the state of neighboring qubits. We will now define these conditional gates properly.

\subsection{Conditional gates and generators}
\label{sec:ConditionalGatesDefinition}
Let $V \in \text{SU(2)}$ be a single-qubit quantum gate with a potentially time-dependent generator $iH_{V,t}$ such that
\begin{equation}
    V = \mathcal{T} \exp\left(-i\int_0^\tau H_{V,t} dt \right) \ .
\end{equation}
In this work, we will often consider unitaries generated by Hamiltonians projected onto the computational subspace $\mathcal{P}$ by the projector $P$. This leads to conditional single-qubit gates, conditioned on the state of their neighboring qubits in the graph. We will denote these conditional gates by the superscript $V_n^{(c)}$, while the subscript indicates that the gate acts on qubit $n$ conditioned on the state of its neighbors, i.e.
\begin{equation}
    V_n^{(c)} := \mathcal{T} \exp \left( -i \int_0^\tau P H_{V_n,t} P dt\right) \ .
\end{equation}
From the definition, it is clear that $[V^{(c)}, P] = 0$, since $[P, PH_{V,t}P] = 0$. While the CP-method deals with conditional gates $V^{(c)}$, we will derive some results for the unconditional version $V$ and then explicitly show that their validity holds in the conditional case as well. Note one special case when $[H_{V,t}, P]=0$. Then
\begin{equation}
    \begin{aligned}
        V^{(c)} &= \mathcal{T} \exp \left( -i \int_0^\tau P H_{V,t} dt\right) \\
        &= \mathcal{T} \exp \left( -i \int_0^\tau H_{V,t} P dt\right) \\
        &= PV = VP \ .
    \end{aligned}
\end{equation}
From this it follows easily for two conditional gates $V^{(c)}$ and $W^{(c)}$, where $[H_{V,t}, P] = 0$, that $V^{(c)} W^{(c)} = V W^{(c)}$ and $W^{(c)} V^{(c)} = W^{(c)} V$. This is especially true if $V^{(c)}$ is a phase gate. Furthermore, it will be convenient to write the single-qubit gates as $2\times2$-matrices
\begin{equation}
    V^{(c)} = \begin{pmatrix}
        v_{00} & v_{01} \\
        v_{10} & v_{11}
    \end{pmatrix}^{(c)} \ ,
\end{equation}
where we apply the superscript $(c)$ to indicate that the gate is applied conditionally.\\
It is worth highlighting one additional detail about the projected generators. The projector $P$ is a global operator that acts on all qubits. However, restricted to $\mathcal{P}$, the projector acts effectively locally, in the sense that only the state of the neighbors of qubit $n$ is relevant. Consider the Hamiltonian $H_t = \sum_n H_{n,t}$, where $H_n$ acts on qubit $n$. Then $H_t$ generates the unitary
\begin{equation}
    U = \mathcal{T}\exp \left( -i\int_0^\tau \sum_n H_{n,t} dt \right) \ .
\end{equation}
Then considering the series expansion of $U$, it follows that
\begin{equation}
    \begin{aligned}
        PU^{(c)}P &= P \mathcal{T}\exp \left( -i\int_0^\tau PH_tP dt \right) P \\
        &= P - i \int_0^\tau \sum_n P H_{n,t} P dt+ (-i)^2 \int_0^\tau \int_0^t \sum_{m,n} P H_{m,t} P H_{n,t} P dt' dt + ... \\
        &= P + ... + (-i)^k \int_{0\leq t' \leq t'' \leq ... \leq t^{(k)}\leq \tau} \sum_{n_1, ..., n_k} P H_{n_k, t^{(k)}} P H_{n_{k-1}, t^{(k-1)}} P ... P H_{n_1,t'} P dt' ... dt^{(k)}
    \end{aligned}
\end{equation}
Since $P_m$ commutes with $H_{n,t}$ if $m$ and $n$ are non-adjacent, the different edge projectors commute also with the respective Hamiltonian terms. In total, it suffices to project $H_n$ by the $P_{mn}$ for incident edges, while all other edge projectors can be moved to the boundary of the factor $H_{n_1} H_{n_2} ... H_{n_k}$. Therefore, we can replace $PH_{n,t}P = P_n H_{n,t} P_n$ and we find
\begin{equation}
    \begin{aligned}
        &PU^{(c)} P = P + ... \\
        &+ (-i)^k \int_{0\leq t' \leq t'' \leq ... \leq t^{(k)}\leq \tau} P \left( \sum_{n_1, ..., n_k} P_{n_k} H_{n_k, t^{(k)}} P_{n_k} P_{n_{k-1}} H_{n_{k-1}, t^{(k-1)}} P_{n_{k-1}} ... P_{n_1} H_{n_1,t'} P_{n_1} dt' ... dt^{(k)} \right) P \\
        &= P \mathcal{T} \exp\left(-i \int_0^\tau \sum_n P_n H_{n,t} P_n dt \right) P \ .
    \end{aligned}
\end{equation}
Therefore, projected onto $\mathcal{P}$, the global projector $P$ acts effectively locally only considering the neighborhood of a qubit $n$ in $\mathcal{G}$. Specifically, if $PU^{(c)}P$ only drives a subset $\mathcal{I}$ of qubits, then it acts on qubits adjacent to $\mathcal{I}$ via the diagonal projectors. At the same time, it acts trivially on qubits that are neither in $\mathcal{I}$, nor adjacent to it. Furthermore, if $\mathcal{I}$ is a subset of non-adjacent qubits, $U^{(c)}$ factorizes such that
\begin{equation}
    U^{(c)} = \prod_{n\in\mathcal{I}} \mathcal{T} \exp \left( -i\int_0^\tau PH_{n,t} P dt\right) \ ,
\end{equation}
since the Hamiltonians $H_{n,t}$ only act non-trivially on non-intersecting subsets of qubits.

\subsection{$\Gamma$-pulses and phase gates}
\label{sec:GammaPulseDefinition}
We will show below a decomposition of conditional quantum gates into a sequence of $\Gamma$-pulses, where the transverse field is turned on. When no pulse is applied, the system evolves under the diagonal Hamiltonian $H_Z$. The global $\Gamma$-pulses are generated by turning on the transverse field $\Gamma_t = \Gamma^* g_t$, where $\left| \int_0^\tau g_t e^{-i2h_\mathcal{K}t}dt \right| = \Theta(1)$ for all $\mathcal{K} \in \{ \mathcal{A}, \mathcal{B}, \mathcal{C}, \mathcal{D} \}$. This definition of $\Gamma_t$ is mainly for convenience, since it allows us to control the pulse strength via the scaling parameter $\Gamma^*$. We also require $ |g_{t_2} - g_{t_1} | \leq K |t_2 - t_1|$ for a Lipschitz constant $K=\mathcal{O}(1)$. Although not strictly necessary, it is furthermore convenient to assume $g_0 = g_\tau = 0$. Note that these requirements are fairly broad and that a large class of functions satisfies these conditions, so our results are valid for a large number of devices. Note that since our method will operate in the limit of small $\Gamma^*$, conditions on the Lipschitz constant like $|\Gamma_{t_2} - \Gamma_{t_1}| / |t_2 - t_1| \leq K$ are naturally met. While one could also allow for programmable $\tau$, here we will assume $\tau = \mathcal{O}(1)$ to be a constant of the device. This avoids the introduction of additional assumptions on device capabilities.\\
Consider the full time-dependent Hamiltonian $H_t = H_Z + \Gamma_t H_X$ as in Eq.~\eqref{eq:DefFullHamiltonian}, which over the time $T$ generates the time evolution
\begin{equation}
    U_T = \mathcal{T}\exp\left( -i \int_0^T H_t dt \right) \ .
\end{equation}
The full schedule will be a concatenation of $\Gamma$-pulses
\begin{equation}
    U_\Gamma = \mathcal{T} \exp \left(-i \int_0^\tau (H_Z + \Gamma_t H_X)dt \right) \ ,
\end{equation}
where $\Gamma_t$ changes as a function of time as described here, and waiting times, where the transverse field is turned off and the system evolves under the diagonal Hamiltonian $H_Z$ only, i.e.
\begin{equation}
    U_Z(\tau_i) = \exp \left(-i\tau_i H_Z \right) = \exp \left( -i \tau_i \left(\sum_m h_m Z_m + \frac{J}{4}\sum_{(m,n) \in E(\mathcal{G})} (1-Z_m) (1-Z_n)\right) \right) \ ,
\end{equation}
as depicted in Figure~\ref{fig:BasicScheme}(d). The waiting times also correspond to globally acting phase gates. We assume that all $\Gamma$-pulses are identical. Then, the total evolution $U_T$ decomposes into a sequence of $\Gamma$-pulses and phase gates
\begin{equation}
    U_T = U_Z(\tau_0) \prod_{i=1}^N U_\Gamma U_Z(\tau_i) \ .
\end{equation}
In the following, we will argue that for a sufficiently large $N$, $U_T$ can simulate any quantum  circuit.

%% file: 09_appendix_b.tex
\section{Technical lemmata}
\label{sec:AppendixAuxLemmas}
We will require the following results from the literature for the formal proof. Lemma~\ref{lemma:UniversalSU2} shows how to generate universal single-qubit gates using almost arbitrary pulses to stimulate transitions between the energy levels and control of the phases between the pulses. Results from ergodic theory (Theorem~\ref{theorem:FillingTime}) are then used to show how we can obtain quasi-independent control of the phases in distinct qubit types. Lemma~\ref{lemma:ExponentialSeriesBound} and the Gershgorin Circle Theorem (Theorem~\ref{theorem:Gershgorin}) will be used to establish bounds on operator norms, while Lemma~\ref{lemma:TimeEvolutionError} will allow us to bound the error of a time evolution under a perturbation.

\subsection{Universal SU(2)-gates from pulses and phases}
\label{sec:Appendix:UniversalSU2}
We will require universal single-qubit unitaries, which can be constructed from $\Gamma$-pulses and phase gates. To this end, we use Proposition 1 from~\cite{Werner_2026}. Here, we restate it as a Lemma and a proof for completeness, while slightly simplifying the assumptions.\\
Consider a single-qubit Hamiltonian $H_t$ and the respective unitary generated by $iH_t$. It is well-known that any SU(2) element can be decomposed as
\begin{equation}
    \mathcal{T} \exp \left( -i \int_0^\tau \Gamma_t H_t dt\right) = \begin{pmatrix}
        e^{-i\phi_0} & 0 \\
        0 & e^{i\phi_0}
    \end{pmatrix}\begin{pmatrix}
        \cos \lambda & \sin \lambda \\
        -\sin \lambda & \cos \lambda
    \end{pmatrix} \begin{pmatrix}
        e^{-i\theta_0} & 0 \\
        0 & e^{i\theta_0}
    \end{pmatrix}
\end{equation}
for angles $\phi_0$, $\theta_0$ and $\lambda$. Note that it is possible to assume $\lambda \in [0, \pi/2]$ for proper choices of $\phi_0$ and $\theta_0$. If we had sufficient control over all three angles, this would already define a universal SU(2)-gate. However, we will require a method to generate universal SU(2)-gates even if $\lambda$ is restricted. Introducing a phase gate before and after the $\Gamma$-pulse allows us to define the gate
\begin{equation}
    U(\phi, \theta) = \begin{pmatrix}
        e^{i(\phi-\phi_0)} & 0 \\
        0 & e^{-i(\phi-\phi_0)}
    \end{pmatrix}\begin{pmatrix}
        \cos \lambda & \sin \lambda \\
        -\sin \lambda & \cos \lambda
    \end{pmatrix} \begin{pmatrix}
        e^{i(\theta - \theta_0)} & 0 \\
        0 & e^{-i(\theta - \theta_0)}
    \end{pmatrix}\ .
    \label{eq:WaitPulseWaitGate}
\end{equation}
In~\cite{Werner_2026}, it is shown for an arbitrary fixed $\lambda$, but since for the result here we require the limit of $\sin^2 \lambda \ll 1$, it simplifies the proof to consider $0 < \lambda < \pi/4$. The phases $\phi$ and $\theta$ will be the parameters by which we control the gates using the following Lemma.\\
Since the phases in the gate Eq.~\eqref{eq:WaitPulseWaitGate} are freely tunable parameters, we will several times absorb phase offsets into the tunable parameter to simplify notation. E.g. we would replace $\phi - \phi_0$ simply with $\phi$ and have it understood that a choice of the phase $\phi$ would require replacement with the appropriate offset value. In this spirit, let us for the proof consider the gate
\begin{equation}
    U(\phi, \theta) = \begin{pmatrix}
        e^{i\phi} & 0 \\
        0 & e^{-i\phi}
    \end{pmatrix}\begin{pmatrix}
        \cos \lambda & \sin \lambda \\
        -\sin \lambda & \cos \lambda
    \end{pmatrix} \begin{pmatrix}
        e^{i\theta} & 0 \\
        0 & e^{-i\theta}
    \end{pmatrix}\ .
    \label{eq:WaitPulseWaitSimplified}
\end{equation}
\begin{lemma}
    \label{lemma:UniversalSU2} (Proposition 1 in Ref.~\cite{Werner_2026})
    Let $U(\phi, \theta)$ as in Eq.~\eqref{eq:WaitPulseWaitSimplified} with $0 < \lambda \leq \pi/ 4$. For any $V \in \text{SU}(2)$, there is a $N_0 \in \mathbb{N}$ with $N_0 = 2^{\lceil \log_2 \frac{\pi}{2 \lambda} \rceil}$, such that for any $N\geq N_0$ there is a sequence of angles $(\phi_n, \theta_n)$ with
    \begin{equation}
        \prod_{n=1}^N U(\phi_n, \theta_n) = V \ .
        \label{eq:UniversalGateSequence}
    \end{equation}
\end{lemma}
\begin{proof}
    The proof is rather straightforward, using trigonometric identities and the standard decomposition of any element in $V \in \text{SU}(2)$ as
    \begin{equation}
        V = \exp(i\phi_V Z) \exp(i\lambda_V Y) \exp(i \theta_V Z) \ .
    \end{equation}
    The phases are simply given by the first and last phase gates in the sequence Eq.~\eqref{eq:UniversalGateSequence}. What is missing is to show that any $Y$-rotation can be decomposed into a sequence of phase gates and unitaries with fixed transition probability $\sin^2 \lambda$. This can be shown with an iterative argument.\\
    Consider the concatenation of two gates Eq.~\eqref{eq:WaitPulseWaitSimplified}
    \begin{equation}
        U(\phi_1, \theta_1) U(\phi_2, \theta_2) = \begin{pmatrix}
        e^{i\phi_1} & 0 \\
        0 & e^{-i\phi_1}
    \end{pmatrix}\begin{pmatrix}
        \cos \lambda & \sin \lambda \\
        -\sin \lambda & \cos \lambda
    \end{pmatrix} \begin{pmatrix}
        e^{i\phi'} & 0 \\
        0 & e^{-i\phi'}
    \end{pmatrix} \begin{pmatrix}
        \cos \lambda & \sin \lambda \\
        -\sin \lambda & \cos \lambda
    \end{pmatrix} \begin{pmatrix}
        e^{i\theta_2} & 0 \\
        0 & e^{-i\theta_2}
    \end{pmatrix} \ ,
    \end{equation}
    where we have combined the two consecutive phase gates to a single one with $\phi' := \theta_1 + \phi_2$. Evaluating the matrix product of the inner three matrices, we get
    \begin{equation}
    \begin{aligned}
        U(\phi_1, \theta_1) U(\phi_2, \theta_2) &= \begin{pmatrix}
        e^{i\phi_1} & 0 \\
        0 & e^{-i\phi_1}
    \end{pmatrix}
    \begin{pmatrix}
        e^{i\phi'} \cos^2 \lambda - e^{-i\phi'} \sin^2 \lambda & \cos \lambda \sin \lambda (e^{i\phi'} + e^{-i\phi'}) \\
        -\cos \lambda \sin \lambda (e^{i\phi'} + e^{-i\phi'}) & e^{-i\phi'} \cos^2 \lambda - e^{i\phi'} \sin^2 \lambda
    \end{pmatrix}
    \begin{pmatrix}
        e^{i\theta_2} & 0 \\
        0 & e^{-i\theta_2}
    \end{pmatrix} \\
    &= \begin{pmatrix}
        e^{i\phi_1} & 0 \\
        0 & e^{-i\phi_1}
    \end{pmatrix}
    \begin{pmatrix}
        \cos(2\lambda) \cos(\phi') + i \sin \phi' & \sin (2\lambda) \cos \phi' \\
        -\sin (2\lambda) \cos \phi' & \cos(2\lambda) \cos(\phi') - i \sin\phi'
    \end{pmatrix}
    \begin{pmatrix}
        e^{i\theta_2} & 0 \\
        0 & e^{-i\theta_2}
    \end{pmatrix} \ ,
    \end{aligned}
    \end{equation}
    where we used the double-angle identities $\sin2\lambda = 2 \cos \lambda \sin \lambda$ and $2\cos^2 \lambda - 1 = 1 - 2\sin^2\lambda = \cos 2 \lambda$. Using the fact that any element of $U_{\text{real}} \in \text{SU}(2)$ with real off-diagonal elements can be decomposed as
    \begin{equation}
        U_{\text{real}} = \begin{pmatrix}
            e^{i\alpha} & 0 \\ 0 & e^{-i\alpha}
        \end{pmatrix}
        \begin{pmatrix}
            \cos \lambda' & \sin \lambda' \\
            -\sin \lambda' & \cos \lambda'
        \end{pmatrix}
        \begin{pmatrix}
            e^{i\alpha} & 0 \\ 0 & e^{-i\alpha}
        \end{pmatrix}
    \end{equation}
    this defines a new gate
    \begin{equation}
        \begin{aligned}
            U'(\phi_1, \theta_2, \lambda') :=& U(\phi_1, \theta_1)U(\phi_2, \theta_2) \\
            =& \begin{pmatrix}
                e^{i(\phi_1 + \alpha)} & 0 \\
                0 & e^{-i(\phi_1 + \alpha)}
            \end{pmatrix}
            \begin{pmatrix}
                \cos \lambda' & \sin \lambda' \\
                -\sin \lambda' & \cos \lambda'
            \end{pmatrix}
            \begin{pmatrix}
                e^{i(\theta_2 + \alpha)} & 0 \\
                0 & e^{-i(\theta_2 + \alpha)}
            \end{pmatrix} \ ,
        \end{aligned}
        \label{eq:Uprime}
    \end{equation}
    where $\lambda' = \lambda'(\theta_1, \phi_2)$ is now a tunable parameter via the phase $\phi' = \theta_1 + \phi_2$ such that $0 \leq |\sin \lambda'| \leq |\sin(2\lambda)|$. Note that if $\lambda = \pi/4$, we can choose $\lambda' = \pi / 2$ and therefore any $V\in \text{SU}(2)$ can be decomposed as Eq.~\eqref{eq:Uprime} for appropriate choices of $\phi_1$, $\theta_2$ and $\lambda'$, while for $0 < \lambda < \pi/4$ we have effectively amplified the accessible angle of the $Y$-rotation.\\
    Since by assumption $\lambda \leq \pi/4$, we can distinguish the case $\lambda = \pi/4$, in which case $U'$ is universal, and the case $\lambda < \pi/4$, in which case we set $\lambda' = \min(\pi/4, 2\lambda)$. Subsequently, we concatenate two gates of $U'$, choosing the phases between the two $U'$ such that we obtain the gate $U''$ with $0 \leq |\sin \lambda''| \leq |\sin(2\lambda')| $. After $n$ iterations of this argument, we achieve $U^{(n)} = U^{(n-1)} U^{(n-1)}$ with
    \begin{equation}
        \sin \lambda^{(n)} = \sin \left( \lambda \prod_{k=1}^n \kappa^{(k)}\right) \ ,
    \end{equation}
    with $0\leq \kappa^{(k)} \leq 2$ the accessible ranges for all $k$, which are chosen such that $0\leq \lambda \prod_{k=1}^n\kappa^{(k)} \leq \pi/4$. The effective angle doubles at every iteration and it is clear that the upper bound can be saturated only after at least $n = \lceil \log_2\left(\frac{\pi}{4\lambda} \right) \rceil$. Choosing $\lambda^{(n)} = \pi/4$, the next iteration yields a freely tunable $\lambda^{(n+1)} \in [0, \pi/2]$ and thus, universality is reached, together with the freely tunable phase gates before and after the sequence. Note that the sequence length $N$ doubles at every iteration, and therefore universality is reached after $N_0 = 2^{\lceil \log_2\left(\frac{\pi}{2\lambda} \right) \rceil}$ gates. Since the $\kappa^{(k)}$ can be chosen to be small, the universality also holds for $N>N_0$. Thus, the claim of the lemma follows.
\end{proof}
Note that for $0< \lambda \ll \pi /4$, $N_0 = \Theta \left(\lambda^{-1} \right) = \Theta \left(1 / \sin\lambda \right)$. Lemma~\ref{lemma:UniversalSU2} tells us that we can decompose any single-qubit gate into a series of pulses with phase gates between them. Therefore, in order to address the distinct qubit groups independently, we require independent phase control, at least up to a small error.

\subsection{Quasi-independent phase control}
Lloyd proposed to gain quasi-independent control over different Hamiltonian terms by choosing four $\mathbb{Q}$-linearly independent coefficients, such that each Hamiltonian could effectively be singled out up to an error $\varepsilon_{\text{phase}}$ by evolving the system for a time $\mathcal{O}\left( \varepsilon_{\text{phase}}^{-4}\right)$~\cite{Lloyd_2018}. Here, we will make this notion more concrete by invoking results from ergodic theory. This will also yield an improved exponent. First, we will require a basic number theoretical result derived from the Schmidt subspace theorem.
\begin{theorem}(Theorem 7.3.2 from Chapter 7 in Ref.~\cite{Bombieri_Gubler_2006})
    Let $\Omega_1, ..., \Omega_n$ be algebraic numbers. Then for every $\nu > 0$ the inequality
    \begin{equation}
        0 < |\Omega_1 k_1 + ... + \Omega_n k_n| \leq  \| k \|_\infty^{1-n-\nu}
        \label{eq:SchmidtInequality}
    \end{equation}
    has only finitely many solutions $(k_1, ..., k_n) \in \mathbb{Z}^n$.
    \label{theorem:Schmidt}
\end{theorem}
Let us denote the vector of algebraic numbers with $\Omega:=(\Omega_1, ..., \Omega_n)$. We will use Theorem~\ref{theorem:Schmidt} to check for membership of a vector $\Omega$ to the so-called Diophantine set~\cite{Bourgain_1998, Dumas_2021}, which is defined for $\nu > 0$ and $\gamma > 0$ as
\begin{equation}
    D_n(\nu, \gamma) := \left\{ \Omega \in \mathbb{R}^n : |k \cdot \Omega| \geq \gamma \| k \|^{1-n-\nu} \text{ for any } k \in \mathbb{Z}^n \setminus \{ 0 \} \right\} \ .
\end{equation}
For a given vector $\Omega \in \mathbb{R}^n$ we can define a trajectory on the torus $\mathbb{R}^n / \mathbb{Z}^n$ as
\begin{equation}
    \omega_t(A) := A + \Omega t \mod \mathbb{Z}^n
\end{equation}
for an initial point $A \in \mathbb{R}^n$. Vectors of algebraic numbers in $D_n(\nu, \gamma)$ have a very useful property, in that their filling times can be bounded. The filling time is defined as the shortest time that is required for a linear flow to $\varepsilon_{\text{phase}}$-approximate any point on a torus, i.e.
\begin{equation}
    \tau_{\text{fill}}(\varepsilon_{\text{phase}}) = \inf \{ t \in \mathbb{R}_+ :  \forall x \in \mathbb{R}^n / \mathbb{Z}^n , d(x, \omega_t) \leq \varepsilon_{\text{phase}}\} \ ,
\end{equation}
where $d(\cdot, \cdot)$ denotes the distance. The question of estimating filling times is extensively studied~\cite{Bourgain_1998, Dumas_2021}, and we will use the following theorem, which is optimal with respect to the exponent.
\begin{theorem} (Theorem D in~\cite{Bourgain_1998} or Remark 4.1 in~\cite{Berti_2003})
    For all $n\in \mathbb{N} \setminus \{ 0 \}$, there is a constant $C(n, \nu)$ such that for all $\nu>0$ and all $\Omega \in D_n(\nu, \gamma)$ we have
    \begin{equation}
        \tau_{fill}(\varepsilon_{\text{phase}}) \leq C(n, \nu) \varepsilon_{\text{phase}}^{1-n-\nu} \ .
        \label{eq:EffectiveBound}
    \end{equation}
    \label{theorem:FillingTime}
\end{theorem}
Note that Theorem~\ref{theorem:FillingTime} only depends on the algebraic properties of $\Omega$, but is otherwise independent of the initial point $A$. It is clear that any vector $\Omega$ satisfying the conditions in Theorem~\ref{theorem:Schmidt} non-trivially is also contained in $D_n(\nu, \gamma)$ for some $\gamma$ and any $\nu > 0$. This follows, since according to Theorem~\ref{theorem:Schmidt}, there are only finitely many $k \in \mathbb{Z}^n$ that satisfy the inequality Eq.~\eqref{eq:SchmidtInequality}, so $k\cdot \Omega$ can be evaluated on all of them and $\gamma$ can be picked accordingly. Note that Theorem~\ref{theorem:Schmidt} and the definition of $D_n(\nu, \gamma)$ use different norms, which can be absorbed into the prefactor $\gamma$ according to the equivalence of norms in finite dimension.\\
We can use this result by choosing the local longitudinal fields $h_{\mathcal{K}}$ such that they fulfill the conditions of Theorem~\ref{theorem:Schmidt} non-trivially. This can, for example, be achieved by choosing $h_{\mathcal{K}} = \alpha, \alpha^2, \alpha^3, \alpha^4$, where $\alpha$ is the root of the irreducible polynomial $\alpha^5 - \alpha - 1 = 0$. Then, the vector of the $h_{\mathcal{K}}$ is in $D_4(\nu, \gamma)$ and we can apply Theorem~\ref{theorem:FillingTime}. It then follows that any point on the torus $\mathbb{R}^4 / \mathbb{Z}^4$ can be approximated up to an error $\varepsilon_{\text{phase}}$ by evolving for a time $\mathcal{O}\left( \varepsilon_{\text{phase}}^{-3-\nu} \right)$, for any $\nu >0$. Technically, we are interested in the torus $\mathbb{R}^4 / (2\pi \mathbb{Z}^4)$, which we can account for via a trivial re-scaling.\\
Assume a given set of phases $\phi_{\mathcal{K}}$ for all $\mathcal{K} \in \{ \mathcal{A}, \mathcal{B}, \mathcal{C}, \mathcal{D} \}$. Since the global phase gate acts on all qubit groups $\mathcal{K}$ for the same amount of time $\tau_i$, we need to solve the set of equations
\begin{equation}
    |h_\mathcal{K} \tilde{\tau}_i - \phi_\mathcal{K} / (2\pi)| \leq \varepsilon_{\text{phase}} / (2\pi) \mod 1
\end{equation}
for all $\mathcal{K}$. According to Theorem~\ref{theorem:FillingTime}, there is a $\tilde{\tau}_i$ which satisfies all four inequalities. Then, choosing $\tau_i = 2\pi \cdot \tilde{\tau}_i$ gives the required $\tau_i$ such that
\begin{equation}
    |h_{\mathcal{K}} \tau_i - \phi_\mathcal{K} | \leq \varepsilon_{\text{phase}} \mod 2\pi \ .
\end{equation}
Furthermore, Eq.~\eqref{eq:EffectiveBound} in Theorem~\ref{theorem:FillingTime} states that then $\tau_i = \mathcal{O}\left( \varepsilon_{\text{phase}}^{-3-\nu}\right)$.
Therefore, we gain quasi-independent control of the phases acting on each group of qubits at the cost of time polynomial in the inverse error $\varepsilon_{\text{phase}}^{-1}$. We can then use this phase control to realize any set of single-qubit gates acting on the groups simultaneously by concatenating sufficiently many pulses and phase gates. Furthermore, this sequence can be compiled efficiently, as we show in the following.

\subsection{Efficient compilation of single-qubit gates}
The proof of Lemma~\ref{lemma:UniversalSU2} directly implies an efficient classical compilation algorithm for arbitrary $V\in\text{SU(2)}$, assuming that $\lambda$ is sufficiently large, as we argue in the following. Recall that if $\sin \lambda = 1/\sqrt{2}$, i.e. $\lambda = \pi/4$, then a single concatenation step and a proper choice of phases renders $U(\phi_1, \theta_1)U(\phi_2, \theta_2) = V$. In the case of $0 < \lambda < \pi/4$, this implies that one could precompile a sequence of $N_0/2$ gates $U(\phi, \theta)$ to prepare a gate sequence with $\lambda = \pi/4$ and then simply concatenate two of these sequences with appropriate choices of phases before, after and between the two sequences to generate any gate $V$. Note that this procedure implies that in total only $\mathcal{O} \left( \lceil \log_2 \frac{\pi}{2\lambda} \rceil \right)$ angles need to be determined for any gate. Of all the angles, three angles are specific to $V$, while all others are generic. Furthermore, since only $\mathcal{O} \left( \lceil \log_2 \frac{\pi}{2\lambda} \rceil \right)$ angles need to be determined, we could potentially compile any $V$ even if $\lambda = \Theta(2^{-n_q})$. However, for efficient physical execution of the gate, we require that $\lambda$ does not vanish super-polynomially fast as a function of the system size to be simulated, i.e. $\lambda \geq \Omega(n_q^{-c})$ for some constant $c>0$.\\
In the context we consider here, we will need to compile gates for four types of qubits. The global $\Gamma$-pulse will act on all of them and stimulate a qubit transition for each type $\mathcal{K}$, i.e. $r_\mathcal{K} = \sin^2 \lambda_\mathcal{K}$. Since all qubits will experience the same number of $\Gamma$-pulses, we can consider sequences of length $N = \max_\mathcal{K} 2^{\lceil  \log_2 \frac{\pi}{2\lambda_{\mathcal{K}}} \rceil}$. For any collection of $V_\mathcal{K} \in \text{SU(2)}$, we can then determine the required angles for each $\mathcal{K}$, resulting in a set of $\mathcal{O}(\log N)$ angles. This can be done efficiently on a classical computer. For longitudinal fields $h_\mathcal{K}$ fulfilling Theorem~\ref{theorem:Schmidt}, it is then possible to find a global $\tau_i$ that results in phases that are $\varepsilon_{\text{phase}}$-close to the desired phases. It is easy to see that the solution $\tau_i$ can be found simply by brute-force search, discretizing $\tau_i$ into $\mathcal{O}\left( \varepsilon_{\text{phase}}^{-1} \right)$ steps and searching up to $\tau_i = \mathcal{O}\left( \varepsilon_{\text{phase}}^{-3-\nu} \right)$. Therefore, the compilation time is also polynomial in $\varepsilon_{\text{phase}}$. The compilation method presented here is sufficient for the polynomial equivalence we are proving in the following. However, it is very straightforward and more efficient solvers may exist. For example, one could generate an $\varepsilon_0$-grid of single-qubit gates via the method presented here for some small $\varepsilon_0 = \mathcal{O}(1)$. This can then be combined with the inverse-free Solovay-Kitaev algorithm~\cite{Bouland_2021} to reduce the error to any desired level $\varepsilon$ at the cost of multiplying the sequence length by a factor $\log^c (1/\varepsilon)$ for some constant $c>0$. We leave the detailed discussion for future research.

\subsection{Bound on exponential series}
We will make use of the following standard result of analysis several times. We present the proof here for completeness.
\begin{lemma}
    \label{lemma:ExponentialSeriesBound}
    Let $x \in \mathbb{R}$ with $x \geq 0$. Then
    \begin{equation}
        \sum_{k=2}^\infty \frac{x^k}{k!} \leq x^2 e^x \ .
    \end{equation}
\end{lemma}
\begin{proof}
    It is easy to see that the series converges point-wise, since it is the Taylor series of $e^x$ minus the first two terms. So we can algebraically manipulate the left-hand side such that
    \begin{equation}
        \begin{aligned}
            \sum_{k=2}^\infty \frac{x^k}{k!} = x^2 \sum_{k=2}^\infty \frac{x^{k-2}}{k!} = x^2 \sum_{k=0}^\infty \frac{x^k}{(k+2)!} \ .
        \end{aligned}
    \end{equation}
    Since $x\geq 0$, for each term we have $\frac{x^k}{(k+2)!} \leq \frac{x^k}{k!}$ and it follows that
    \begin{equation}
        x^2 \sum_{k=0}^\infty \frac{x^k}{(k+2)!} \leq x^2 \sum_{k=0}^\infty \frac{x^k}{k!} = x^2 e^x \ . 
    \end{equation}
\end{proof}
Note that for small $x$, Lemma~\ref{lemma:ExponentialSeriesBound} implies $\sum_{k=2}^\infty \frac{x^k}{k!} = \mathcal{O}(x^2)$.

\subsection{Perturbed time evolution}
\label{sec:AppendixTimeEvolutionError}
We will require a bound on the error of a time evolution under a perturbed Hamiltonian. To this end, we will use the following standard result.
\begin{lemma}
    Consider a time-dependent Hamiltonian $H_t' = H_t + V_t$, where $H_t$ is considered the true Hamiltonian and $V_t$ is a time-dependent error, both defined on a time interval $t \in [0, \tau]$ almost everywhere (a.e.). The respective time-evolution operators are given as $U_\tau = \mathcal{T}\exp \left( -i \int_0^\tau H_t dt\right)$ and $U_\tau' = \mathcal{T}\exp \left( -i \int_0^\tau H_t' dt\right)$. Then,
    \begin{equation}
        \begin{aligned}
            U_\tau' &= U_\tau - i \int_0^\tau U_{\tau, t} V_t U_{t, 0} dt \\
            &+ (-i)^2 \int_0^\tau \int_0^t U_{\tau, t}V_t U_{t, t'} V_{t'} U_{t', 0} dt' dt \\
            &+ ... \ ,
        \end{aligned}
        \label{eq:TimeEvolutionWithError}
    \end{equation}
    where $U_{t_1, t_2} := U_{t_1}U_{t_2}^\dagger$. Furthermore, for small $\tau \cdot \esssup_{t\in[0,\tau]} \|V_t \|$, we find
    \begin{equation}
        \| U_\tau' - U_\tau \| = \mathcal{O}\left( \tau \esssup_{t \in [0,\tau]} \|V_t \| \right) \ .
    \end{equation}
    \label{lemma:TimeEvolutionError}
\end{lemma}
\begin{proof}
    Consider the operator $V_t$ in the time-dependent rotating frame $U_t = \mathcal{T} \exp \left( -i\int_0^t H_{t'} dt' \right)$, which is the solution to the Schr\"odinger equation
    \begin{equation}
        H_t U_t = i \partial_t U_t \ \text{ a.e.}
    \end{equation}
    The Schr\"odinger equation using the full Hamiltonian
    \begin{equation}
        (H_t + V_t) U_t' = i \partial_t U_t' \ \text{ a.e.}
    \end{equation}
    transforms into the rotating frame as
    \begin{equation}
        \tilde{V}_t \tilde{U}_t = i \partial_t \tilde{U}_t \ \text{ a.e.,}
        \label{eq:TransformedDysonSeries}
    \end{equation}
    where $\tilde{V}_t := U_t^\dagger V_t U_t$ and $\tilde{U}_t := U_t^\dagger U_t' U_0$. This can easily be verified by evaluating the time derivative of $\tilde{U}_t$ on the right-hand side almost everywhere. Integrating Eq.~\eqref{eq:TransformedDysonSeries} results in
    \begin{equation}
        \tilde{U}_\tau = \mathbb{I} - i \int_0^\tau \tilde{V}_t dt + (-i)^2 \int_0^\tau \int_0^t \tilde{V}_t \tilde{V}_{t'} dt' dt + ... \ .
    \end{equation}
    Then, with $U_0 = \mathbb{I}$ and plugging the definitions into the Dyson series
    \begin{equation}
        U_\tau^\dagger U_\tau' = \mathbb{I} - i \int_0^\tau U_t^\dagger V_t U_t dt + (-i)^2 \int_0^\tau \int_0^t U_t^\dagger V_t U_t U_{t'}^\dagger V_{t'} U_{t'} dt' dt + ...
    \end{equation}
    and multiplying from the left with $U_\tau$ yields
    \begin{equation}
        U_\tau' = U_\tau - i \int_0^\tau U_{\tau, t} V_t U_{t, 0} dt + (-i)^2 \int_0^\tau \int_0^t U_{\tau, t}V_t U_{t, t'} V_{t'} U_{t', 0} dt' dt + ... \ ,
    \end{equation}
    where $U_{t_2, t_1} := U_{t_2} U_{t_1}^\dagger$, thus proving Eq.~\eqref{eq:TimeEvolutionWithError}. Using Eq.~\eqref{eq:TimeEvolutionWithError}, we can bound the error of the time evolution as
    \begin{equation}
        \begin{aligned}
            \left\| U_\tau' - U_\tau \right\| &= \left\| - i \int_0^\tau U_{\tau, t} V_t U_{t, 0} dt + (-i)^2 \int_0^\tau \int_0^t U_{\tau, t}V_t U_{t, t'} V_{t'} U_{t', 0} dt' dt + ... \right\| \\
            &\leq \left\| \int_0^\tau U_{\tau, t} V_t U_{t, 0} dt \right\| + \left\| \int_0^\tau \int_0^t U_{\tau, t}V_t U_{t, t'} V_{t'} U_{t', 0} dt' dt \right\| + ... \\
            &\leq \int_0^\tau \| U_{\tau, t} V_t U_{t, 0} \| dt + \int_0^\tau \int_0^t \| U_{\tau, t}V_t U_{t, t'} V_{t'} U_{t', 0} \| dt' dt + ... \\
            &\leq \int_0^\tau \| V_t  \| dt + \int_0^\tau \int_0^t \| V_t \| \| V_{t'} \| dt' dt + ... \\
            &\leq \sum_{n=1}^\infty \frac{\tau^n}{n!} \esssup_{t\in[0,\tau]} \|V_t \|^n \ ,
        \end{aligned}
    \end{equation}
    where we used the triangle inequality on the second and third lines and the submultiplicativity of the operator norm and that $\|U\| = 1$ for all unitaries $U$ on the fourth line. Thus, as $\tau \cdot \esssup_{t\in[0,\tau]} \|V_t \| \rightarrow 0$ and using Lemma~\ref{lemma:ExponentialSeriesBound}, in the leading order, we have $\left\| U_\tau' - U_\tau \right\| \leq \tau \esssup_{t\in[0,\tau]} \|V_t \| + \mathcal{O}\left( \tau^2 \esssup_{t\in[0,\tau]} \|V_t \|^2 \right) = \mathcal{O}\left( \tau \esssup_{t \in [0,\tau]} \|V_t \| \right)$.
\end{proof}

\subsection{Gershgorin circle theorem}
We will also use the Gershgorin circle theorem to establish bounds on operator norms. For the proof, we refer to the literature.
\begin{theorem} (Gershgorin Circle Theorem~\cite{Gershgorin31})
    \label{theorem:Gershgorin}
    Let $M$ be a complex square matrix with elements $m_{ij}$. All eigenvalues $\lambda$ of $M$ lie in the union of the disks $D(m_{ii}, R_i)$ centered on $m_{ii}$ and with radii
    \begin{equation}
        R_i = \sum_{j\neq i} |m_{ij}| \ .
    \end{equation}
\end{theorem}
With these tools in place, we can proceed to prove the main result formally.

%% file: 10_appendix_c.tex
\section{Proof of Theorem~\ref{theorem:MainResult}}
\label{sec:AppendixProofOfMainResult}
After having set the stage, we can prove the main result Theorem~\ref{theorem:MainResult}. As we have discussed in the main text, in order to prove polynomial equivalence, one has to show that the gate model simulates the global TFIM and that the global TFIM simulates the gate model. The gate model can simulate the global TFIM trivially, so to prove Theorem~\ref{theorem:MainResult} what is missing is to show the other direction. In Theorem~\ref{theorem:MainResultFormal}, we provide a constructive simulation method, including the scaling of the required resources. Theorem~\ref{theorem:MainResult} then follows directly from Theorem~\ref{theorem:MainResultFormal}, as $n_q', T, |J| = \poly \left(n_q, p, \varepsilon^{-1} \right)$, thus the global TFIM can simulate the gate model with polynomial overhead.\\
We will use the fact that any quantum circuit $U_{target}$ consisting of $p$ gates acting on $n_q$ qubits initialized in the state $\ket{0}$ can be mapped to a $n_q'$-qubit Hamiltonian with diagonal interactions and a sequence of $p'$ quantum gates conditionally applied to four groups of qubits by the CP-method. After the CP-method, the logical qubits live on wires, where they can be moved back and forth by an appropriate pulse sequence. We denote with $\Tr_{RO}$ the partial trace that only leaves the $n_q$ physical qubits where the logical quantum state lives at the end of the computation. This could be either at whichever location the computation terminates or at a dedicated read-out set of qubits. If executed correctly, the CP-method leaves all other qubits that do not carry logical quantum information in a classical product state. We will show that the $p'$ gates from the CP-method can be mapped to truly global pulses of the transverse field of the global TFIM, up to a controllable error.
\begin{theorem}
    Given any quantum circuit $U_{\text{target}}$ with $p$ gates acting on $n_q$ qubits. Then there is a Hamiltonian $H_t$ as in Eq.~\eqref{eq:DefFullHamiltonian} and a Lipschitz continuous schedule $\Gamma_t: [0, T] \rightarrow [0, \Gamma^*]$ with Lipschitz constant $\mathcal{O}(\Gamma^*)$, constructed from a concatenation of $\Gamma$-pulses and wait-times $\tau_i$ as defined in Section~\ref{sec:GammaPulseDefinition}, that generates a unitary
    \begin{equation}
        U_{T} = \mathcal{T} \exp\left( -i \int_0^T H_t dt \right)
    \end{equation}
    acting on $n_q'$ qubits such that
    \begin{equation}
        \| \Tr_{RO}(U_{T}\ket{0}^{\otimes n_q'}\bra{0}^{\otimes n_q'}U_{T}^\dagger) - U_{\text{target}} \ket{0}^{\otimes n_q}\bra{0}^{\otimes n_q}U_{\text{target}}^\dagger \| \leq \varepsilon \ .
    \end{equation}
    Furthermore, $n_q' = \mathcal{O}(n_q^2)$, $|J| = \mathcal{O}(p n_q^3 \varepsilon^{-1})$ and $T = \mathcal{O} \left( p^{8+2\nu} n_q^{30 + 8\nu} \varepsilon^{-7-2\nu} \right)$ for any $\nu >0$.
    \label{theorem:MainResultFormal}
\end{theorem}
\begin{proof}
    Using the CP-method, the circuit can be mapped to $n_q' = \mathcal{O}(n_q^2)$ qubits separated into four types $\mathcal{K} \in \{ \mathcal{A}, \mathcal{B}, \mathcal{C}, \mathcal{D} \}$ arranged on a ZZ-coupling graph $\mathcal{G}$ and a sequence of $p' = \mathcal{O} \left( n_q p \right)$ gates $\{ U_{\mathcal{I}_i, i}\}_{i=1}^{p'}$ acting globally on all qubits of type $\mathcal{I}_i \in \{ \mathcal{A}, \mathcal{B}, \mathcal{C}, \mathcal{D}, \mathcal{B} \cup \mathcal{D} \}$ under the condition that the neighboring qubits are in the ground state, as defined in Section~\ref{sec:ConditionalGatesDefinition}. In the limit of $|J|\rightarrow \infty$, when the gates act exactly conditionally, the CP-method allows for exact simulation of the circuit. Therefore, for the proof, we need to show two things:
    \begin{enumerate}
        \item that by controlling the schedule $\Gamma_t$, as defined in Section~\ref{sec:GammaPulseDefinition}, we can realize the conditional gates $U_{\mathcal{I}_i, i}$ necessary for the CP-method with a controllable error;
        \item that the error can also be controlled when $|J| = \poly\left(n_q, p, \varepsilon^{-1}\right)$.
    \end{enumerate}
    We arrange the qubits in the universal arrangement $\mathcal{G}$, as depicted in Figure~\ref{fig:UniversalConfig}, and couple them to $H_Z$ as in Eq.~\eqref{eq:RydbergHamiltonian}, which is a special case of $H_Z$ in Eq.~\eqref{eq:DefHz}. We choose $h_{\mathcal{K}} = \mathcal{O}(1)$ to be algebraic $\mathbb{Q}$-linearly independent numbers such that Theorem~\ref{theorem:Schmidt} is satisfied non-trivially, and a $\Gamma_t: [0, \tau] \rightarrow \mathbb{R}_{\geq 0}$ with $\left| \int_0^\tau \Gamma_t e^{i2h_\mathcal{K}t} dt \right| = \Theta(\Gamma^*)$ $\forall \mathcal{K}$, where $\Gamma^* := \max_{[0, \tau]} |\Gamma_t|$ can be controlled. We will first show universality in the limit $|J| \rightarrow \infty$. In this limit, the $J$-couplings act effectively as projectors $P$ onto the computational subspace, as defined in Section~\ref{sec:ComputationalSubspaceDefinition}. In a second step, we consider the error introduced by finite $|J|$.

    \subsection{Basic properties of the generated unitary}
    In the limit of infinite $|J|$, the effective Hamiltonian on $\mathcal{P}$ reads
    \begin{equation}
        H_{P, t} = PH_tP = PH_ZP + \Gamma_t P H_X P \ ,
    \end{equation}
    where $H_X = \sum_n X_n$, $H_Z = \sum_n h_n Z_n + \frac{J}{4} \sum_{(m,n) \in E(\mathcal{G})} ( 1 - Z_m )( 1 - Z_n )$ and $P$ is the projector onto the computational subspace. Since the interacting terms of $H_Z$ are zero in the computational subspace, as no neighboring qubits can be excited, we get
    \begin{equation}
        \begin{aligned}
            PH_ZP &= \sum_n h_n PZ_nP + \frac{J}{4} \sum_{(m,n) \in E(\mathcal{G})} P(1-Z_m)(1-Z_n)P \\
            &= \sum_n h_n PZ_nP \ .
        \end{aligned}
    \end{equation}
    Consider the unitary generated by a $\Gamma$-pulse
    \begin{equation}
        U_\Gamma = \mathcal{T} \exp \left( -i \int_0^\tau H_{P,t} dt \right) \ .
    \end{equation}
    Consider the basis change to the rotating frame $U_{Z,t} = \exp \left( itPH_ZP\right)$. Note that $[P, H_Z] = 0$ and therefore $[P, U_{Z,t}] = 0$. In the rotating frame, $U_\Gamma$ transforms to
    \begin{equation}
        \begin{aligned}
            \tilde{U}_\Gamma &= \mathcal{T} \exp \left( -i \int_0^\tau \Gamma_t U_{Z,t} P H_X P U_{Z,t}^\dagger dt  \right) \\
            &= \mathcal{T} \exp \left( -i \int_0^\tau \Gamma_t P U_{Z,t} H_X U_{Z,t}^\dagger P dt  \right) \\
            &= \mathbb{I} -i \int_0^\tau \Gamma_t P U_{Z,t} H_X U_{Z,t}^\dagger P dt + (-i)^2 \int_0^\tau \int_0^t \Gamma_t \Gamma_{t'} P U_{Z,t} H_X U_{Z,t}^\dagger P U_{Z,t'} H_X U_{Z,t'}^\dagger P dt' dt + ... \ ,
        \end{aligned}
        \label{eq:DysonSeries}
    \end{equation}
    where in the last line we used the Dyson series. We will use $\tilde{A}_t$ to denote the operators $A_t$ transformed into the rotating frame. The generator transforms to
    \begin{equation}
        \begin{aligned}
            \tilde{H}_t :=& \Gamma_t P U_{Z, t} H_X U_{Z,t}^\dagger P \\
            =& \Gamma_t \sum_n P \exp\left( i t h_n Z_n \right) X_n  \exp\left( -i t h_n Z_n \right) P \\
            =& \Gamma_t  \sum_n P \begin{pmatrix}
                0 & e^{it2h_n} \\
                e^{-it2h_n} & 0
            \end{pmatrix}_n P \ ,
        \end{aligned}
        \label{eq:DysonGenerator}
    \end{equation}
    where the subscript $n$ indicates that the $2\times2$-matrix acts on qubit $n$. Using the triangle inequalities and submultiplicativity of operator norms, the norm of $\tilde{H}_{t}$ can be bounded by
    \begin{equation}
        \begin{aligned}
            \| \tilde{H}_t \| &= |\Gamma_t| \left\| \sum_n P \exp\left( i t h_n Z_n \right) X_n  \exp\left( -i t h_n Z_n \right) P \right\| \\
            &\leq |\Gamma_t|  \sum_n \left\| P \exp\left( i t h_n Z_n \right) X_n  \exp\left( -i t h_n Z_n \right) P \right\| \\
            &\leq |\Gamma_t|  \sum_n \| P\|^2 \|X_n \| \| \exp\left( i t h_n Z_n \right) \| \| \exp\left(- i t h_n Z_n \right) \| \\
            &= |\Gamma_t| n_q' \leq \Gamma^* n_q'\ ,
        \end{aligned}
        \label{eq:GeneratorBound}
    \end{equation}
    where in the last line we used that $\|X_n\| = \|\exp\left( i t h_n Z_n \right)\| = \| P \| = 1$.\\
    We will operate in the limit of weak drive $\Gamma^* \ll 1$, such that the global drive via the transverse field can be considered to be an individual drive of each qubit. More formally, we will later choose $\Gamma^*$ such that we can approximate $\tilde{U}_\Gamma$ by its Dyson series up to first order. We need to show that the truncation error is controlled. To this end, let us denote the terms of the Dyson series Eq.~\eqref{eq:DysonSeries} as
    \begin{equation}
        \tilde{U}_{k,t} := (-i)^k \int_{0 \leq t'\leq t'' ... \leq t^{(k)} \leq t} \tilde{H}_{t^{(k)}} ... \tilde{H}_{t'} \ dt' ... dt^{(k)} \ ,
    \end{equation}
    such that Eq.~\eqref{eq:DysonSeries} can be expressed as
    \begin{equation}
        \tilde{U}_\Gamma = \mathbb{I} + \sum_{k=1}^\infty \tilde{U}_{k,\tau} \ .
    \end{equation}
    Using the submultiplicativity of operator norms, triangle inequalities and Eq.~\eqref{eq:GeneratorBound}, we can establish the bound
    \begin{equation}
        \begin{aligned}
            \| \tilde{U}_{k,t} \| &= \left\| \int_{0 \leq t'\leq t'' ... \leq t^{(k)} \leq t} \tilde{H}_{t^{(k)}} ... \tilde{H}_{t'} \ dt' ... dt^{(k)} \right\| \\
            &\leq \int_{0 \leq t'\leq t'' ... \leq t^{(k)} \leq t} \left\| \tilde{H}_{t^{(k)}} ... \tilde{H}_{t'} \right\| dt' ... dt^{(k)} \\
            &\leq \int_{0 \leq t'\leq t'' ... \leq t^{(k)} \leq t} \| \tilde{H}_{t^{(k)}} \| ... \| \tilde{H}_{t'} \| dt' ... dt^{(k)} \\
            &\leq (n_q')^k \int_{0 \leq t'\leq t'' ... \leq t^{(k)} \leq t} \left|\Gamma_{t^{(k)}}\right| ... \left|\Gamma_{t'}\right| dt' ... dt^{(k)} \\
            &\leq \frac{(n_q' \tau \Gamma^*)^k}{k!} \ .
        \end{aligned}
        \label{eq:DysonTermsBound}
    \end{equation}
    From this it follows straightforwardly that the Dyson series Eq.~\eqref{eq:DysonSeries} converges absolutely. Let us denote by $\Delta_\Gamma$ the error of the first-order Dyson expansion
    \begin{equation}
        \Delta_\Gamma := \tilde{U}_\Gamma - \mathbb{I} - \tilde{U}_{1,\tau} \ .
    \end{equation}
    We can bound the norm of the error using Lemma~\ref{lemma:ExponentialSeriesBound} by
    \begin{equation}
        \begin{aligned}
            \left\| \Delta_\Gamma \right\| &= \left\| \sum_{k=2}^\infty \tilde{U}_{k, \tau} \right\| \leq \sum_{k=2}^\infty \left\| \tilde{U}_{k, \tau} \right\| \leq \sum_{k=2}^\infty \frac{(n_q' \tau \Gamma^*)^k}{k!} \leq \frac{(n_q' \tau \Gamma^*)^2}{2} e^{n_q' \tau \Gamma^*} \ ,
        \end{aligned}
        \label{eq:FirstOrderError}
    \end{equation}
    and thus $\|\Delta_\Gamma \| = \mathcal{O} \left( n_q'^{2} \Gamma^{*2} \right)$ for small $n_q' \Gamma^*$. By choosing the maximum transverse field $\Gamma^*$ small enough, we can up to a small error $\delta$ approximate the effect of the $\Gamma$-pulse $\tilde{U}_\Gamma$ by the Dyson expansion up to first order.

    \subsection{Independent drives on subsets of qubits}
    For the implementation of the CP method, it is required to drive independent sets of non-adjacent qubits. Consider $\mathcal{I} \in \{ \mathcal{A}, \mathcal{B}, \mathcal{C}, \mathcal{D}, \mathcal{B} \cup \mathcal{D} \}$ the set of non-adjacent qubits required to be driven by a given gate $U_{\mathcal{I}_i, i}$ given by the CP-method, while $\mathcal{I}^c$ denotes its complement. Note that $|\mathcal{I}| + |\mathcal{I}^c| = n_q'$. The $\Gamma$-pulse acting only on the subset of qubits in $\mathcal{I}$ is given by the unitary
    \begin{equation}
        \begin{aligned}
            \tilde{U}_{\mathcal{I}} &= \mathcal{T}\exp\left(-i \int_0^\tau \sum_{n \in \mathcal{I}} \Gamma_t P \tilde{X}_{n,t} P dt\right) \\
            &= \mathbb{I} -i \int_0^\tau \sum_{n\in\mathcal{I}} \Gamma_t P \tilde{X}_{n,t} P dt + \int_0^\tau \int_0^t \sum_{n \in \mathcal{I}} \sum_{m \in \mathcal{I}} \Gamma_t \Gamma_{t'} P \tilde{X}_{n,t} P \tilde{X}_{m,t'} P dt' dt +  \ ...   \ .
        \end{aligned}
    \end{equation}
    By the argument presented above for the terms of the Dyson expansion of $\tilde{U}_\Gamma$, it follows analogously that
    \begin{equation}
        \| \tilde{U}_{\mathcal{I}, k} \| \leq \frac{(|\mathcal{I}| \tau\Gamma^*)^k}{k!} \ .
    \end{equation}
    Respectively, we define $\tilde{U}_{\mathcal{I}^c}$ with the analogous bound on the terms of the Dyson expansion. Then we can also define the errors $\Delta_{\mathcal{I}}$ and $\Delta_{\mathcal{I}^c}$ of the first-order approximation and bound their norms as
    \begin{equation}
        \begin{aligned}
            \| \Delta_\mathcal{I} \| &\leq \frac{(|\mathcal{I}| \tau \Gamma^*)^2}{2} e^{|\mathcal{I}| \tau \Gamma^*} = \mathcal{O} \left( \Gamma^{*2} |\mathcal{I}|^2 \right) \ , \\
            \| \Delta_{\mathcal{I}^c} \| &\leq \frac{(|\mathcal{I}^c| \tau\Gamma^*)^2}{2} e^{|\mathcal{I}^c| \tau \Gamma^*} = \mathcal{O} \left( \Gamma^{*2} |\mathcal{I}^c|^2 \right) \ .
        \end{aligned}
        \label{eq:DysonErrorSubsets}
    \end{equation}
    Compare this with the globally driven Ising model, where we can realize the unitary $\tilde{U}_{\Gamma}$ as in Eq.~\eqref{eq:DysonSeries}. We will proceed by proving two lemmata. First, we show that the following holds for small $\Gamma^*$:
    \begin{lemma}
        \label{lemma:FirstStep}
        Let $\tilde{U}_\Gamma$, $\tilde{U}_\mathcal{I}$ and $\tilde{U}_{\mathcal{I}^c}$ as above. Then, for small $n_q'\Gamma^*$, it holds that
        \begin{equation}
            \| \tilde{U}_{\mathcal{I}} \tilde{U}_{\mathcal{I}^c } - \tilde{U}_\Gamma \| = \mathcal{O}\left( (n_q' \Gamma^*)^2 \right) \ ,
            \label{eq:ErrorClaim1}
        \end{equation}
        and
        \begin{equation}
            \| \tilde{U}_{\mathcal{I}^c} \tilde{U}_{\mathcal{I} } - \tilde{U}_\Gamma \| = \mathcal{O}\left( (n_q' \Gamma^*)^2 \right) \ .
            \label{eq:ErrorClaim2}
        \end{equation}
    \end{lemma}
    We will present the proof of Lemma~\ref{lemma:FirstStep} below, but first we outline for what we will use it. If we let the system evolve for a time $\tau_i$ under the Hamiltonian $H_Z$, i.e. at $\Gamma_t = 0$, the resulting unitary is
    \begin{equation}
        U_Z(\tau_i) = \exp \left( -i\sum_{n} h_n \tau_i Z_n\right) = \prod_n \exp \left( -i h_n \tau_i Z_n\right) \ .
    \end{equation}
    As we require them in the following, we additionally define the phase gates acting exclusively on $\mathcal{I}$ and $\mathcal{I}^c$, respectively, as
    \begin{equation}
        \begin{aligned}
            U_{Z,\mathcal{I}}(\tau_i) &:= \exp \left(-i\sum_{n\in \mathcal{I}} h_n \tau_i Z_n \right) \ , \\
            U_{Z,\mathcal{I}^c}(\tau_i) &:= \exp \left(-i\sum_{n\in \mathcal{I}^c} h_n \tau_i Z_n \right) \ ,
        \end{aligned}
    \end{equation}
    Note that $U_Z(\tau_i) = U_{Z, \mathcal{I}}(\tau_i) U_{Z, \mathcal{I}^c}(\tau_i) = U_{Z, \mathcal{I}^c}(\tau_i) U_{Z, \mathcal{I}}(\tau_i)$. Using these phase gates and Lemma~\ref{lemma:FirstStep}, in a second step, we show that with appropriate phase gates $U_Z(\tau_i)$ chosen between and after two consecutive $\Gamma$-pulses, we can isolate the effect of the $\Gamma$-pulses on $\mathcal{I}$.\\
    We will use the $\tau_i$ to realize a certain set of phases $\phi_\mathcal{K}$. Therefore, to simplify notation, in the following we will replace the arguments $\tau_i$ with the phases $\phi_{\mathcal{I}}$, where $\phi_\mathcal{I} = \{ \phi_\mathcal{K} : \mathcal{K} \subseteq \mathcal{I} \}$ denotes the set of phases acting on qubit groups $\mathcal{K} \subseteq \mathcal{I}$. For example, we will consider the phase gates $U_{Z,\mathcal{I}}(\phi_\mathcal{I})$ as functions of the $\phi_{\mathcal{K}}$ for $\mathcal{K} \subseteq \mathcal{I}$. Theorem~\ref{theorem:FillingTime} and the appropriate choice of the $h_\mathcal{K}$ guarantee that there exist a $\tau_i$, such that we can consider the phases $\phi_\mathcal{K}$ as independent variables, up to an arbitrarily small error $\varepsilon_{\text{phase}} >0$. Consider the following lemma.
    \begin{lemma}
        \label{lemma:SecondStep}
        Let $\tilde{U}_\Gamma$, $\tilde{U}_\mathcal{I}$, $\tilde{U}_{\mathcal{I}^c}$, $U_Z$, $U_{Z,\mathcal{I}}$ and $U_{Z, \mathcal{I}^c}$ be as above and given two sets of phases $\phi_{\mathcal{I}, 1}$ and $\phi_{\mathcal{I}, 2}$. Then it holds for small $n_q' \Gamma^*$, $\varepsilon_{\text{phase}}>0$ and $\nu>0$ that there are $\tau_1, \tau_2 = \mathcal{O}\left( \varepsilon_{\text{phase}}^{-3-\nu}\right)$ such that
        \begin{equation}
            \left\|\tilde{U}_{\mathcal{I}} U_{Z,\mathcal{I}}(\phi_{\mathcal{I}, 1}) \tilde{U}_{\mathcal{I}} U_{Z,\mathcal{I}}(\phi_{\mathcal{I}, 2}) - \tilde{U}_\Gamma U_Z(\tau_1) \tilde{U}_\Gamma U_Z(\tau_2) \right\| = \mathcal{O}\left( n_q'^2 \Gamma^{*2} + n_q' \varepsilon_{\text{phase}} \right) \ .
            \label{eq:IndependentDriveClaim}
        \end{equation}
    \end{lemma}
    Lemma~\ref{lemma:SecondStep} shows that we can construct a sequence of global pulses with controllable $\phi_{\mathcal{I}, i}$, which only acts non-trivially on the desired subset $\mathcal{I}$ up to an error controlled by $\Gamma^*$ and $\varepsilon_{\text{phase}}$. We begin by proving the first step.

    \subsubsection{Proof of Lemma~\ref{lemma:FirstStep}}
    \begin{proof}
        Eq.~\eqref{eq:ErrorClaim1} follows directly from the Dyson series of $\tilde{U}_\mathcal{I}$, $\tilde{U}_{\mathcal{I}^c}$ and $\tilde{U}_\Gamma$ and that for the first-order terms we find by definition
        \begin{equation}
            \tilde{U}_{\mathcal{I}, 1} + \tilde{U}_{\mathcal{I}^c, 1} = \tilde{U}_{\Gamma, 1} \ .
        \end{equation}
        We can evaluate
        \begin{equation}
            \begin{aligned}
                &\tilde{U}_{\mathcal{I}} \tilde{U}_{\mathcal{I}^c} = \left( \mathbb{I} + \tilde{U}_{\mathcal{I}, 1} + \Delta_\mathcal{I} \right) \left( \mathbb{I} + \tilde{U}_{\mathcal{I}^c, 1} + \Delta_{\mathcal{I}^c} \right) \\
                =&\mathbb{I} + \tilde{U}_{\mathcal{I}, 1} + \tilde{U}_{\mathcal{I}^c, 1} + \Delta_{\mathcal{I}} + \Delta_{\mathcal{I}^c} + \Delta_{\mathcal{I}}\tilde{U}_{\mathcal{I}^c, 1} + \tilde{U}_{\mathcal{I}, 1} \Delta_{\mathcal{I}^c} + \Delta_{\mathcal{I}} \Delta_{\mathcal{I}^c} \\
                =&\mathbb{I} + \tilde{U}_{\Gamma,1} + \Delta_{\mathcal{I}} + \Delta_{\mathcal{I}^c} + \Delta_{\mathcal{I}}\tilde{U}_{\mathcal{I}^c, 1} + \tilde{U}_{\mathcal{I}, 1} \Delta_{\mathcal{I}^c} + \Delta_{\mathcal{I}} \Delta_{\mathcal{I}^c} \ ,
            \end{aligned}
        \end{equation}
        Recalling the bounds Eq.~\eqref{eq:FirstOrderError} and Eq.~\eqref{eq:DysonErrorSubsets}, we find for the norm of the difference for small $n_q'\Gamma^*$ using the triangle inequality and the submultiplicativity of operator norms
        \begin{equation}
            \begin{aligned}
                \left\| \tilde{U}_{\mathcal{I}} \tilde{U}_{\mathcal{I}^c} - \tilde{U}_\Gamma\right\| &= \left\| \Delta_{\mathcal{I}} + \Delta_{\mathcal{I}^c} + \Delta_{\mathcal{I}}\tilde{U}_{\mathcal{I}^c, 1} + \tilde{U}_{\mathcal{I}, 1} \Delta_{\mathcal{I}^c} + \Delta_{\mathcal{I}} \Delta_{\mathcal{I}^c} - \Delta_\Gamma \right\| \\
                &\leq \left\| \Delta_{\mathcal{I}} \right\| + \left\|\Delta_{\mathcal{I}^c} \right\| + \left\| \Delta_{\mathcal{I}}\tilde{U}_{\mathcal{I}^c, 1} \right\| + \left\|\tilde{U}_{\mathcal{I}, 1} \Delta_{\mathcal{I}^c} \right\| + \left\|\Delta_{\mathcal{I}} \Delta_{\mathcal{I}^c} \right\|  + \| \Delta_\Gamma \| \\
                &\leq \left\| \Delta_{\mathcal{I}} \right\| + \left\|\Delta_{\mathcal{I}^c} \right\| + \left\| \Delta_{\mathcal{I}} \right\| \left\| \tilde{U}_{\mathcal{I}^c, 1} \right\| + \left\|\tilde{U}_{\mathcal{I}, 1} \right\| \left\| \Delta_{\mathcal{I}^c} \right\| + \left\|\Delta_{\mathcal{I}} \right\| \left\| \Delta_{\mathcal{I}^c} \right\| + \| \Delta_\Gamma \|\\
                &=\mathcal{O}\left( (n_q' \tau \Gamma^*)^2 \right).
            \end{aligned}
        \end{equation}
        Analogously, Eq.~\eqref{eq:ErrorClaim2} can be derived, the only difference being the order of some of the operators, which does not affect the bounds. This concludes the proof.
    \end{proof}
    Lemma~\ref{lemma:FirstStep} is a tool that now allows us to prove the next key step. For this step, we will additionally use the quasi-independent phase control of the qubit groups, which is guaranteed by Theorem~\ref{theorem:FillingTime}.

    \subsubsection{Proof of Lemma~\ref{lemma:SecondStep}}
    \begin{proof}
        We continue to prove the claim Eq.~\eqref{eq:IndependentDriveClaim}.
        According to Eq.~\eqref{eq:ErrorClaim1} and Eq.~\eqref{eq:ErrorClaim2} in Lemma~\ref{lemma:FirstStep}, we can state that
        \begin{equation}
            \begin{aligned}
                \tilde{U}_\Gamma = \tilde{U}_{\mathcal{I}} \tilde{U}_{\mathcal{I}^c} + \Delta_{\mathcal{I}\mathcal{I}^c} = \tilde{U}_{\mathcal{I}^c} \tilde{U}_{\mathcal{I}} + \Delta_{\mathcal{I}^c\mathcal{I}} \ ,
            \end{aligned}
            \label{eq:SecondStep01}
        \end{equation}
        where $\Delta_{\mathcal{I}\mathcal{I}^c}$, $\Delta_{\mathcal{I}^c\mathcal{I}}$ are the errors, for which $\| \Delta_{\mathcal{I}\mathcal{I}^c} \|, \| \Delta_{\mathcal{I}^c\mathcal{I}} \| = \mathcal{O}\left( (n_q' \Gamma^*)^2 \right) $. Using Eq.~\eqref{eq:SecondStep01}, we can write
        \begin{equation}
            \begin{aligned}
                \tilde{U}_\Gamma U_Z(\tau_1) \tilde{U}_\Gamma U_Z(\tau_2) &= \left( \tilde{U}_{\mathcal{I}}\tilde{U}_{\mathcal{I}^c} + \Delta_{\mathcal{I}\mathcal{I}^c} \right) U_Z(\tau_1) \left( \tilde{U}_{\mathcal{I}^c}\tilde{U}_{\mathcal{I}} + \Delta_{\mathcal{I}^c\mathcal{I}} \right) U_Z(\tau_2) \ . 
            \end{aligned}
            \label{eq:SecondStep02}
        \end{equation}
        We can further decompose $U_Z$. Recall that $U_Z(\tau_i) = U_{Z,\mathcal{I}}(\tau_i)U_{Z,\mathcal{I}^c}(\tau_i) = U_{Z,\mathcal{I}^c}(\tau_i) U_{Z,\mathcal{I}}(\tau_i)$. Note that since $\tilde{U}_{\mathcal{I}} \ (\tilde{U}_{\mathcal{I}^c})$ directly drives qubits in $\mathcal{I}$ ($\mathcal{I}^c$) and only acts on the qubits in $\mathcal{I}^c \ (\mathcal{I})$ via the diagonal projector $P$, it commutes with $U_{Z,\mathcal{I}^c}$ ($U_{Z,\mathcal{I}}$). Furthermore, according to Theorem~\ref{theorem:FillingTime}, for any arbitrarily small error $\varepsilon_{\text{phase}} > 0$ and $\nu > 0$ there are $\tau_i = \mathcal{O}\left( \varepsilon_{\text{phase}}^{-3-\nu} \right)$, $i=1,2$, such that for all $\mathcal{K}$ and for any set of phases $\phi_{\mathcal{K},i}$ we find that
        \begin{equation}
            |h_\mathcal{K} \tau_i - \phi_{\mathcal{K}, i}| \leq \varepsilon_{\text{phase}} \mod 2\pi \ .
        \end{equation}
        Since by definition the set $\mathcal{I}$ contains all qubits of one or two types, we can choose phases $\phi_{\mathcal{K},i}$ such that for any $\mathcal{K} \subseteq \mathcal{I}^c$
        \begin{equation}
            \begin{aligned}
                |h_\mathcal{K}\tau_1 - \pi/2 | &\leq \varepsilon_{\text{phase}} \mod 2\pi \ , \\
                |h_\mathcal{K}\tau_2 - 3\pi/2 | &\leq \varepsilon_{\text{phase}} \mod 2\pi \ ,
            \end{aligned}
        \end{equation}
        and for given $\phi_{\mathcal{K},i}$ when $\mathcal{K} \subseteq \mathcal{I}$ that
        \begin{equation}
            \begin{aligned}
                |h_{\mathcal{K}}\tau_i - \phi_{\mathcal{K},i} | &\leq \varepsilon_{\text{phase}} \mod 2\pi \ .
            \end{aligned}
        \end{equation}
        Therefore, we find that
        \begin{equation}
            \begin{aligned}
                U_{Z,\mathcal{I}^c}(\tau_i) &= \exp \left( \pm i \frac{\pi}{2} \sum_{n \in \mathcal{I}^c} Z_n \right) \exp \left( i \sum_{n \in \mathcal{I}^c} \varepsilon_{n,i} Z_n \right) \\
                &= \left( \prod_{n \in \mathcal{I}^c} \pm iZ_n \right) \exp \left( i \sum_{n \in \mathcal{I}^c} \varepsilon_{n,i} Z_n \right) \ ,
            \end{aligned}
        \end{equation}
        and
        \begin{equation}
            \begin{aligned}
                U_{Z, \mathcal{I}}(\tau_i) &= \exp \left( -i \sum_{n \in \mathcal{I}} \phi_{n,i} Z_n \right) \exp \left( -i \sum_{n \in \mathcal{I}} \varepsilon_{n,i} Z_n \right) \\
                &= U_{Z, \mathcal{I}}(\phi_{\mathcal{I}, i}) \exp \left( -i \sum_{n \in \mathcal{I}} \varepsilon_{n,i} Z_n \right) \ ,
            \end{aligned}
        \end{equation}
        where $\varepsilon_{n,i}$ are the respective errors with $|\varepsilon_{n,i}| \leq \varepsilon_{\text{phase}}$ and the signs $\pm$ depend on $i=1$ or $i=2$. Hence, we can write for small $\varepsilon_{\text{phase}}$ that
        \begin{equation}
            \begin{aligned}
                \left\| U_{Z,\mathcal{I}^c}(\tau_i) - \prod_{n\in\mathcal{I}^c} \pm iZ_n \right\| &< C_{\mathcal{I}^c}|\mathcal{I}^c| \varepsilon_{\text{phase}} \ , \\
                \left\| U_{Z,\mathcal{I}}(\tau_i) - U_{Z, \mathcal{I}}(\phi_{\mathcal{I}, i}) \right\| &< C_{\mathcal{I}} |\mathcal{I}| \varepsilon_{\text{phase}} \ ,
            \end{aligned}
            \label{eq:PhaseError}
        \end{equation}
        for constants $C_{\mathcal{I}}$, $C_{\mathcal{I}^c}$, which follows easily from the Taylor expansion of the errors $\exp(-i\varepsilon_{n,i} Z_n)$. Equivalently, we can state for small $\varepsilon_{\text{phase}}$
        \begin{equation}
            \begin{aligned}
                U_Z(\tau_i) &= U_{Z, \mathcal{I}}(\phi_{\mathcal{I}, i}) \prod_{n \in \mathcal{I}^c} \pm iZ_n + \Delta_{Z,i} \ ,
            \end{aligned}
            \label{eq:SecondStep03}
        \end{equation}
        where $\Delta_{Z,i}$ is the respective error term with $\|\Delta_{Z,i} \| \leq C_Z n_q' \varepsilon_{\text{phase}}$ for some constant $C_Z$.\\
        Using Eq.~\eqref{eq:SecondStep01} and Eq.~\eqref{eq:SecondStep03} in Eq.~\eqref{eq:SecondStep02}, we find
        \begin{equation}
            \begin{aligned}
                &\tilde{U}_\Gamma U_Z(\tau_1) \tilde{U}_\Gamma U_Z(\tau_2) \\
                = &\left( \tilde{U}_{\mathcal{I}}\tilde{U}_{\mathcal{I}^c} + \Delta_{\mathcal{I}\mathcal{I}^c} \right) \left( U_{Z, \mathcal{I}}(\phi_{\mathcal{I}, 1}) \prod_{n \in \mathcal{I}^c} +iZ_n + \Delta_{Z,1} \right) \left( \tilde{U}_{\mathcal{I}^c}\tilde{U}_{\mathcal{I}} + \Delta_{\mathcal{I}^c\mathcal{I}} \  \right) \left( U_{Z, \mathcal{I}}(\phi_{\mathcal{I}, 2}) \prod_{n \in \mathcal{I}^c} -iZ_n + \Delta_{Z,2} \right) \\
                = &\tilde{U}_{\mathcal{I}}\tilde{U}_{\mathcal{I}^c}  \left( U_{Z, \mathcal{I}}(\phi_{\mathcal{I}, 1}) \prod_{n \in \mathcal{I}^c} iZ_n \right) \tilde{U}_{\mathcal{I}^c}\tilde{U}_{\mathcal{I}} \left( U_{Z, \mathcal{I}}(\phi_{\mathcal{I}, 2}) \prod_{n \in \mathcal{I}^c} -iZ_n  \right) + \Delta_{\Gamma Z} \ ,
            \end{aligned}
            \label{eq:SecondStep04}
        \end{equation}
        where the operator $\Delta_{\Gamma Z}$ absorbs all the error terms that occur when multiplying the terms in the second line of Eq.~\eqref{eq:SecondStep04}. The error terms in Eq.~\eqref{eq:SecondStep04} are either multiplied by the unitary terms, leaving the norms unchanged, or by other error terms, which for small errors leaves the leading order unchanged. Using the triangle inequality and the submultiplicativity of the operator norms, we can thus bound the norm of $\Delta_{\Gamma Z}$ as
        \begin{equation}
            \begin{aligned}
                \| \Delta_{\Gamma Z} \| &\leq \|\Delta_{\mathcal{I}\mathcal{I}^c} \| + \|\Delta_{\mathcal{I}^c\mathcal{I}} \| + \| \Delta_{Z,1} \| + \| \Delta_{Z,2} \| + \ ... \ + \|\Delta_{\mathcal{I}\mathcal{I}^c} \| \|\Delta_{\mathcal{I}^c\mathcal{I}} \| \| \Delta_{Z,1} \| \| \Delta_{Z,2} \| \\
                &= \mathcal{O} \left( n_q' \varepsilon_{\text{phase}} + (n_q' \Gamma^*)^{2} \right)
            \end{aligned}
            \label{eq:DeltaGammaZ}
        \end{equation}
        in the limit of small $\varepsilon_{\text{phase}}$ and $n_q' \Gamma^*$. Using that $U_{Z,\mathcal{I}^c}$ ($U_{Z,\mathcal{I}}$) commutes with $\tilde{U}_{\mathcal{I}}$ ($\tilde{U}_{\mathcal{I}^c}$), we further write
        \begin{equation}
            \begin{aligned}
            &\tilde{U}_\Gamma U_Z(\tau_1) \tilde{U}_\Gamma U_Z(\tau_2) \\
            = & \tilde{U}_{\mathcal{I}} U_{Z, \mathcal{I}}(\phi_{\mathcal{I}, 1}) \tilde{U}_{\mathcal{I}^c} \left( \prod_{n \in \mathcal{I}^c} iZ_n \right) \tilde{U}_{\mathcal{I}^c} \left( \prod_{n \in \mathcal{I}^c} -iZ_n \right) \tilde{U}_{\mathcal{I}} U_{Z, \mathcal{I}}(\phi_{\mathcal{I}, 2}) + \Delta_{\Gamma Z} \ ,
            \end{aligned}
            \label{eq:MultiQubitWithZ}
        \end{equation}
        so that we can focus on the operators that act on $\mathcal{I}^c$. Recall the expansion to first order
        \begin{equation}
            \tilde{U}_{\mathcal{I}^c} = \mathbb{I} - i \sum_{n \in \mathcal{I}^c} \int_0^\tau \Gamma_t P \tilde{X}_{n,t} P dt + \Delta_{\mathcal{I}^c} \ ,
        \end{equation}
        and thus, using $[\tilde{X}_{n,t}, Z_m] = 0$ if $m \neq n$ and $[P, Z_n] = 0$ for all $n$, as well as the definition of $\tilde{X}_{n,t} = \exp (ith_nZ_n) X_n \exp (-ith_nZ_n)$ and $ZXZ = -X$,
        \begin{equation}
            \begin{aligned}
                \left( \prod_{n \in \mathcal{I}^c} iZ_n \right) \tilde{U}_{\mathcal{I}^c} \left( \prod_{n \in \mathcal{I}^c} -iZ_n \right) &= \mathbb{I} - i \cdot i \cdot (-i) \cdot \sum_{n \in \mathcal{I}^c} \int_0^\tau \Gamma_t Z_n P \tilde{X}_{n,t} P Z_n dt + \mathcal{O}(\| \Delta_{\mathcal{I}^c} \|) \\
                &= \mathbb{I} - i  \sum_{n \in \mathcal{I}^c} \int_0^\tau \Gamma_t P Z_n \tilde{X}_{n,t} Z_n P dt + \mathcal{O}(\| \Delta_{\mathcal{I}^c} \|) \\
                &= \mathbb{I} + i \sum_{n \in \mathcal{I}^c} \int_0^\tau \Gamma_t P \tilde{X}_{n,t} P dt + \mathcal{O}(\| \Delta_{\mathcal{I}^c} \|) \ ,
            \end{aligned}
        \end{equation}
        which is the inverse of $\tilde{U}_{\mathcal{I}^c}$ up to first order. Then, we find for the unitaries acting on $\mathcal{I}^c$ in Eq.~\eqref{eq:MultiQubitWithZ} that
        \begin{equation}
            \begin{aligned}
                \tilde{U}_{\mathcal{I}^c} \left( \prod_{n \in \mathcal{I}^c} iZ_n \right) \tilde{U}_{\mathcal{I}^c} \left( \prod_{n \in \mathcal{I}^c} -iZ_n \right) &= \mathbb{I} + \Delta_{\mathcal{I}^c}' \ ,
            \end{aligned}
            \label{eq:CancelledDrive}
        \end{equation}
        where $\| \Delta_{\mathcal{I}^c}' \| =  \mathcal{O} \left( \| \Delta_{\mathcal{I}^c} \| \right)$. For the operator Eq.~\eqref{eq:MultiQubitWithZ}, we find with the error Eq.~\eqref{eq:DeltaGammaZ} and Eq.~\eqref{eq:CancelledDrive} that
        \begin{equation}
            \begin{aligned}
                \tilde{U}_\Gamma U_Z(\tau_1) \tilde{U}_\Gamma U_Z(\tau_2) &= \tilde{U}_{\mathcal{I}} U_{Z, \mathcal{I}}(\phi_{\mathcal{I}, 1}) \left( \mathbb{I}+\Delta_{\mathcal{I}^c}' \right) \tilde{U}_{\mathcal{I}} U_{Z, \mathcal{I}}(\phi_{\mathcal{I}, 2}) + \Delta_{\Gamma Z} \\
                &= \tilde{U}_{\mathcal{I}} U_{Z, \mathcal{I}}(\phi_{\mathcal{I}, 1}) \tilde{U}_{\mathcal{I}} U_{Z, \mathcal{I}}(\phi_{\mathcal{I}, 2}) + \mathcal{O}\left( \| \Delta_{\mathcal{I}^c}' \| \right) + \Delta_{\Gamma Z} \\
                &= \tilde{U}_{\mathcal{I}} U_{Z, \mathcal{I}}(\phi_{\mathcal{I}, 1}) \tilde{U}_{\mathcal{I}} U_{Z, \mathcal{I}}(\phi_{\mathcal{I}, 2}) + \mathcal{O}\left(n_q' \varepsilon_{\text{phase}} + n_q'^2 \Gamma^{*2} \right) \ ,
            \end{aligned}
            \label{eq:IndependentDriveResult}
        \end{equation}
        where $\|\Delta_{\mathcal{I}^c}' \| =  \mathcal{O} \left( \| \Delta_{\mathcal{I}^c} \| \right) = \mathcal{O} \left( (n_q' \Gamma^*)^2 \right)$ from Eq.~\eqref{eq:DysonErrorSubsets}, which concludes the proof.
    \end{proof}
    Thus, we can conclude from Lemma~\ref{lemma:SecondStep} that for the right choice of waiting times $\tau_1$ and $\tau_2$ between and after two $\Gamma$-pulses, we can construct independent driving of the subset of qubits $\mathcal{I}$ programmable via the phases $\phi_{\mathcal{I},i}$ up to an error controlled by the $\tau_i$ and the maximal transverse field $\Gamma^*$.

    \subsection{Universal gates on subset $\mathcal{I}$}
    Having established how to emulate a targeted drive on a selected subset of qubits $\mathcal{I}$, we now show how to construct arbitrary conditional gates acting on qubits in $\mathcal{I} \in \{ \mathcal{A}, \mathcal{B}, \mathcal{C}, \mathcal{D}, \mathcal{B} \cup \mathcal{D} \}$. To this end, we will prove the following lemma using Lemma~\ref{lemma:UniversalSU2}. Recall that $\mathcal{I}$ only contains non-adjacent qubits.
    \begin{lemma}
        \label{lemma:UniversalGateSequence}
        Let $\tilde{U}_\Gamma$ and $U_Z$ as above. For any $V_\mathcal{K} \in \text{SU}(2)$, $\mathcal{K} \in \{ \mathcal{A}, \mathcal{B}, \mathcal{C}, \mathcal{D} \} \subseteq \mathcal{I}$, and its projected application $V_\mathcal{K}^{(c)}$ defined as in Section~\ref{sec:ConditionalGatesDefinition}, and any $\nu >0$ and $\varepsilon_{\text{phase}} >0$, there exist an $N\in\mathbb{N}$, $N = \Theta(\Gamma^{*-1})$ and waiting times $\tau_i = \mathcal{O}\left( \varepsilon_{\text{phase}}^{-3 -\nu} \right)$, $i=0, 1, ... N$, such that
        \begin{equation}
            \begin{aligned}
                U_Z(\tau_0)&\prod_{i=1}^{N/2} \tilde{U}_\Gamma U_Z(\tau_{2i-1}) \tilde{U}_\Gamma U_Z(\tau_{2i}) \\
                =& \prod_{\mathcal{K} \subseteq \mathcal{I}} \prod_{n\in\mathcal{K}} V_{\mathcal{K},n}^{(c)} + \mathcal{O}\left( n_q' \varepsilon_{\text{phase}} / \Gamma^{*} + n_q'^2 \Gamma^{*} \right) \ .
            \end{aligned}
            \label{eq:UniversalGateWithError}
        \end{equation}
    \end{lemma}
    \begin{proof}
        Using Lemma~\ref{lemma:SecondStep}, we can independently drive each relevant subset of qubits $\mathcal{I} \in \{ \mathcal{A}, \mathcal{B}, \mathcal{C}, \mathcal{D}, \mathcal{B} \cup \mathcal{D} \}$ by concatenating sequences of $\Gamma$-pulses with the respective waiting times $\tau_i$ between them. Note that $\tilde{U}_\mathcal{I}$ factorizes, since the CP-method only requires simultaneous driving of non-adjacent qubits, and therefore can be written as
        \begin{equation}
            \begin{aligned}
                \tilde{U}_\mathcal{I} &= \prod_{n\in\mathcal{I}} \mathcal{T}\exp\left(-i\int_0^\tau \Gamma_t P\tilde{X}_{n,t}P dt \right) \\
                &= \prod_{n\in\mathcal{I}} \left( \mathbb{I}-i\int_0^\tau \Gamma_t P\tilde{X}_{n,t}P dt + (-i)^2 \int_0^\tau \int_0^t \Gamma_t \Gamma_{t'} P\tilde{X}_{n,t} P \tilde{X}_{n,t'} P dt' dt + ... \right) \ .
            \end{aligned}
        \end{equation}
        Consider a concatenation of $\tilde{U}_\Gamma$ and $U_Z(\tau_i)$ and apply Lemma~\ref{lemma:SecondStep} to obtain
        \begin{equation}
            \begin{aligned}
                &U_Z(\tau_0) \prod_{i=1}^{N/2} \tilde{U}_\Gamma U_Z(\tau_{2i-1}) \tilde{U}_\Gamma U_Z(\tau_{2i}) \\
                =&U_{Z,\mathcal{I}}(\phi_{\mathcal{I},0}) \prod_{i=1}^{N/2} \left[ \tilde{U}_{\mathcal{I}} U_{Z, \mathcal{I}}(\phi_{\mathcal{I}, 2i-1}) \tilde{U}_{\mathcal{I}} U_{Z, \mathcal{I}}(\phi_{\mathcal{I},2i}) + \mathcal{O}\left(n_q' \varepsilon_{\text{phase}} + n_q'^2 \Gamma^{*2} \right) \right] \\
                =&U_{Z,\mathcal{I}}(\phi_{\mathcal{I},0}) \prod_{i=1}^{N/2} \tilde{U}_{\mathcal{I}} U_{Z, \mathcal{I}}(\phi_{\mathcal{I},2i-1}) \tilde{U}_{\mathcal{I}} U_{Z, \mathcal{I}}(\phi_{\mathcal{I},2i}) + \mathcal{O}\left(Nn_q' \varepsilon_{\text{phase}} + N n_q'^2 \Gamma^{*2} \right) \ ,
            \end{aligned}
            \label{eq:ApproximateSequence}
        \end{equation}
        where in the last line we again consider $\varepsilon_{\text{phase}}$ and $n_q' \Gamma^*$ to be small. Furthermore, in Eq.~\eqref{eq:ApproximateSequence}, we have chosen $\tau_0$ such that the phases $\phi_{\mathcal{I}^c, 0} = 0$ up to an error $\leq n_q' \varepsilon_{\text{phase}}$, which is absorbed into the error term.\\
        Recall that the $\Gamma$-pulses act on qubit $n \in \mathcal{K} \subseteq \mathcal{I}$ as a unitary of the form
        \begin{equation}
            \begin{aligned}
                \mathcal{T} \exp \left( -i \int_0^\tau \Gamma_t P\tilde{X}_{n,t}P dt \right) = \begin{pmatrix}
                    \sqrt{1-r_\mathcal{K}} e^{-i(\alpha_n + \beta_n)} & \sqrt{r_\mathcal{K}} e^{-i(\alpha_n - \beta_n)} \\
                    \sqrt{r_\mathcal{K}} e^{i(\alpha_n - \beta_n)} & \sqrt{1-r_\mathcal{K}} e^{i(\alpha_n + \beta_n)}
                \end{pmatrix}^{(c)}_n \ ,
            \end{aligned}
        \end{equation}
        where, as before, the subscript indicates that the operator acts on qubit $n$ and the superscript indicates that the operator acts conditionally if $n$ is not blockaded. If $n$ is blockaded, $\tilde{U}_{\mathcal{I}}$ acts as the identity. However, the phase gates $U_{Z,\mathcal{I}}$ act non-trivially in both cases. By definition, $\mathcal{I}$ only contains non-adjacent qubits. Therefore, the sequence in Eq.~\eqref{eq:ApproximateSequence} does not change if any of the qubits in $\mathcal{I}$ are blockaded or not and either the entire sequence acts as the uncontrolled gate, or only the phases are applied. Thus, we can construct arbitrary conditional gates by a proper choice of phases $\phi_{\mathcal{I},i}$, as we will demonstrate now.\\
        Recall that with Eq.~\eqref{eq:DysonSeries} and Eq.~\eqref{eq:DysonGenerator}, it follows for the off-diagonal element of the $\Gamma$-pulse acting on a non-blockaded qubit $n\in \mathcal{K} \subseteq \mathcal{I}$ that
        \begin{equation}
            \sqrt{r_\mathcal{K}} =\left|\bra{0} \mathcal{T} \exp \left( -i \int_0^\tau \Gamma_t \tilde{X}_{n,t} dt \right) \ket{1} \right| = \left| -i \int_0^\tau \Gamma_t e^{i2h_\mathcal{K} t} dt  + \mathcal{O}\left( \Gamma^{*2} \right) \right| = \Theta(\Gamma^*) \ ,
        \end{equation}
        where the last equality holds by assumption. Applying Lemma~\ref{lemma:UniversalSU2}, there is an $N_0 = \mathcal{O}\left( \max_{\mathcal{K} \subseteq \mathcal{I}} r_\mathcal{K}^{-1/2} \right) = \mathcal{O}\left( \Gamma^{*-1} \right)$ such that for $N / 3 \geq N_0$ the phases $\phi_{\mathcal{I},i}$ can be chosen such that the sequence acts as any given $V_n \in \text{SU(2)}$ element on $\mathcal{I}$, i.e.
        \begin{equation}
            U_{Z, \mathcal{I}}(\phi_{\mathcal{I}, 0})\prod_{i=1}^{N/6} \tilde{U}_{\mathcal{I}} U_{Z, \mathcal{I}}(\phi_{\mathcal{I}, 2i-1}) \tilde{U}_{\mathcal{I}} U_{Z, \mathcal{I}}(\phi_{\mathcal{I}, 2i}) = \prod_{n \in \mathcal{I}} V_n^+ \ ,
        \end{equation}
        where
        \begin{equation}
            V_n^+ = \begin{cases}
                V_n & \text{if $n$ is not blockaded} \\
                \exp(-i\phi_{V_n} Z_n) & \text{if $n$ is blockaded}
            \end{cases}
        \end{equation}
        and
        \begin{equation}
            \exp(-i\phi_{V_n} Z_n):= \prod_{i=0}^{N/3} \exp \left( -ih_n \tau_i Z_n \right) \ .
        \end{equation}
        The phase $\phi_{V_n}$ depends on the choice of $V_n$ via the choice of the waiting times $\tau_i$. Note that due to the phase acting on blockaded $n$, $V_n^+ \neq V_n^{(c)}$. To undo the undesired phase acting on blockaded qubits $n \in \mathcal{I}$, we apply five sequential operations. The sequence reads
        \begin{equation}
            \exp(-i \phi_{\text{corr}}Z) X^+_n \exp(-i \phi_{\text{corr}}Z) X^+_n V_n^+ \ .
        \end{equation}
        The phase gates can always be applied simply by idling at $\Gamma_t = 0$, while the $\tilde{X}_n$-gates have a representation of the form Eq.~\eqref{eq:ApproximateSequence}. The $X_n^+$ will act as bitflips on non-blockaded qubits, while acting as a phase $\exp(-i\phi_X Z)$ on blockaded qubits. Note, that $X_nZ_nX_n = -Z_n$ and therefore the $X_n$-gates and phase corrections cancel on non-blockaded qubits. On blockaded qubits, the sequence acts as phases, i.e.
        \begin{equation}
            \begin{aligned}
                &\exp(-i \phi_{\text{corr}}Z) X^+_n \exp(-i \phi_{\text{corr}}Z) X^+_n V_n^+ \\
                &= \begin{cases}
                \exp(-i \phi_{\text{corr}}Z) X \exp(-i \phi_{\text{corr}}Z) X V_n & \text{if $n$ is not blockaded} \\
                \exp\left(-i(2\phi_{\text{corr}} + 2\phi_X + \phi_V)Z\right) & \text{if $n$ is blockaded}
            \end{cases} \\
            &= \begin{cases}
                V_n & \text{if $n$ is not blockaded} \\
                \exp\left(-i(2\phi_{\text{corr}} + 2\phi_X + \phi_V)Z\right) & \text{if $n$ is blockaded}
            \end{cases}  \ .
            \end{aligned}
        \end{equation}
        Finally, we use the fact that $\phi_{\text{corr}}$ can always be chosen such that
        \begin{equation}
            2\phi_{\text{corr}} + 2\phi_X + \phi_V = 0
        \end{equation}
        and therefore
        \begin{equation}
        \exp(-i \phi_{\text{corr}}Z) X^+_n \exp(-i \phi_{\text{corr}}Z) X^+_n V_n^+ = 
            \begin{cases}
                V_n & \text{if $n$ is not blockaded} \\
                \mathbb{I} & \text{if $n$ is blockaded}
            \end{cases} = V_n^{(c)}
        \end{equation}
        can be realized by a sequence of $N$ $\Gamma$-pulses and phase gates. Each of the $X_n^+$ and $V_n^+$ requires $N/3$ pulses, resulting in a total of $N$ pulses. Thus, concatenating three sequences as in Eq.~\eqref{eq:ApproximateSequence} separated by appropriate phase gates, we can implement arbitrary conditional gates acting on $\mathcal{I}$ up to an additive error $\mathcal{O}\left( N n_q'\varepsilon_{\text{phase}} + N n_q'^2 \Gamma^{*2} \right)$ for some $N = \mathcal{O} \left( \max_{\mathcal{K} \subseteq \mathcal{I}} r_\mathcal{K}^{-1/2} \right)$. Therefore, we can set $N = \Theta(\Gamma^{*-1})$, resulting in
        \begin{equation}
            \begin{aligned}
                &U_Z(\tau_0)\prod_{i=1}^{N/2} \tilde{U}_\Gamma U_Z(\tau_{2i-1}) \tilde{U}_\Gamma U_Z(\tau_{2i}) \\
                =& \prod_{\mathcal{K} \subseteq \mathcal{I}} \prod_{n\in\mathcal{K}} V_{\mathcal{K}, n}^{(c)} + \mathcal{O}\left( n_q' \varepsilon_{\text{phase}} / \Gamma^{*} + n_q'^2 \Gamma^{*} \right)
            \end{aligned}
        \end{equation}
        for any conditional gates $V^{(c)}_{\mathcal{K}, n}$ in the limit of small $\Gamma^*$ and $\varepsilon_{\text{phase}}$. This concludes the proof.
    \end{proof}
    Since Lemma~\ref{lemma:UniversalGateSequence} applies to arbitrary conditional single-qubit gates on non-adjacent subsets $\mathcal{I} \in \{ \mathcal{A}, \mathcal{B}, \mathcal{C}, \mathcal{D}, \mathcal{B} \cup \mathcal{D} \}$, it naturally follows that any of the gates $U_{\mathcal{I}_i, i}$ required for the CP-method can be implemented up to a controllable error. There are two contributions to the error in Lemma~\ref{lemma:UniversalGateSequence}. While reducing the contribution from $n_q'^2 \Gamma^*$ by decreasing $\Gamma^*$ increases the error of the phase contribution $n_q' \varepsilon_{\text{phase}} \Gamma^{*-1}$, the phase error $\varepsilon_{\text{phase}}$ can be reduced independently by choosing longer waiting times $\mathcal{O}( \varepsilon_{\text{phase}}^{-3-\nu} )$, according to Theorem~\ref{theorem:FillingTime}. Therefore, the error contributions can be arbitrarily reduced. Thus, the concatenation of $\Gamma$-pulses with appropriately chosen waiting times $\tau_i$ between them can approximate the conditional single-qubit gates required for the CP-method.\\
    Note that we have proved Lemma~\ref{lemma:UniversalGateSequence} in the frame rotating with the local longitudinal fields $U_{Z,t} = \exp\left( it PH_ZP \right)$. The transformation into the lab frame is trivial and therefore Lemma~\ref{lemma:UniversalGateSequence} also holds for sequences consisting of $U_\Gamma$ and $U_Z(\tau_i)$.

    \subsection{Error due to finite $|J|$}
    Up to now, we considered the effect of the blockade Hamiltonian $H_J$ only in the limit of $|J| \rightarrow \infty$, such that it can be considered to act as the projector $P$ onto the computational subspace. Considering a large but finite $J$, we introduce an additional source of error. We need to show that for finite $J$, the error can be controlled at a cost polynomial in the inverse accuracy. To this end, we employ the Schrieffer-Wolff transformation represented by the unitary $\exp(S_t)$ with a time-dependent anti-Hermitian operator $S_t$. This represents a change into a time-dependent basis. Since we are free to choose this basis, we will select one in which the Hamiltonian generates the dynamics that we are interested in. Then we show that the error due to this basis change is negligible.\\
    Recall the full Hamiltonian
    \begin{equation}
        \begin{aligned}
            H_{t} &= \Gamma_t H_X + H_Z \\
            &= \Gamma_t \sum_n X_n + \sum_n h_n Z_n + H_J \ ,
        \end{aligned}
    \end{equation}
    with $H_J = \frac{J}{4} \sum_{(m,n) \in E(\mathcal{G})} (1-Z_m)(1-Z_n)$ as in Eq.~\eqref{eq:DefHJ}. Since the single-qubit terms commute with the interaction terms if $\Gamma_t = 0$, phase gates do not introduce an error due to finite $J$. It therefore suffices to consider the error accumulated during a $\Gamma$-pulse.\\
    $H_t$ can be written as
    \begin{equation}
        \begin{aligned}
            H_t = \begin{pmatrix}
            PH_tP & PH_t Q \\
            Q H_t P & QH_tQ
        \end{pmatrix} \ ,
        \end{aligned}
    \end{equation}
    where $P$ and $Q$ are the projectors onto the computational subspace $\mathcal{P}$ and its complement $\mathcal{Q}$ respectively. Note that our previous analysis concerned itself with the $\mathcal{P}$-component of $H_t$, i.e. $PH_tP = H_{P,t}$. Recall that we denote the time-evolution operator generated by $H_{P,t}$ with
    \begin{equation}
        U_\Gamma = \mathcal{T} \exp \left(-i\int_0^\tau H_{P,t} dt \right) \ .
    \end{equation}
    Furthermore, let us denote the time-evolution operator under the full Hamiltonian $H_t$ with
    \begin{equation}
        U_{0,\Gamma} = \mathcal{T} \exp \left( -i\int_0^\tau H_t dt \right) \ .
    \end{equation}
    We need to bound the error on the computational subspace $\| PU_{0,\Gamma}P - PU_\Gamma P \|$.  This can be achieved using the approximate time-dependent Schrieffer-Wolff transformation. To this end, we can write $H_{t}$ as
    \begin{equation}
        \begin{aligned}
            H_{t} &= H_d + V_{d, t} + V_{od,t} \ ,
        \end{aligned}
    \end{equation}
    where we use the time-independent diagonal
    \begin{equation}
        H_d = \sum_n h_n Z_n + \frac{J}{4} \sum_{(m,n) \in E(\mathcal{G})} (1 - Z_m) (1 - Z_n) \ ,
    \end{equation}
    while the time-dependent terms of the Hamiltonian are split into the block-diagonal terms
    \begin{equation}
        V_{d,t} = \Gamma_t PH_XP + \Gamma_t QH_XQ
    \end{equation}
    and the block-off-diagonal part
    \begin{equation}
        \begin{aligned}
            V_{od,t} &= \Gamma_t PH_XQ + \Gamma_t QH_XP \\
            &= \Gamma_t \sum_n \left( PX_nQ + QX_nP \right) \ ,
        \end{aligned}
    \end{equation}
    which contains the non-zero matrix elements of $\Gamma_t\sum_n X_n$ connecting the computational subspace $\mathcal{P}$ to its complement $\mathcal{Q}$. Consider the transformed Hamiltonian
    \begin{equation}
        \begin{aligned}
            H_t'=& \exp(S_t) (H_d + V_{d,t} + V_{od,t}) \exp(-S_t) \\
            =& \sum_{k=0}^\infty \frac{1}{k!} [S_t, H_{t}]_k \\
            =& H_d + V_{d,t} + V_{od,t} + [S_t, H_d] + [S_t, V_{d,t}] + [S_t, V_{od,t}] + \frac{1}{2} [S_t,[S_t, H_{t}]] + \ ... \ ,
        \end{aligned}
    \end{equation}
    where we used the Campbell identity. $S_t$ is an anti-Hermitian operator that generates the basis change. By choosing $S_t$ such that
    \begin{equation}
        V_{od,t} + [S_t, H_d] = 0 \ ,
        \label{eq:SWStep1}
    \end{equation}
    $H_t'$ is block-diagonal to first order. Using the definition of $\ad_{H_d} (A) := [H_d, A]$ and considering $d_z$ the eigenvalues of $H_d$, we find for the matrix elements
    \begin{equation}
        \left( \ad_{H_d}(A) \right)_{z,z'} = \bra{z} [H_d, A]\ket{z'} = (d_z - d_{z'})A_{z,z'} \ ,
    \end{equation}
    where $A_{z,z'} = \bra{z}A\ket{z'}$. From this it follows that if
    \begin{equation}
        \forall z, z': d_z - d_{z'} = 0 \Rightarrow A_{z,z'} = 0 \ ,
        \label{eq:InverseExistenceCondition}
    \end{equation}
    the inverse $\ad_{H_d}^{-1}$ exists if we restrict it to operators $A$ that only connect $\mathcal{P}$ and $\mathcal{Q}$ fulfilling the condition Eq.~\eqref{eq:InverseExistenceCondition}. Applying this to the problem at hand, since $V_{od,t}$ only connects $\mathcal{P}$ and $\mathcal{Q}$ via single bitflips, the associated energy difference transitioning from $\ket{z} \in \mathcal{P}$ to $\ket{z'} \in \mathcal{Q}$ is given by
    \begin{equation}
        |d_z - d_{z'}| = |\Delta_{z,z'} \pm J| \ ,
    \end{equation}
    where $\Delta_{z,z'}$ is the energy associated with the flip of the respective qubit, $J$ is the energy penalty imposed by the blockade Hamiltonian $H_J$, and the sign is determined according to the direction of the bit flip. Since $\Delta_{z,z'}$ is determined by a single bit flip, we have $|\Delta_{z,z'}| = 2|h_{\mathcal{K}}|$ for the $\mathcal{K}$ that contains the flipped qubit. Then, choosing $J$ such that $|d_z - d_{z'}| \neq 0 \ \forall z,z'$ with a Hamming distance of one (i.e. $z$ and $z'$ that differ by one bitflip), e.g. such that $|J| \gg \max_{\mathcal{K}} 2|h_\mathcal{K}|$, the inverse $\ad_{H_d}^{-1}$ exists and we can fulfill Eq.~\eqref{eq:SWStep1} by choosing
    \begin{equation}
        \begin{aligned}
            S_t :=& \ad_{H_d}^{-1}(V_{od,t}) \\
            =& \Gamma_t \ad_{H_d}^{-1}\left(\sum_n P X_n Q + Q X_n P \right) \ .
        \end{aligned}
        \label{eq:DefS}
    \end{equation}
    We will specify the proper choice of $J$ in the following. Note that $S_t$ is also Lipschitz continuous, since by assumption $\Gamma_t$ is Lipschitz continuous. Then, as per Rademacher's theorem, $S_t$ is differentiable almost everywhere (a.e.). Since the basis change generated by $S_t$ is time-dependent, the effective Hamiltonian in the time-dependent basis $\exp(S_t)$ includes an adiabatic gauge potential
    \begin{equation}
        H_{eff,t} = H_t' - i \partial_t S_t \ \text{ a.e.}
    \end{equation}
    Consider the definition of the nested commutator of the operators $A$ and $B$ as
    \begin{equation}
        \begin{aligned}
            &[A, B]_{1} := [A, B] \ , \\
            &[A, B]_k := [A, [A, B]_{k-1}] \ \text{if } k>1 \ .
        \end{aligned}
    \end{equation}
    Using $V_{d,t} + V_{od, t} = \Gamma_t \sum_n X_n$, we can define the error term a.e.
    \begin{equation}
        \begin{aligned}
            V_{\varepsilon, t} &:= -i\partial_t S_t + [S_t, V_{d,t}] + [S_t, V_{od,t}] + \sum_{k=2}^\infty \frac{1}{k!} [S_t, H_{t}]_k \\
            &= -i\partial_t S_t + \left[ S_t, \Gamma_t\sum_n X_n \right] + \sum_{k=2}^\infty \frac{1}{k!} [S_t, H_{t}]_k \ ,
        \end{aligned}
        \label{eq:DefVeps}
    \end{equation}
    such that
    \begin{equation}
        H_{eff, t} = H_d + V_{d,t} + V_{\varepsilon, t} = PH_tP + QH_tQ + V_{\varepsilon, t} \ .
    \end{equation}
    Therefore, in the time-dependent frame generated by $S_t$, the effective Hamiltonian is approximately block-diagonal acting on $\mathcal{P}$ as $H_{P,t}$. The error of the approximation is given by $V_{\varepsilon, t}$. The next steps are taken to bound the norm of $V_{\varepsilon, t}$. We will first bound the norm of $S_t$ and $\partial_t S_t$ using Gershgorin's circle theorem~\cite{Gershgorin31} for large $|J|$ as
     \begin{equation}
         \| S_t \| , \| \partial _tS_t \| = \mathcal{O} \left( \frac{n_q' \Gamma^*}{|J|} \right) \ .
     \end{equation}
    Subsequently, we show that
    \begin{equation}
        \| V_{\varepsilon, t} \| = \mathcal{O} \left( \Gamma^* n_q' |J|^{-1} +  \Gamma^{*2} n_q'^2|J|^{-1} \right) \ ,
    \end{equation}
    also in the limit of large $|J|$.
    \begin{lemma}
        \label{lemma:SBounds}
        Let $S_t$ be defined as in Eq.~\eqref{eq:DefS}. Then
        \begin{equation}
            \| S_t \| = \mathcal{O}\left( n_q' \Gamma^* |J|^{-1}\right)
        \end{equation}
        and almost everywhere
        \begin{equation}
            \| \partial_t S_t \| = \mathcal{O}\left( n_q' \Gamma^* |J|^{-1}\right) \ .
        \end{equation}
    \end{lemma}
    \begin{proof}
        Since $S_t$ is anti-Hermitian, and therefore also normal, by definition, its norm $\| S_t \|$ is its spectral radius $\max_i |\lambda_i|$, where $\lambda_i$ are the eigenvalues of $S_t$. The spectral radius in turn can be bounded by the Gershgorin circle theorem (Theorem~\ref{theorem:Gershgorin}). Since the diagonal elements of $S_t$ are all zero and the off-diagonal elements only connect eigenstates of $H_d$ in $\mathcal{P}$ and $\mathcal{Q}$ respectively, Gershgorin's circle theorem renders
        \begin{equation}
            \begin{aligned}
                \| S_t \| &= \max_i |\lambda_i|\\
                &\leq \max_z \sum_{z'} |\bra{z} S_t \ket{z'}| \\
                &=  \max_z \sum_{z'} \left| \frac{\bra{z} V_{od,t} \ket{z'}}{\Delta_{z,z'} \pm J} \right| \\
                &\leq \max_z \sum_{z'} \left| \frac{\bra{z} V_{od,t} \ket{z'}}{\min_{z,z'}|\Delta_{z,z'} \pm J|  } \right| \\
                &\leq n_q' \left| \frac{\Gamma_t}{\min_{z,z'}|\Delta_{z,z'} \pm J|  } \right| \\
                &\leq n_q' \frac{|\Gamma_t|}{|J| - \max_{z,z'} |\Delta_{z,z'}|} \ \text{ (for }|J| > \max_{z,z'}|\Delta_{z,z'}| \text{)} \ .
            \end{aligned}
            \label{eq:SBound1}
        \end{equation}
        Choosing a constant $C_S > 1$, we can further bound for $|J| > \frac{C_S}{C_S - 1} \max_{z,z'}|\Delta_{z,z'}|$ and for all $t \in [0, \tau]$
        \begin{equation}
            \|S_t\| \leq C_S n_q' \Gamma^* |J|^{-1} \ .
            \label{eq:SBound2}
        \end{equation}
        Note that since $C_S > 1$, then $C_S/(C_S - 1) > 1$ and the condition on $|J|$ in Eq.~\eqref{eq:SBound1} is also contained in the condition in Eq.~\eqref{eq:SBound2}. Furthermore, note that by construction we have $\max_{z,z'} |\Delta_{z,z'}| = \max_\mathcal{K}2|h_{\mathcal{K}}| = \mathcal{O}(1)$, so the choice of $C_S$ is also independent of $n_q'$.\\
        Similarly, since by the definitions of $S_t$ in Eq.~\eqref{eq:DefS} and $\Gamma_t$ in Section~\ref{sec:GammaPulseDefinition}, $\partial_t \Gamma_t$ exists a.e. with $|\partial_t \Gamma_t| = \mathcal{O}(\Gamma^*)$ and therefore by application of the above argument we find
        \begin{equation}
            \| \partial_t S_t \| = C_{\partial S} n_q' \Gamma^{*} |J|^{-1} \ \text{a.e.},
        \end{equation}
        again for some constant $C_{\partial S} > 0$. This concludes the proof.
    \end{proof}
    We now show that the norm of the error $V_{\varepsilon, t}$ vanishes for large $|J|$ and for small $n_q' \Gamma^*$ respectively. To this end, we will prove the following lemma.
    \begin{lemma}
        \label{lemma:HamiltonianErrorBound}
        Let $V_{\varepsilon, t}$, $S_t$ and $H_t$ be as above. Then
        \begin{equation}
            \| V_{\varepsilon, t} \| = \mathcal{O} \left( \Gamma^* n_q' |J|^{-1} +  \Gamma^{*2} n_q'^2|J|^{-1} \right) \ \text{a.e.}
        \end{equation}
    \end{lemma}
    \begin{proof}
        We will bound the terms in the definition of $V_{\varepsilon,t}$ in Eq.~\eqref{eq:DefVeps} and then combine the term-wise bounds using the triangle inequality. The bound on the first term follows directly from Lemma~\ref{lemma:SBounds}, which results in
        \begin{equation}
            \label{eq:FirstTermBound}
            \begin{aligned}
                \| \partial_t S_t \| &= C_{\partial S} n_q' \Gamma^* |J|^{-1} \ \text{a.e.}
            \end{aligned}
        \end{equation}
        for sufficiently large $|J|$. Similarly, the second term can be bounded as
        \begin{equation}
            \label{eq:SecondTermBound}
            \left\| \left[ S_t, \Gamma_t \sum_n X_n \right] \right\| \leq 2\|S_t\| \Gamma^* n_q' \leq 2 C_S \Gamma^{*2} n_q'^2 |J|^{-1} \ ,
        \end{equation}
        again using Lemma~\ref{lemma:SBounds} for the last inequality. Let us turn to the infinite sum. Consider that, using the definition of $S_t$ from Eq.~\eqref{eq:DefS} and $V_{d,t} + V_{od,t} = \Gamma_t \sum_n X_n$, we find
        \begin{equation}
            \begin{aligned}
                [S_t, H_t] &= [S_t, H_d] + [S_t, V_{d,t}] + [S_t, V_{od,t}] \\
                &= -V_{od,t} + \left[ S_t, \Gamma_t \sum_n X_n \right]
            \end{aligned}
        \end{equation}
        and therefore by using the linearity of the nested commutator
        \begin{equation}
            [S_t, H_t]_k = [S_t, -V_{od,t}]_{k-1} + \left[ S_t, \Gamma_t \sum_n X_n \right]_k \ .
        \end{equation}
        Using the triangle inequality and submultiplicativity of operator norms, we can bound
        \begin{equation}
            \begin{aligned}
                \left\| [S_t, H_t]_k \right\| &= \left\| [S_t, -V_{od,t}]_{k-1} + \left[ S_t, \Gamma_t \sum_n X_n \right]_k \right\| \\
                &\leq \left\| [S_t, -V_{od,t}]_{k-1} \right\| + \left\| \left[ S_t, \Gamma_t \sum_n X_n \right]_k \right\| \\
                &\leq 2^{k-1} \|S_t\|^{k-1} \|V_{od,t}\| + 2^k \| S_t \|^k \left\| \Gamma_t \sum_n X_n \right\| \ ,
            \end{aligned}
            \label{eq:NestedCommBound}
        \end{equation}
        where we have used the standard bound $\| [A, B ]_k \| \leq 2^k \|A\|^k \|B\|$. Using Lemma~\ref{lemma:SBounds}, $\| V_{od,t} \| \leq n_q' \Gamma^*$ and $\| \Gamma_t \sum_n X_n \| \leq n_q' \Gamma^*$ in Eq.~\eqref{eq:NestedCommBound}, we find
        \begin{equation}
            \begin{aligned}
                \| [S_t, H_t]_k \| &\leq \left( 2 C_S n_q' \Gamma^{*} |J|^{-1} \right)^{k-1} n_q' \Gamma^{*} + \left( 2 C_S n_q' \Gamma^{*} |J|^{-1} \right)^k  n_q' \Gamma^{*} \ .
            \end{aligned}
        \end{equation}
        It follows that the series in the definition of $V_{\varepsilon, t}$ in Eq.~\eqref{eq:DefVeps} converges absolutely. Then, we can bound the infinite sum using Lemma~\ref{lemma:ExponentialSeriesBound} according to
        \begin{equation}
            \label{eq:ThirdTermBound}
            \begin{aligned}
                \left\| \sum_{k=2}^\infty \frac{1}{k!} [S_t, H_{t}]_k \right\| &\leq \sum_{k=2}^\infty \frac{1}{k!} \left\| [S_t, H_{t}]_k \right\| \\
                &\leq n_q' \Gamma^{*} \sum_{k=2}^\infty \frac{1}{k!} \left( \left( 2 C_S n_q' \Gamma^{*} |J|^{-1} \right)^{k-1} + \left( 2 C_S n_q' \Gamma^{*} |J|^{-1} \right)^k \right) \\
                &= n_q' \Gamma^{*} \sum_{k=2}^\infty \frac{1}{k!} \left( 2 C_S n_q' \Gamma^{*} |J|^{-1} \right)^{k-1} + n_q' \Gamma^{*}\sum_{k=2}^\infty \frac{1}{k!} \left( 2 C_S n_q' \Gamma^{*} |J|^{-1} \right)^k \\
                &= n_q' \Gamma^{*} \left( \left( 2 C_S n_q' \Gamma^{*} |J|^{-1} \right)^{-1} + 1 \right) \sum_{k=2}^\infty \frac{1}{k!} \left( 2 C_S n_q' \Gamma^{*} |J|^{-1} \right)^k \\
                &\leq n_q' \Gamma^{*} \left( \left( 2 C_S n_q' \Gamma^{*} |J|^{-1} \right)^{-1} + 1 \right) \left( 2 C_S n_q' \Gamma^{*} |J|^{-1} \right)^2 e^{ 2 C_S n_q' \Gamma^{*} |J|^{-1}} \ .
            \end{aligned}
        \end{equation}
        Then, using Eq.~\eqref{eq:FirstTermBound}, Eq.~\eqref{eq:SecondTermBound} and Eq.~\eqref{eq:ThirdTermBound}, we can bound the error term almost everywhere
        \begin{equation}
            \begin{aligned}
                \| V_{\varepsilon, t} \| \leq& \|\partial_t S_t \| + 2 n_q' \Gamma^* \| S_t\| + \sum_{k=2}^\infty \frac{1}{k!} \| [S_t, H_{t}]_k \| \\ 
                \leq& C_{\partial S} n_q' \Gamma^* |J|^{-1} + 2 C_S n_q'^2 \Gamma^{*2} |J|^{-1} + 2 C_S n_q'^2 \Gamma^{*2} |J|^{-1} e^{ 2 C_S n_q' \Gamma^{*} |J|^{-1}} \\
                &+ n_q' \Gamma^{*} \left( 2 C_S n_q' \Gamma^{*} |J|^{-1} \right)^2 e^{ 2 C_S n_q' \Gamma^{*} |J|^{-1}}
            \end{aligned}
            \label{eq:VBound1}
        \end{equation}
        Hence, for small $n_q' \Gamma^{*} |J|^{-1}$, Eq.~\eqref{eq:VBound1} gives
        \begin{equation}
            \| V_{\varepsilon, t} \| = \mathcal{O} \left( \Gamma^* n_q' |J|^{-1} +  \Gamma^{*2} n_q'^2|J|^{-1} \right) \ \text{a.e.},
        \end{equation}
        which concludes the proof.
    \end{proof}
    Lemma~\ref{lemma:HamiltonianErrorBound} establishes the bound
    \begin{equation}
        \| H_d + V_{d,t} - H_{eff,t} \| = \mathcal{O} \left( \Gamma^* n_q' |J|^{-1} + \Gamma^{*2} n_q'^2|J|^{-1} \right) \ \text{a.e.},
    \end{equation}
    which we can now use to bound the error on the time evolution. Considering the time-evolution under $H_d + V_{d,t} + V_{\varepsilon, t}$, we can treat $V_{\varepsilon, t}$ as a perturbation for small $\| V_{\varepsilon,t}\|$, which is true for sufficiently large $|J|$, and by applying Lemma~\ref{lemma:TimeEvolutionError}, we find that
    \begin{equation}
        \left\| \mathcal{T} \exp \left(-i\int_0^{\tau} H_{eff,t} dt \right) - \mathcal{T} \exp \left( -i \int_0^\tau (H_d + V_{d,t}) dt \right) \right\| = \mathcal{O}\left( \tau\Gamma^* n_q' |J|^{-1} + \tau \Gamma^{*2} n_q'^2 |J|^{-1} \right) \ .
        \label{eq:TimeEvoBound1}
    \end{equation}
    The unitary generated by $H_{eff,t}$ is in the dressed basis generated by $S_t$. However, by construction of $\Gamma_t$ we have $\Gamma_0 = \Gamma_\tau = 0$ and therefore also $S_0 = S_\tau = 0$. Thus, in the computational basis, the unitary reads
    \begin{equation}
        \begin{aligned}
            U_{0,\Gamma} &= e^{-S_\tau} U_{eff,\Gamma} e^{S_0} = U_{eff,\Gamma} \ ,
        \end{aligned}
        \label{eq:TimeEvoBound2}
    \end{equation}
    where we have defined
    \begin{equation}
        U_{eff,\Gamma} := \mathcal{T} \exp \left(-i\int_0^{\tau} H_{eff,t} dt \right) \ .
    \end{equation}
    This is where the assumption $\Gamma_0 = \Gamma_\tau = 0$ is convenient, since without it, $S_0$ and $S_\tau$ would not vanish and we would have to track an additional error $\mathcal{O}(n_q' \Gamma^* |J|^{-1})$, which could be absorbed in the error bounds derived here. Then, using Eq.~\eqref{eq:TimeEvoBound1}, Eq.~\eqref{eq:TimeEvoBound2} and that $\tau = \mathcal{O}(1)$ as defined in Section~\ref{sec:GammaPulseDefinition}, we find
    \begin{equation}
        \left\| U_{0,\Gamma} - U_{\Gamma, PQ} \right\| = \mathcal{O} \left(  \Gamma^* n_q' |J|^{-1} +  \Gamma^{*2}n_q'^2 |J|^{-1}  \right) \ ,
        \label{eq:EffectiveGammaPulse}
    \end{equation}
    where, for ease of notation, we have defined
    \begin{equation}
        U_{\Gamma, PQ} := \mathcal{T} \exp \left( -i\int_0^\tau (H_d + V_{d,t}) dt \right) \ .
    \end{equation}
    $U_{\Gamma, PQ}$ is the ideal $\Gamma$-pulse $U_\Gamma$, except that it also acts non-trivially on the complement of the computational subspace $\mathcal{Q}$. Since $H_d + V_{d,t}$ commutes with $P$, it follows directly that
    \begin{equation}
        [P, U_{\Gamma, PQ}] = 0
    \end{equation}
    and
    \begin{equation}
        PU_{\Gamma,PQ} = U_{\Gamma,PQ}P = PU_\Gamma = U_\Gamma P \ .
    \end{equation}
    Considering a concatenation of $N = \mathcal{O}\left( \Gamma^{*-1} \right)$ $\Gamma$-pulses, we find
    \begin{equation}
        \begin{aligned}
            &U_Z(\tau_{0})\prod_{i=1}^{N/2} U_{0,\Gamma} U_Z(\tau_{2i-1}) U_{0,\Gamma} U_Z(\tau_{2i}) \\
            =& U_Z(\tau_0)\prod_{i=1}^{N/2} \left[ U_{\Gamma, PQ} U_Z(\tau_{2i-1}) U_{\Gamma, PQ} U_Z(\tau_{2i})  + \mathcal{O} \left( \Gamma^* n_q' |J|^{-1} + \Gamma^{*2} n_q'^2 |J|^{-1} \right) \right]\\
            =& U_Z(\tau_0) \prod_{i=1}^{N/2} \left[ U_{\Gamma, PQ} U_Z(\tau_{2i-1}) U_{\Gamma, PQ} U_Z(\tau_{2i}) \right] + \mathcal{O} \left( N\Gamma^* n_q' |J|^{-1} + N\Gamma^{*2} n_q'^2 |J|^{-1} \right) \ ,
        \end{aligned}
        \label{eq:EffectiveGammaPulseSequence}
    \end{equation}
    where the first equality follows using Eq.~\eqref{eq:EffectiveGammaPulse} and the second equality expands the product and considers the limit of small error. Furthermore, consider that with Eq.~\eqref{eq:EffectiveGammaPulse} we find that
    \begin{equation}
        \| U_{0,\Gamma} P - U_{\Gamma, PQ} P \| \leq \| U_{0,\Gamma} - U_{\Gamma, PQ} \| = \mathcal{O} \left(  \Gamma^* n_q' |J|^{-1} +  \Gamma^{*2}n_q'^2 |J|^{-1}  \right) \ .
    \end{equation}
    Then, using that $[U_\Gamma, P] = [U_{\Gamma,PQ}, P] = [U_Z, P] = 0$ and the fact that $N = \mathcal{O}(\Gamma^{*-1})$ we find for the time evolution of an initial state in $\mathcal{P}$ that
    \begin{equation}
        \begin{aligned}
            &\left( U_Z(\tau_{0})\prod_{i=1}^{N/2} U_{0,\Gamma} U_Z(\tau_{2i-1}) U_{0,\Gamma} U_Z(\tau_{2i}) \right) P \\
            = & \left( U_Z(\tau_{0})\prod_{i=1}^{N/2} U_{\Gamma, PQ} U_Z(\tau_{2i-1}) U_{\Gamma, PQ} U_Z(\tau_{2i}) \right) P + \mathcal{O}\left( n_q' |J|^{-1} + \Gamma^{*} n_q'^2 |J|^{-1}\right) \\
            = &\left( U_Z(\tau_{0})\prod_{i=1}^{N/2} U_{\Gamma} U_Z(\tau_{2i-1}) U_{\Gamma} U_Z(\tau_{2i}) \right) P + \mathcal{O}\left( n_q' |J|^{-1} + \Gamma^{*} n_q'^2 |J|^{-1}\right)\\
            = & \left( \prod_{\mathcal{K} \subseteq \mathcal{I}} \prod_{n\in\mathcal{K}} V_{\mathcal{K}, n}^{(c)} \right) P + \mathcal{O}\left( n_q' \varepsilon_{\text{phase}} / \Gamma^{*} + n_q'^2 \Gamma^* + n_q' |J|^{-1} + \Gamma^{*} n_q'^2 |J|^{-1}\right) \ ,
        \end{aligned}
    \end{equation}
    where the first equality follows from Eq.~\eqref{eq:EffectiveGammaPulseSequence}, the second equality from the observation that $U_{\Gamma, PQ} P = P U_{\Gamma, PQ} P = P U_{\Gamma} P = U_{\Gamma} P$, and the last equality follows from Lemma~\ref{lemma:UniversalGateSequence} for any $V_\mathcal{K} \in \text{SU}(2)$, $\mathcal{K} \in \{ \mathcal{A}, \mathcal{B}, \mathcal{C}, \mathcal{D} \} \subseteq \mathcal{I}$. In the following, we will show that we require a polynomially large $|J|\gg 1$, so $|J|^{-1} \ll 1$ and the $\Gamma^* n_q'^2 |J|^{-1}$ term in the error can be absorbed into the $\Gamma^* n_q'^2$ term. Thus, we have proved the following lemma.
    \begin{lemma}
        \label{lemma:FullConditionalGateWithError}
        For any set of conditional gates $V_{\mathcal{K}, n}^{(c)}$ acting on non-adjacent qubits $\mathcal{I} \in \{ \mathcal{A}, \mathcal{B}, \mathcal{C}, \mathcal{D}, \mathcal{B} \cup \mathcal{D} \}$ and any $\nu>0$ and $\varepsilon_{\text{phase}} >0$, there is a sequence of $N$ $\Gamma$-pulses and $N+1$ waiting times $\tau_i = \mathcal{O} \left( \varepsilon_{\text{phase}}^{-3-\nu} \right)$ such that
        \begin{equation}
            \begin{aligned}
                &\left\| \left( \prod_{\mathcal{K} \subseteq \mathcal{I}} \prod_{n\in\mathcal{K}} V_{\mathcal{K}, n}^{(c)} - U_Z(\tau_0) \prod_{i=1}^{N/2} U_{0,\Gamma} U_Z(\tau_{2i-1}) U_{0, \Gamma} U_Z(\tau_{2i}) \right) P \right\| \\
                &\leq \delta_{\text{gate}} = \mathcal{O}\left( n_q' \varepsilon_{\text{phase}} / \Gamma^{*} + n_q'^2 \Gamma^* + n_q' |J|^{-1} \right) \ ,
            \end{aligned}
            \label{eq:FullConditionalGateWithError}
        \end{equation}
        where $N = \mathcal{O}\left( \Gamma^{*-1} \right)$.
    \end{lemma}
    Lemma~\ref{lemma:FullConditionalGateWithError} tells us that for a state initialized in the computational subspace $\mathcal{P}$, the evolution under $\Gamma$-pulses and waiting times can implement controlled gates on non-adjacent subsets of qubits. There are three contributions to the error scaling in Lemma~\ref{lemma:FullConditionalGateWithError}, but we also have three parameters to control it. We will conclude with an analysis of the required resources to reduce the error below any desired level.

    \subsection{Total error scaling and resources}
    Using the CP-method, a quantum circuit on $n_q$ qubits and $p$ gates is mapped to a Hamiltonian of $n_q' = \mathcal{O}(n_q^2)$ qubits and $p' = \mathcal{O}(n_q p)$ conditional quantum gates $\{ U_{\mathcal{I}_i, i}\}_{i=1}^{p'}$. Since the CP-method simulates any quantum circuit exactly, we find that if we initialize the $n_q'$ qubits in $\ket{0}^{\otimes n_q'}$ and  apply the gates, we obtain
    \begin{equation}
        U_{\mathcal{I}_{p'}, p'} U_{\mathcal{I}_{p'-1}, p'-1} ... U_{\mathcal{I}_{1}, 1} \ket{0}^{\otimes n_q'} = U_{CP} \ket{0}^{\otimes n_q'} \ ,
    \end{equation}
    where we defined $U_{CP}:= U_{\mathcal{I}_{p'}, p'} U_{\mathcal{I}_{p'-1}, p'-1} ... U_{\mathcal{I}_{1}, 1}$ as the CP-mapped realization of the target circuit $U_{\text{target}}$. Tracing over all qubits that do not carry the logical quantum information, which we denote with $\Tr_{RO}$, we find that
    \begin{equation}
        \Tr_{RO} \left( U_{CP} \ket{0}^{\otimes n_q'} \bra{0}^{\otimes n_q'} U_{CP}^\dagger \right) = U_{\text{target}} \ket{0}^{\otimes n_q} \bra{0}^{\otimes n_q} U_{\text{target}}^{\dagger} \ .
        \label{eq:CPvsTarget}
    \end{equation}
    As we have shown with Lemma~\ref{lemma:FullConditionalGateWithError}, each of the $p'$ gates $U_{\mathcal{I}_{i}, i}$ can be expressed by concatenated $\Gamma$-pulses with waiting times $\tau_i$ between them up to an error bounded by $\delta_{\text{gate}}$ as in Eq.~\eqref{eq:FullConditionalGateWithError}. We can then similarly concatenate the pulse sequences from all the $p'$ gates to apply all the conditional gates required to run the CP-method. In the limit of small $\varepsilon_{\text{phase}}, \Gamma^*$, and $|J|^{-1}$, the errors are additive and the total error of $p'$ gates compiled into sequences as in Eq.~\eqref{eq:FullConditionalGateWithError} reads
    \begin{equation}
        \left\| \left(U_Z(\tau_0) \prod_{i=1}^{p'N/2} U_{0,\Gamma} U_Z(\tau_{2i-1}) U_{0, \Gamma} U_Z(\tau_{2i}) - U_{CP} \right) P \right\| \leq \delta_{\text{total}} \ ,
        \label{eq:CPCircuitError}
    \end{equation}
    where
    \begin{equation}
        \begin{aligned}
            \delta_{\text{total}} &\leq p'\delta_{\text{gate}} = \mathcal{O} \left( p'n_q' \varepsilon_{\text{phase}} / \Gamma^{*} + p'n_q'^2 \Gamma^{*} + p' n_q' |J|^{-1} \right) \\
            &= \mathcal{O} \left( p n_q^3 \varepsilon_{\text{phase}} / \Gamma^{*} + p n_q^5 \Gamma^{*}  + p n_q^3 |J|^{-1}\right) \ .
        \end{aligned}
        \label{eq:TotalError}
    \end{equation}
    The full unitary generated by the $\Gamma$-pulses is given by
    \begin{equation}
        U_T := U_Z(\tau_0) \prod_{i=1}^{p'N/2} U_{0,\Gamma} U_Z(\tau_{2i-1}) U_{0, \Gamma} U_Z(\tau_{2i}) \ .
    \end{equation}
    Note that since $U_T$ in this definition is a concatenation of waiting times ($\Gamma_t = 0$) and $\Gamma$-pulses, we can concatenate these pieces to create the schedule $\Gamma_t: [0, T] \rightarrow [0, \Gamma^*]$ such that $U_T = \mathcal{T} \exp (-i\int_0^T H_t dt )$, with $H_t$ as in Eq.~\eqref{eq:DefFullHamiltonian}.\\
    Since $P \ket{0}^{\otimes n_q'} = \ket{0}^{\otimes n_q'}$, then $U_T P \ket{0}^{\otimes n_q'} = U_T \ket{0}^{\otimes n_q'}$ and $U_{CP} P \ket{0}^{\otimes n_q'} = U_{CP} \ket{0}^{\otimes n_q'}$ and thus with Eq.~\eqref{eq:CPCircuitError}
    \begin{equation}
        \left\| U_T \ket{0}^{\otimes n_q'} \bra{0}^{\otimes n_q'} U_T^\dagger - U_{CP} \ket{0}^{\otimes n_q'} \bra{0}^{\otimes n_q'} U_{CP}^\dagger \right\| = \mathcal{O}(\delta_{\text{total}}) \ ,
    \end{equation}
    where we used the standard inequality
    \begin{equation}
        \| U P \ket{\psi} \bra{\psi} P U^{\dagger} - V P \ket{\psi}\bra{\psi}P V^\dagger \| \leq 2 \| UP - VP \| \leq 2 \| U - V \| \ .
    \end{equation}
    Then, with Eq.~\eqref{eq:CPvsTarget} and Eq.~\eqref{eq:TotalError}, we have
    \begin{equation}
        \begin{aligned}
            &\left\| \Tr_{RO} \left( U_T \ket{0}^{\otimes n_q'} \bra{0}^{\otimes n_q'} U_T^\dagger \right) - U_{\text{target}} \ket{0}^{\otimes n_q} \bra{0}^{\otimes n_q} U_{\text{target}}^{\dagger} \right\|\\
            =& \left\| \Tr_{RO} \left( U_T \ket{0}^{\otimes n_q'} \bra{0}^{\otimes n_q'} U_T^\dagger \right) -\Tr_{RO} \left( U_{CP} \ket{0}^{\otimes n_q'} \bra{0}^{\otimes n_q'} U_{CP}^\dagger \right) \right\| \\
            =& \delta_{\text{total}}' = \mathcal{O}(\delta_{\text{total}}) = \mathcal{O} \left( p n_q^3 \varepsilon_{\text{phase}} / \Gamma^{*} + p n_q^5 \Gamma^{*}  + p n_q^3 |J|^{-1}\right) \ .
        \end{aligned}
    \end{equation}
    Therefore, in order to simulate a given circuit with a maximal error $\delta_{\text{total}}' \leq \varepsilon$, we require the parameters
    \begin{equation}
        \begin{aligned}
            |J| &= \mathcal{O}\left(p n_q^3 \varepsilon^{-1} \right) \ , \\
            \Gamma^* &= \mathcal{O} \left( p^{-1} n_q^{-5} \varepsilon \right) \ , \\
            \varepsilon_{\text{phase}} &= \mathcal{O} \left( p^{-1} n_q^{-3}  \Gamma^* \varepsilon\right) \\
            &= \mathcal{O} \left( p^{-2} n_q^{-8}  \varepsilon^2 \right) \ .
        \end{aligned}
    \end{equation}
    While the scaling of $|J|$ gives the required energy scale, and the scaling of $\Gamma^*$ gives a precision of the time-dependent controls, the main contribution to the time overhead is given by Theorem~\ref{theorem:FillingTime} since $\tau_i = \mathcal{O}(\varepsilon_{\text{phase}}^{-3-\nu})$. Since there are $p'$ gates from the CP-method, each of which is compiled into $\mathcal{O}(\Gamma^{*-1})$ $\Gamma$-pulses and phase gates, the total computation time $T$ is
    \begin{equation}
        \begin{aligned}
            T &= \mathcal{O} \left( p' \Gamma^{*-1} \varepsilon_{\text{phase}}^{-3-\nu} \right) \\
            &= \mathcal{O} \left( p^{8+2\nu} n_q^{30 + 8\nu} \varepsilon^{-7-2\nu}\right)
        \end{aligned}
    \end{equation}
    for any $\nu > 0$. This concludes the proof of Theorem~\ref{theorem:MainResultFormal}.
\end{proof}